\documentclass[a4paper,11pt,fleqn]{article}
\usepackage[official]{eurosym}
\usepackage{amsmath}
\usepackage{amssymb}
\usepackage{theorem}
\usepackage[usenames,dvipsnames]{pstricks}
\usepackage{euscript}
\usepackage{graphicx}
\usepackage{charter}

\usepackage[T1]{fontenc}
\usepackage{array}
\usepackage{fixltx2e}
\usepackage{url}
\usepackage{amsfonts}
\usepackage{lmodern}
\usepackage{algpseudocode}
\usepackage{algorithmicx}
\usepackage{algorithm}
\usepackage{xfrac}
\usepackage{caption}
\usepackage{subcaption}
\usepackage{float}
\usepackage{multirow}
\usepackage{balance}
\usepackage{xargs}

\usepackage{pgf,tikz}
\usetikzlibrary{arrows}
\usetikzlibrary{arrows, positioning, decorations.markings}
\usetikzlibrary{shapes,snakes}
\tikzstyle{every picture}+=[remember picture]

\newcommand{\email}{}
\definecolor{labelkey}{rgb}{0,0.08,0.45}
\definecolor{refkey}{rgb}{0,0.6,0.0}
\definecolor{Brown}{rgb}{0.45,0.0,0.05}
\definecolor{dgreen}{rgb}{0.00,0.49,0.00}
\definecolor{dblue}{rgb}{0,0.08,0.75}
\RequirePackage[colorlinks,hyperindex]{hyperref} 
\hypersetup{linktocpage=true,citecolor=dblue,linkcolor=dgreen}
\PassOptionsToPackage{normalem}{ulem}
\usepackage{ulem}

\DeclareMathOperator*{\Argmin}{Argmin}

\DeclareMathOperator*{\Argmax}{Argmax}

\newcommand{\argmind}[2]{\ensuremath{\underset{\substack{{#1}}}%
{\mathrm{argmin}}\;\;#2 }}
\newcommand{\Argmind}[2]{\ensuremath{\underset{\substack{{#1}}}%
{\mathrm{Argmin}}\;\;#2 }}

\renewcommand{\leq}{\ensuremath{\leqslant}}
\renewcommand{\geq}{\ensuremath{\geqslant}}
\renewcommand{\le}{\ensuremath{\leqslant}}
\renewcommand{\ge}{\ensuremath{\geqslant}}

\newcommand{\emp}{\ensuremath{{\varnothing}}}

\newcommand{\proj}{\ensuremath{\Pi}}

\newcommand{\dist}{\ensuremath{\operatorname{dist}}}

\newcommand{\Bc}{\mathcal{B}}
\newcommand{\Cc}{\mathcal{C}}

\newcommand{\Lc}{\mathcal{L}}
\newcommand{\Mc}{\mathcal{M}}
\newcommand{\Sc}{\mathcal{S}}

\newcommand{\C}{\ensuremath{\mathbb{C}}}

\newcommand{\N}{\ensuremath{\mathbb{N}}}

\newcommand{\R}{\ensuremath{\mathbb{R}}}

\newcounter{hypoconbis}
\newcounter{saveconbis}
\newcommand\debutH{\begin{list}
{\textbf{H\arabic{hypoconbis}}}{\usecounter{hypoconbis}}\setcounter{hypoconbis}{\value{saveconbis}}}
\newcommand\finH{\end{list}\setcounter{saveconbis}{\value{hypoconbis}}}

\algblock[Block]{Block}{EndBlock}
\algblockdefx[Block]{Block}{EndBlock}%
[1]{{#1}}%
[1]{{#1}}

\newtheorem{theorem}{Theorem}[section]

\theoremstyle{plain}{\theorembodyfont{\rmfamily}%
}
\theoremstyle{plain}{\theorembodyfont{\rmfamily}%
\newtheorem{assumption}[theorem]{Assumption}}
\theoremstyle{plain}{\theorembodyfont{\rmfamily}%
}
\theoremstyle{plain}{\theorembodyfont{\rmfamily}%
}
\theoremstyle{plain}{\theorembodyfont{\rmfamily}%
\newtheorem{remark}[theorem]{Remark}}
\theoremstyle{plain}{\theorembodyfont{\rmfamily}%
\newtheorem{definition}[theorem]{Definition}}
\theoremstyle{plain}{\theorembodyfont{\rmfamily}%
}

\numberwithin{equation}{section}
\begin{document}

\title{Scalable Bayesian Uncertainty Quantification in Imaging Inverse Problems via Convex Optimization\footnote{This work was supported by the UK Engineering and Physical Sciences Research Council (EP/M008843/1 and EP/M019306/1).}}
\author{A. Repetti, 
M. Pereyra, 
and Y. Wiaux\\[5mm]
\small
\small Heriot-Watt University\\
\small Edinburgh EH14 4AS, United Kingdom\\
\small \email{\{a.repetti, m.pereyra, y.wiaux\}@hw.ac.uk}}
\date{}

\maketitle

\vskip 8mm

\begin{abstract}
We propose a Bayesian uncertainty quantification method for large-scale imaging inverse problems. Our method applies to all Bayesian models that are log-concave, where maximum-a-posteriori (MAP) estimation is a convex optimization problem. The method is a framework to analyse the confidence in specific structures observed in MAP estimates (e.g., lesions in medical imaging, celestial sources in astronomical imaging), to enable using them as evidence to inform decisions and conclusions. Precisely, following Bayesian decision theory, we seek to assert the structures under scrutiny by performing a Bayesian hypothesis test that proceeds as follows: firstly, it postulates that the structures are not present in the true image, and then seeks to use the data and prior knowledge to reject this null hypothesis with high probability. Computing such tests for imaging problems is generally very difficult because of the high dimensionality involved. A main feature of this work is to leverage probability concentration phenomena and the underlying convex geometry to formulate the Bayesian hypothesis test as a convex problem, that we then efficiently solve by using scalable optimization algorithms. This allows scaling to high-resolution and high-sensitivity imaging problems that are computationally unaffordable for other Bayesian computation approaches. We illustrate our methodology, dubbed BUQO (Bayesian Uncertainty Quantification by Optimization), on a range of challenging Fourier imaging problems arising in astronomy and medicine.
\textsc{Matlab} code for the proposed uncertainty quantification method is available on GitHub.
\end{abstract}



\section{Introduction} 
\label{Sec:Intro}

In this paper, we consider the problem of estimating an unknown image ${x} \in \R^N$ from an observation $y \in \mathbb{C}^M$, related to $x$ by a statistical model $p(y|x)$. We focus on linear problems of the form
\begin{equation}\label{pb:inv_pb}
y = \Phi {x} + w,
\end{equation}
where $\Phi \colon \R^N \to \C^M$ is a known observation operator and $w \in \mathbb{C}^M$ is a realization of random noise with bounded energy (i.e., we assume that $\|w\|^2 \le \epsilon$ with $\epsilon > 0$ known, and $\| \cdot \|$ being the usual Euclidean norm). 
Assuming the exact noise model is unknown, we simply postulate a uniform likelihood $p(y|x) \propto \boldsymbol{1}_{\mathcal{B}(y,\epsilon)} (\Phi x)$, where $\Bc_2(y,\varepsilon)$ denotes the $\ell_2$ ball centred in $y$ with radius $\epsilon$, and where, for every $ s \in \mathbb{C}^M$, the function $\boldsymbol{1}_{\mathcal{B}(y,\epsilon)}(s) = 1$ if $s \in \mathcal{B}(y,\epsilon)$, and $\boldsymbol{1}_{\mathcal{B}(y,\epsilon)}(s) = 0$ otherwise.\footnote{The likelihood $p(x|y) \propto \boldsymbol{1}_{\mathcal{B}(y,\epsilon)} (\Phi x)$ can also be used as an approximation in cases where $\|w\|^2 \le \epsilon$ holds with high probability. See Section \ref{sec:results} for more details.}
This likelihood is commonly used in computational imaging, for example in astronomical imaging \cite{onose2016scalable, wiaux2009compressed} and medical imaging \cite{Haldar2011}. 
Section~\ref{Ssec:discuss:noise} explains how to generalise the methodology proposed in this paper to other noise models. This generalisation is straightforward; however, for presentation clarity and conciseness here we use the model \eqref{pb:inv_pb}.

In imaging sciences, the problem of estimating $x$ from $y$ is often ill-posed or ill-conditioned, resulting in significant uncertainty about the true value of $x$ \cite{robert2007bayesian} (this arises for example in compressive sensing problems where the dimensions $M \ll N$). Bayesian imaging methods address this difficulty by using prior knowledge about $x$ to regularise the estimation problem and reduce the uncertainty about $x$ \cite{robert2007bayesian}. Formally, they model $x$ as a random vector with prior distribution $p(x)$ promoting expected properties (e.g., sparsity or smoothness), and combine observed and prior information by using Bayes' theorem to produce the posterior distribution \cite{robert2007bayesian}
\begin{equation}\label{eq:posterior1}
p(x|y) = \frac{p(y|x)p(x)}{\int_{\R^N} p(y|x)p(x)\textrm{d}x}\,,
\end{equation}
which models our knowledge about $x$ after observing $y$. \footnote{Notice that we use \emph{generic} Bayesian notation. We use $p$ for all density functions, and use conditioning to implicitly distinguish between random variables and their realization.}

Bayesian methods have been successfully applied to a wide range of imaging problems, including for example image denoising \cite{Lebrun2013}, inpainting \cite{Niknejad2018}, deblurring \cite{BioucasDias2006}, fusion \cite{Wei2016}, unmixing \cite{Altmann2015}, tomographic reconstruction \cite{Giovannelli_book_2015}, compressive sensing sparse regression \cite{Wipf2007}, and segmentation \cite{Pereyra2013}. Solutions can then be computed by using advanced stochastic simulation and optimisation algorithms, as well as deterministic algorithms related to variational Bayes and message passing approximations \cite{pereyra2016survey, Pustelnik2016}. Moreover, log-concave formulations have also received a lot of attention lately because they lead to solutions that can be efficiently computed by using modern convex optimisation methods \cite{Chambolle2016}.

In addition to the chosen log-concave uniform likelihood, in a manner akin to \cite{pereyra2017maximum}, here we assume that the prior distribution $p(x)$ of $x$ is log-concave, and that the following Assumption holds, where $\Gamma_0(\R^N)$ denotes the set of lower semi-continuous, proper, convex functions from $\R^N$ to $]-\infty,+\infty[$.
\begin{assumption}\label{Ass:posterior}
The posterior distribution $p(x|y)$ is given by 
\begin{equation}
\label{eq:posterior2}
\begin{cases}
p(x|y) \propto {\exp \big({-g_1(x,y) - g_2(x)} \big)}\, ,	\\
g_1(x, y) =  \iota_{\mathcal{B}(y,\epsilon)} (\Phi x),
\end{cases}
\end{equation}
where $g_2(x) = - \log p(x)  \in \Gamma_0(\R^N)$, and $\iota_{\mathcal{B}(y,\epsilon)}$ denotes the indicator function\footnote{For a closed, non-empty, convex subset $\Cc$ of $\R^N$, the indicator function of $\Cc$ at a point $x \in \R^N$ is defined by $\iota_{\Cc} (x) = 0$ if $x \in \Cc$, and $\iota_{\Cc} (x) = +\infty$ otherwise.} of the $\ell_2$ ball $\mathcal{B}(y,\epsilon)$. 
\end{assumption}
As a consequence of Assumption~\ref{Ass:posterior}, the model described in equation~\eqref{eq:posterior2} is log-concave.

For example, in many imaging problems $g_2$ is of the form
\begin{equation}	\label{def:reg_gen}
(\forall x \in \R^N)\quad
g_2(x) = \lambda f( \Psi x) + \iota_{\mathcal{C}}(x),
\end{equation}
where $\lambda>0$ is the regularization parameter, $\Psi \colon \R^N \to \R^L$ is an analysis operator, $f \in \Gamma_0(\R^L)$ typically corresponds to an $\ell_p$ norm ($p\ge 1$) promoting regularity or sparsity in the domain induced by $\Psi$, and $\mathcal{C}$ is a closed non-empty convex subset of $\R^N$ encoding constraints on the solution space. 
Observe that \eqref{def:reg_gen} encompasses sparsity aware models developed during the last decade in the compressed sensing framework \cite{donoho2006compressed, candes2006compressive}. In particular, $\Psi$ may be related to a differential operator (e.g. the horizontal and vertical gradients defining the total variation (TV) image prior \cite{chambolle2004algorithm, rudin1992nonlinear}, or a possibly redundant wavelet transform \cite{Mallat_book}.

Once a model $p(x|y)$ has been defined, imaging methods generally solve the image estimation problem by computing a point estimator of $x|y$. In particular, most modern methods exploit the convexity properties of $p(x|y)$ and use the MAP estimator
\begin{equation}\label{hpd2}
x^\dagger \in \Argmax_{x \in \mathbb{R}^N} p(x|y) 
\quad \Leftrightarrow \quad
x^\dagger \in \Argmin_{x \in \mathbb{R}^N} g_1(x,y) + g_2(x) ,
\end{equation}
which can be computed efficiently using convex optimization techniques \cite{boyd2004convex,combettes2011proximal,komodakis2015playing}. In particular, the so-called proximal optimization methods received a lot of attention, for example forward-backward algorithms \cite{Attouch_Bolte_2011, beck2009fast, Chambolle2015, Chouzenoux_2016, combettes2005signal, Tseng_P_2000_j-siam-control-optim_Modified_fbs}, and primal-dual algorithms \cite{Alotaibi_A_2014_Solving_ccm, Bot_R_2014_jmiv_conv_primal_apd, Briceno_L_2011_j-siam-opt_mon_ssm, Chambolle_A_2010_first_opdacpai, condat2013primal, Combettes_P_2012_j-svva_pri_dsa, Esser_E_2010_j-siam-is_gen_fcf, komodakis2015playing, vu2013splitting}.

\begin{figure}[h]
\begin{tabular}{@{}cc@{}}
	\includegraphics[height=5.3cm]{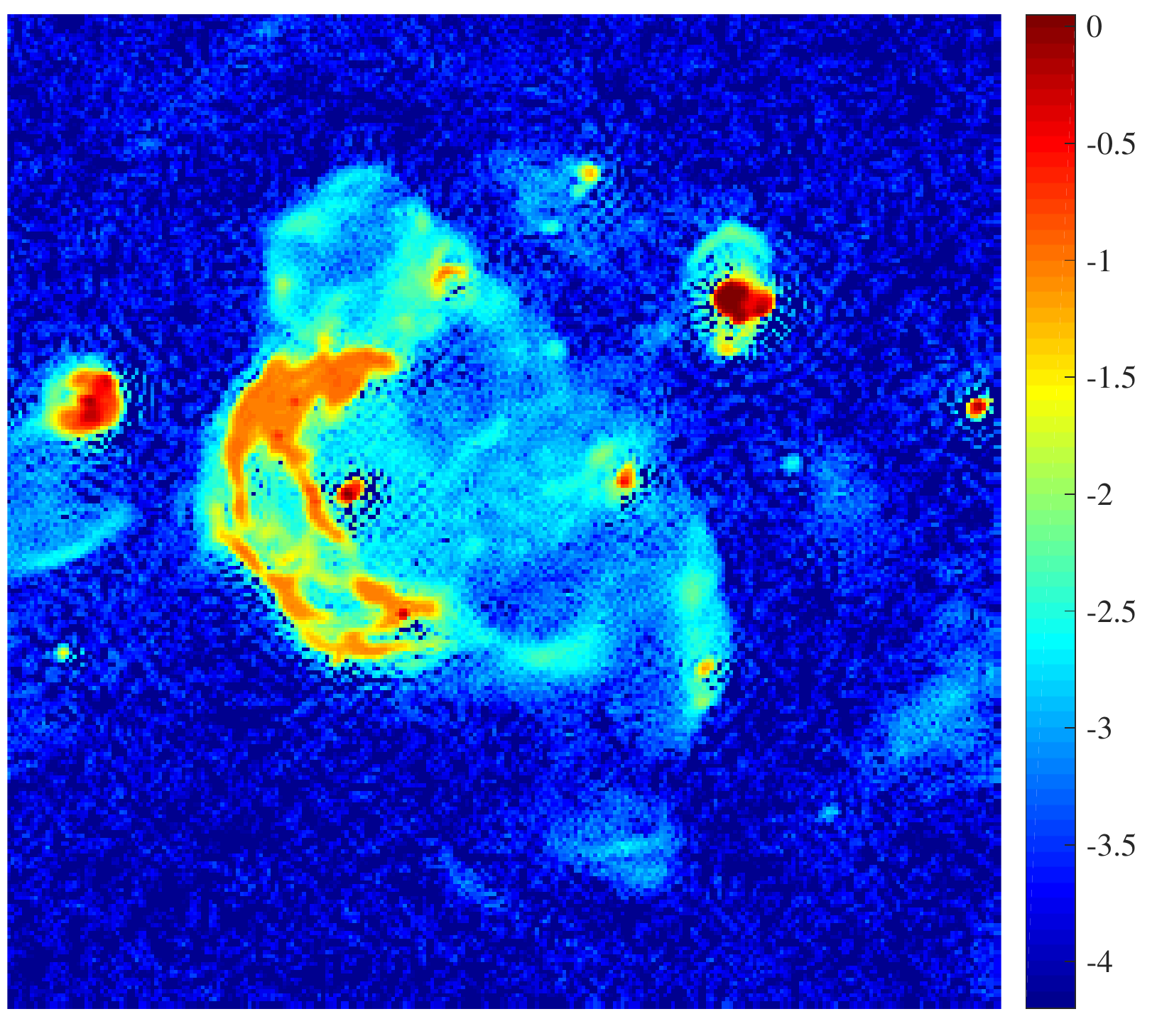}
&	\includegraphics[height=5.0cm]{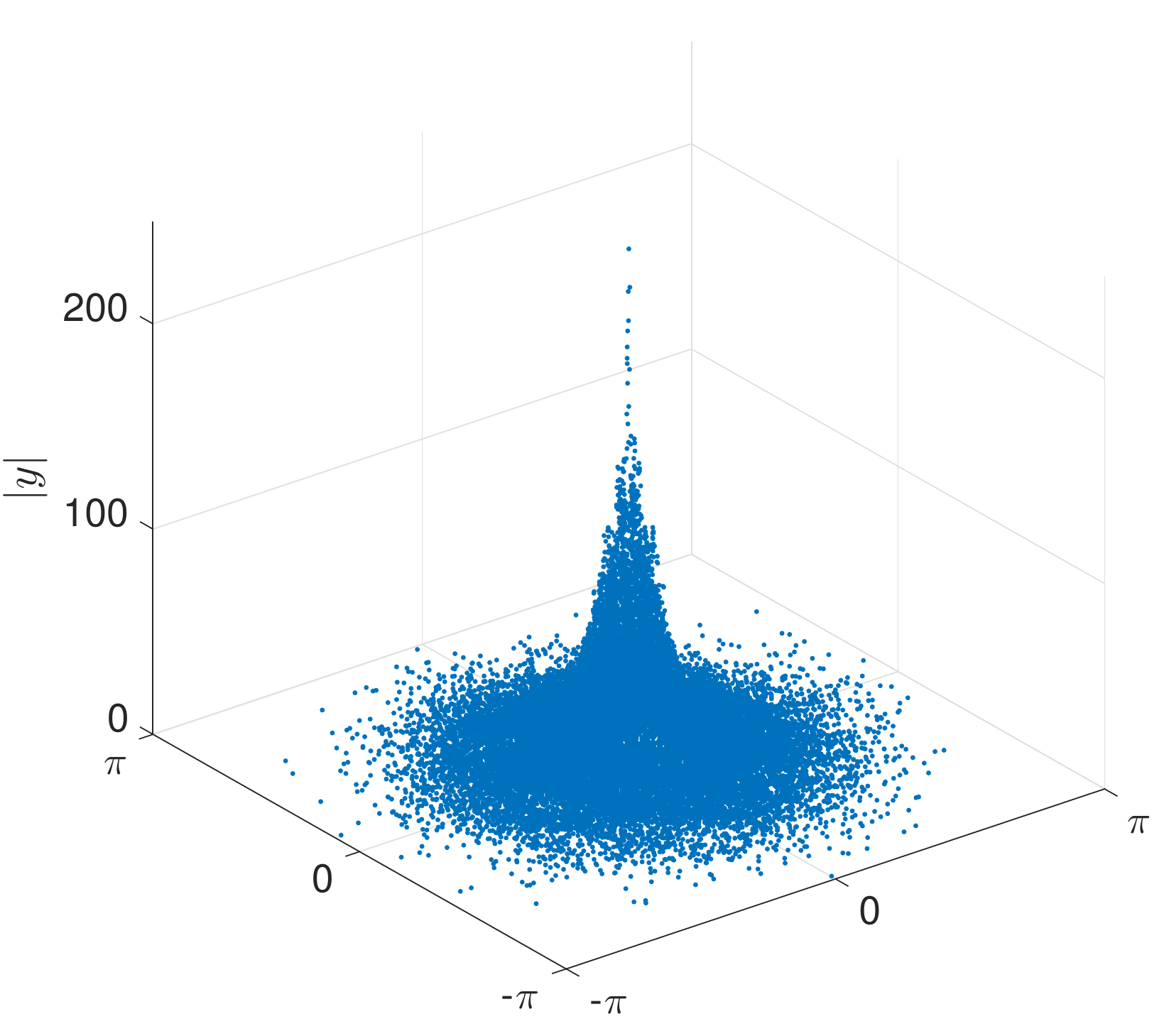}\\
(a) & (b)
\end{tabular}
\caption{\label{Fig:exampleUQ}
Illustration for the need of uncertainty quantification in the context of RI imaging with random Gaussian Fourier samplings. 
(a) MAP estimate $x^\dagger \in \R^N$ of the W28 supernova, in log scale, obtained from under-sampled (continuous) Fourier measurements $y\in \C^M$, with $N=256 \times 256$, $M = N/2$, and $\sigma^2 =0.01$. 
(b) Normalized continuous Fourier space showing the magnitude of the measurements $y$. 
}
\end{figure}

Summarising $x|y$ with a single point $x^\dagger$ has the key advantage of producing a solution that can be easily displayed and visually analysed. However, a main limitation of this approach is that it does not provide any information regarding the uncertainty in the solution delivered \cite{biegler2011large}. As explained previously, quantifying this uncertainty is important in many applications related to quantitative imaging, scientific inquiry, and image-driven decision-making, where it is necessary to analyse images as high-dimensional physical measurements and not as pictures. This analysis is particularly important in imaging problems that are ill-posed or ill-conditioned because of their high intrinsic uncertainty. For illustration, Fig.~\ref{Fig:exampleUQ}(a) shows an estimate $x^\dagger$ of the W28 supernova, obtained from the under-sampled Fourier measurements of Fig.~\ref{Fig:exampleUQ} (b), with $M/N = 0.5$, by using a Bayesian model tailored to radio-astronomical imaging \cite{wiaux2009compressed}. Clearly, the estimation problem is challenging given the severe under-sampling. For this specific imaging setup, what is the uncertainty involved in the estimate $x^\dagger$? In particular, are we confident about the different structures observed in $x^\dagger$? We expect the main structures to be reliably recovered, but is this also true for the structures of weak amplitude in the background? Perhaps they are reconstruction artefacts. 

The objective of this paper is contribute statistical imaging methodology to probe the data and investigate this type of questions. 
The proposed method, namely BUQO (Bayesian Uncertainty Quantification by Optimization), consists in quantifying the uncertainty of the structures under scrutiny by performing a Bayesian hypothesis test. This test consists of two steps: firstly, it postulates that the structures are not present in the true image, and secondly the data and prior knowledge are used to determine if this null hypothesis is rejected with high probability. Computing such tests for imaging problems is often intractable due to the high dimensionality involved. In this work, we propose to leverage probability concentration phenomena and the underlying convex geometry to formulate the Bayesian hypothesis test as a convex problem. The resulting problem can then be solved efficiently by using scalable optimization algorithms. This allows scaling to high-resolution and high-sensitivity imaging problems that are computationally unaffordable for other Bayesian computation approaches. 
To illustrate the proposed BUQO methodology, we apply it to a range of challenging Fourier imaging problems arising in astronomy and medicine.

The remainder of the paper is organized as follows. Section~\ref{Ssec:Bayes_quant},  introduces the Bayesian uncertainty quantification framework that underpins our work. The proposed methodology is presented in Section~\ref{sec:method}. Section~\ref{sec:results} illustrates the method on two challenging Fourier imaging problems related to radio astronomy and magnetic resonance imaging. Section~\ref{sec:discuss} is a discussion of the proposed methodology. Conclusions and perspectives for future work are finally reported in  Section~\ref{Sec:conclusion}.

\section{Imaging and Bayesian uncertainty quantification}
\label{Ssec:Bayes_quant}

The Bayesian \linebreak paradigm provides a powerful methodological framework to analyse uncertainty in imaging inverse problems. One main approach, adopted in \cite{pereyra2017maximum, Cai_2017_BayesII}, is to compute confidence or credible regions that indicate where $x|y$ takes values with high probability. This allows testing if specific images belong to the set of likely solutions and making some preliminary analyses. However, its capacity for formal uncertainty quantification is very limited. 

To properly assess the degree of confidence in specific image structures it is necessary to perform a Bayesian hypothesis test. Formally, we postulate two hypotheses: 
\begin{equation*}
\begin{split}
H_0:\quad& \textrm{The structure of interest is ABSENT in the true image}\, ,\\
H_1:\quad& \textrm{The structure of interest is PRESENT in the true image}.
\end{split}
\end{equation*}
These hypotheses split the image space $\mathbb{R}^N$ in two regions: a set $\Sc \subset \R^N$ associated with $H_0$ containing all the images (i.e. solutions) \emph{without the structure of interest}, and the complement $\mathbb{R}^N \setminus \Sc$ associated with $H_1$. 
The goal of the hypothesis test is then to determine if the observed data $y$ supports $H_0$ or $H_1$; that is, if it supports the claim that the estimated structure is real or corresponds to a reconstruction artefact.
This is formalized by using Bayesian decision theory \cite{robert2007bayesian}, a statistical framework for decision-making under uncertainty. 
Precisely, from Bayesian decision theory, we reject $H_0$ in favour of $H_1$ with significance level $\alpha \in ]0,1[$ if
\begin{equation}\label{test:eqn}
\begin{split}
P\left[H_0 | y \right] &= P\left[x \in \Sc | y \right] = \int_{\Sc} p(x|y)\textrm{d}x \leq \alpha\,, 
\end{split}
\end{equation}
or equivalently, if the ratio of posterior probabilities
\begin{equation*}
\frac{P\left[H_1 | y \right]}{P\left[H_0 | y \right]} = \frac{P\left[x \in \mathbb{R}^N \setminus \Sc | y \right]}{P\left[x \in \Sc | y \right]} \geq \dfrac{1-\alpha}{\alpha}\, , 
\end{equation*}
where we recall that rejecting $H_0$ means that the structure considered is real (i.e. not an artefact).

Unfortunately, computing hypothesis tests for images requires calculating probabilities w.r.t. $p(x|y)$, which are generally intractable because of the high-\linebreak dimensionality involved. These probabilities can be approximated with high accuracy by Monte Carlo integration \cite{robert2004monte} (for example by using the state-of-the-art proximal Markov chain Monte Carlo (MCMC) algorithm \cite{durmus2016efficient, pereyra2016proximal}). Nevertheless, the computational cost associated with these methods is often several orders of magnitude higher than that involved in computing the MAP estimator by convex optimization \cite{Chambolle2016}, which will be discussed later in Section~\ref{Ssec:discuss:comp_MCMC}, in the context of our simulations. Consequently, most of the imaging methods used in practice do not quantify uncertainty.

\section{Proposed BUQO method}
\label{sec:method}

\subsection{Uncertainty quantification approach}

A main contribution of this paper is to exploit the log-concavity of $p(x|y)$ to formulate the hypothesis test \eqref{test:eqn} as a convex program that can be solved straightforwardly by using modern convex optimization algorithms when $\Sc$ is a convex set. The proposed method only assumes knowledge of the MAP estimator $x^\dagger$, and does not require computing probabilities. We first introduce the convex program associated with \eqref{test:eqn}, then describe the proposed convex optimization algorithm used to solve it, and subsequently present our approach to specify the set $\Sc$ associated with $H_0$. 
In the remainder of the paper, we make the following assumption on $\Sc$.
\begin{assumption}	\label{Ass:S}
The subset $\Sc$ of $\R^N$ is convex.
\end{assumption}

The proposed method solves the hypothesis test by comparing $\Sc$ with the region of the solution space where most of the posterior probability mass of $x|y$ lies. Such regions are known as posterior credible sets in the Bayesian literature \cite{pereyra2017maximum}. Precisely, a set $\Cc_{\alpha}$ is a posterior credible region with confidence level $(1-\alpha)$ if $P\left(x \in \Cc_{\alpha}|y\right) = 1-\alpha$ for $\alpha \in ]0,1[$. Computing credible regions exactly is difficult because it requires calculating probabilities w.r.t. $p(x|y)$, which is too computationally expensive when the dimension of $x$ is large. Here we take advantage of the conservative credible region recently proposed in \cite{pereyra2017maximum}, which is available for free in problems solved by MAP estimation. Precisely, for any $\alpha \in ]4\exp(-N/3),1[$, we use the region
\begin{equation}\label{eq:hpdapp1}
\widetilde{\Cc}_{\alpha} = \left\{ x \in \R^N \mid \Phi x \in \Bc_2(y,\varepsilon) \text{ and } g_2(x) \leq \tilde{\eta}_\alpha\right\},  
\end{equation}
where the threshold $\tilde{\eta}_\alpha = g_2(x^\dagger) + N(\tau_\alpha + 1)$ with $\tau_\alpha = \sqrt{{16 \log(3/\alpha)}/{N}}$ and $x^\dagger $ is the MAP estimator \eqref{hpd2} such that $\Phi x^\dagger \in \Bc_2(y,\varepsilon)$.

The set $\widetilde{\Cc}_{\alpha}$ is a conservative Bayesian confidence region for $x|y$; i.e., $P(x \in \widetilde{\Cc}_{\alpha} | y ) \geq 1-\alpha$. Observe that, in addition to being computationally straightforward, $\widetilde{\Cc}_{\alpha}$ is also a convex set because $p(x|y)$ is log-concave and has convex superlevel sets. This property will play a central role in our algorithm to compute the hypothesis test.
Also note that the highest-posterior-density region $\mathcal{C}^*_\alpha = \{x | g_1(x,y)+g_2(x) \le \eta_\alpha\} $, with $\eta_\alpha \in \mathbb{R}$ chosen such that $\int_{\mathcal{C}_\alpha^*} p(x|y)dx = 1-\alpha$, is the tightest credibility region in the sense of compactness or minimum volume \cite{robert2007bayesian}. 
It is also a convex set as it corresponds to the sublevel set of a convex function. The approximate credibility region $\widetilde{\mathcal{C}}_\alpha$ defined in \eqref{eq:hpdapp1} results from an analytical approximation $\widetilde{\eta}_\alpha$, with $\widetilde{\eta}_\alpha \ge \eta_\alpha$, that can be obtained by leveraging the concentration of measure phenomenon. 
The set $\widetilde{\mathcal{C}}_\alpha$ is the tightest approximation of $\mathcal{C}^*_\alpha$ that can be obtained from the knowledge of the MAP estimate (which is computed by convex optimisation) \cite{pereyra2017maximum}.
It also follows from its definition that $\widetilde{\mathcal{C}}_\alpha$ is convex.

\begin{theorem}	\label{Thm:proba_feas}
Consider the posterior distribution $p(x|y)$ given in \linebreak Assumption~\ref{Ass:posterior}. 
For any $\alpha \in ]4\exp(-N/3),1[$, let $\widetilde{\Cc}_\alpha$ be the convex set \eqref{eq:hpdapp1}, and $\Sc$ be the set associated with the null hypothesis $H_0$ satisfying Assumption~\ref{Ass:S}. If the following non-feasibility condition holds
$$
\widetilde{\Cc}_\alpha \cap \Sc = \emp\,,  
$$
then $H_0$ is rejected in favour of $H_1$ with significance $\alpha$,
$$
P\left[H_0 | y \right] \leq \alpha\quad \textrm{and}\quad \frac{P\left[H_1 | y \right]}{P\left[H_0 | y \right]} \geq \frac{1-\alpha}{\alpha}\,.
$$
\end{theorem}

\begin{proof}
If $\widetilde{\Cc}_\alpha \cap \Sc = \emp$, then we have $\Sc \subset \mathbb{R}^N \setminus \widetilde{\Cc}_\alpha$, 
which implies that
$P\left[H_0 | y \right] = P\left[x \in \Sc | y \right] \le 1-P\left[x \in \widetilde{\Cc}_\alpha | y \right]$. 
In addition, according to \cite[Theorem 3.1.]{pereyra2017maximum}, for any $\alpha \in ]4\exp(-N/3),1[$, we have $P\left[H_0 | y \right] \leq \alpha$, hence concluding the proof.
\end{proof}

\begin{remark}	\label{rq:contrap}
The converse of Theorem \ref{Thm:proba_feas} is not true; i.e.,  $\widetilde{\Cc}_\alpha \cap \Sc \neq \emp$ does not imply $P\left[H_0 | y \right] \geq \alpha$. It is possible that $\Sc \subset \widetilde{\Cc}_\alpha$ with $P\left[H_0 | y \right]$ arbitrarily small. Hence, when $\widetilde{\Cc}_\alpha \cap \Sc \neq \emp$ we fail to reject the null hypothesis $H_0$.
\end{remark}

From Theorem~\ref{Thm:proba_feas}, we can verify if $P\left[H_0 | y \right] \leq \alpha$ by solving the following problem:
\begin{equation}	\label{pb:feas_emp}
\text{determine if }  \widetilde{\Cc}_\alpha \cap \Sc = \emp\,.
\end{equation}
There are two possible outcomes: either $\widetilde{\Cc}_\alpha \cap \Sc = \emp$ or $\widetilde{\Cc}_\alpha \cap \Sc \neq \emp$. If $\widetilde{\Cc}_\alpha \cap \Sc = \emp$ for a small value of $\alpha$, we conclude that there is strong evidence for the structure considered. Moreover, in that case we also compute the distance between $\widetilde{\Cc}_\alpha $ and $ \Sc$,
\begin{equation}
\dist(\widetilde{\Cc}_\alpha, \Sc)
= \inf \| \widetilde{\Cc}_\alpha - \Sc \|
= \inf \Big\{ \| x_{\widetilde{\Cc}_\alpha}- x_{\Sc} \| \, : \, (x_{\widetilde{\Cc}_\alpha}, x_{\Sc}) \in \widetilde{\Cc}_\alpha \times \Sc \Big\}.
\end{equation}
We will later discuss using this distance to quantify the uncertainty in the intensity of the structure considered (precisely, to lower bound the structure's intensity).

If we determine that $\widetilde{\Cc}_\alpha \cap \Sc \neq \emp$, this suggests that the evidence for the structure under scrutiny is weak. In particular, that the structure is not present in all of the images that $p(x|y)$ considers likely solutions to our inverse problem. Following on from this, to produce an example of such solution we solve the feasibility problem
\begin{equation}	\label{pb:feas_gen}
\text{find } x^\ddagger \in \widetilde{\Cc}_\alpha \cap \Sc\,.
\end{equation}
We view $x^\ddagger$ as a counter-example solution where the structure of interest does not exist.

Furthermore, we propose to rely on the von Neumann algorithm \cite{vonNeumann_1951, Halperin_1962, Cheney_goldstein_1959, Bregman_1965} to solve problem~\eqref{pb:feas_emp}-\eqref{pb:feas_gen}. 
This POCS algorithm alternates Euclidean projections onto the set $\widetilde{\Cc}_\alpha$ and the set $\Sc$. 
Formally, the Euclidean projection of $x \in \R^N$ onto $\Sc$ is 
\begin{align}
\proj_{\Sc}(x) 
&=	\argmind{u \in \Sc} \| x - u \|^2.
\end{align}
The main iterations of the von Neumann method are described in Algorithm~\ref{algo:POCS_gen}.

\begin{algorithm}[h!]
	\caption{POCS algorithm to solve problem~\eqref{pb:feas_gen}.}
	\begin{algorithmic}[1]
        \State 	
        \textbf{Initialization:} 
        Let $x^{(0)} \in \Sc$.
        \newline
        \State	\vspace{0.5em}
        \textbf{For} $k = 0, 1, \ldots$
		\State	\vspace{0.3em}	\label{algo:POCS_gen:projC}
		$\displaystyle \quad\quad
		x^{(k+\frac12)} = 
		\proj_{\widetilde{\Cc}_\alpha} \big( x^{(k)} \big)$
		\State	\vspace{0.3em}	\label{algo:POCS_gen:projS}
		$ \displaystyle \quad\quad
		x^{(k+1)} = 
		\proj_{\Sc} \big( x^{(k+\frac12)} \big)	$
        \State	\vspace{0.5em}
        \textbf{end for}
	\end{algorithmic}
    \label{algo:POCS_gen}
\end{algorithm}

The following convergence result from \cite[Thm. 4.8]{Bauschke_Borwein_1994} allows to determine if the intersection between $\widetilde{\Cc}_\alpha$ and $ \Sc$ is empty or not. 
\begin{theorem}[Thm. 4.8 in \cite{Bauschke_Borwein_1994}]		\label{thm:POCS:cvg}
Let $(x^{(k+\frac12)})_{k \in \N}$ and $(x^{(k)})_{k \in \N}$ be sequences generated by Algorithm~\ref{algo:POCS_gen}.
Under Assumptions~\ref{Ass:posterior} and~\ref{Ass:S}, the following assertions hold:
\begin{enumerate}
\item	\label{POCS:cv:i}
If $\widetilde{\Cc}_\alpha \cap \Sc \neq \emp$, then the sequences $(x^{(k+\frac12)})_{k \in \N}$ and $(x^{(k)})_{k \in \N}$ both converge to a point $x^\ddagger \in \widetilde{\Cc}_\alpha \cap \Sc$.
\item	\label{POCS:cv:ii}
If $\widetilde{\Cc}_\alpha \cap \Sc = \emp$, then the sequence $(x^{(k+\frac12)})_{k \in \N}$ converges to $x_{\widetilde{\Cc}_\alpha}^\ddagger \in \widetilde{\Cc}_\alpha$ and the sequence $(x^{(k)})_{k \in \N}$ converges to $x_{\Sc}^\ddagger \in \Sc$. In addition, we have $\| x_{\widetilde{\Cc}_\alpha}^\ddagger - x_{\Sc}^\ddagger \| = \dist(\widetilde{\Cc}_\alpha, \Sc)$.
\end{enumerate}
\end{theorem}
Versions of the POCS method with acceleration and approximated projections are discussed in Section~\ref{Ssec:discuss:fast}. 
A simple example illustrating Theorem~\ref{thm:POCS:cvg} is given in Figure~\ref{Fig:POCS_illustration}. 
\begin{figure}
\begin{center}
\includegraphics[width=8cm]{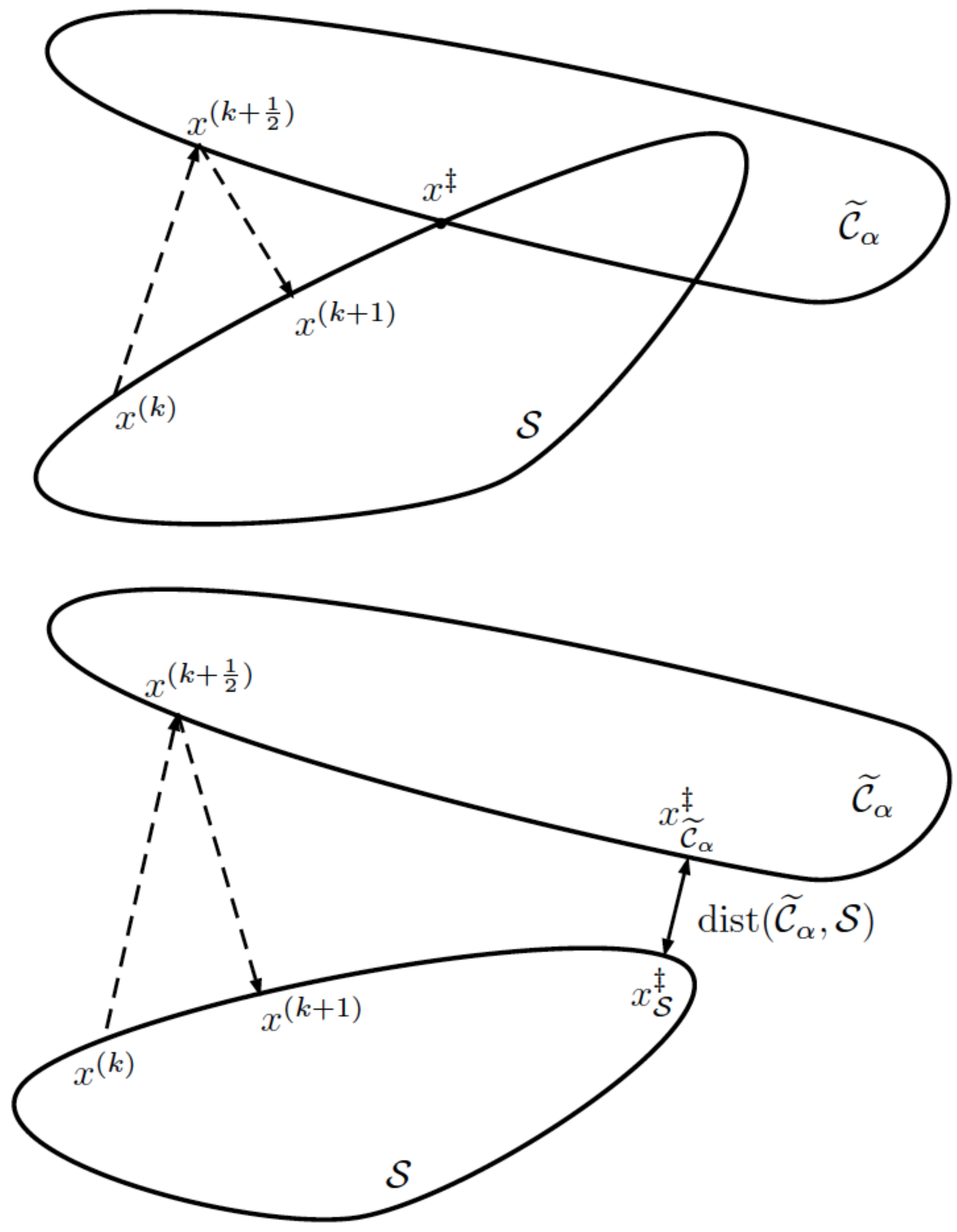} 
\end{center}
\caption{\label{Fig:POCS_illustration}
Illustration of few iterations of Algorithm~\ref{algo:POCS_gen} in the case when $\widetilde{\Cc}_\alpha \cap \Sc \neq \emp$ (top) and $\widetilde{\Cc}_\alpha \cap \Sc = \emp$ (bottom). }
\end{figure}

\subsection{Choice of the set $\Sc$}
\label{Ssec:def_S}

We are now in a position to present our approach to construct the set $\Sc$. 
This construction should be intuitive, easy to interpret, and sufficiently flexible to accommodate a broad range of scenarios. Also, it should guarantee that $\Sc$ is convex (Assumption~\ref{Ass:S}) and hence that the non-feasibility condition $\widetilde{\Cc}_\alpha \cap \Sc = \emp$ is easy to evaluate. 

We define $\Sc$ as the intersection of $L$ convex sets $\Sc_1,\ldots, \Sc_L$ related to different properties that we wish to encode in the test, i.e.,
\begin{equation}\label{defS:eqn}
\Sc = \Big\{ x \in \R^N \, \mid \,
                  (\forall l \in \{1, \ldots, L \}) \quad x \in \Sc_l \Big\}\, .
\end{equation}
It is important to emphasize that the projection onto the set $\Sc$, as defined above, may not have a closed form expression. In this case it is necessary to adopt a sub-iterative approach, for example by using a best-approximation method (e.g. Dykstra's algorithm, see \cite{Bauschke_Borwein_1994, bauschke2017convex} for details). Similarly, when the sets $(\Sc_l)_{1 \le l \le L}$ are sophisticated, then primal-dual methods can be used \cite{komodakis2015playing}.
We illustrate the definition of $\Sc$ with the following two examples that will be also relevant for the experiments that we report in Section~\ref{sec:results}. The first example is related to spatially localized image structures appearing in $x^\dagger$, such as a tumour in a medical image. The second example is related to background removal; this is useful for instance to assess low-intensity sources appearing in the background of astronomical images. 
Before giving the particular definitions associated with the localized structures and the background, we need to introduce an additional image $x^\dagger_{\Sc} \in \Sc$ defined such that $x^\dagger - x^\dagger_{\Sc}$ corresponds to the structure of interest, as it appears in the MAP estimate $x^\dagger$ (formal definitions are given for the two particular types of structures defined below). 
In addition, we introduce the operator $\Mc \colon \R^N \to \R^{N_{\Mc}}$ selecting the structure of interest. For an image $x \in \R^N$, $\Mc(x) \in \R^{N_{\Mc}}$ denotes either the region of the spatially localized structure (in Definition~\ref{example:structure}), or the background (in Definition~\ref{example:background}). In both the cases, $\Mc^c \colon \R^N \to \R^{N-N_{\Mc}}$ denotes the complementary operator of $\Mc$.
 
\begin{definition}[Spatially localized structures]	\label{example:structure}
To assess the confidence in a structure localized in a region $\Mc(x)$ of the image $x$, we construct $\Sc$ by using an inpainting technique $\Lc$ that fills the pixels $\Mc(x)$ with the information in the other image pixels $\Mc^c(x) $. To ensure that $\Sc$ is convex we define $\Lc \colon [0,+\infty[^{N-N_{\Mc}} \to [0,+\infty[^{N_{\Mc}}$ as a positive linear operator, and allow deviations from this linear inpainting by as much as $\pm \tau$ per pixel, for some tolerance value $\tau > 0$. This inpainting could potentially amplify the energy in $\Mc(x)$ and lead to artificial structures. To prevent this we enforce that $\Mc(x)\in\Bc_2 ( b,\theta )$ where $b \in \R^{N_{\Mc}}$ is a reference background level for $\Mc^c(x)$ and $\theta > 0$ controls the energy in $\Mc(x)$. Formally, we use \eqref{defS:eqn} with $L=3$ given by
\begin{equation}	\label{def:S:structure}
\begin{cases}
\Sc_1 = [0,+\infty[^N,	\\
\Sc_2 = \Big\{ x \in \R^N \, \mid \,
				     \Mc(x) - \Lc \big( \Mc^c (x) \big) \in [-\tau, \tau]^{N_{\Mc}} \Big\},	\\
\Sc_3 = \Big\{ x \in \R^N \, \mid \,
				     \Mc(x) \in \Bc_2 ( b,\theta ) \Big\}.
\end{cases}
\end{equation}
In this case, we define $x^\dagger_{\Sc}$ such that $\Mc(x^\dagger_{\Sc}) = \Lc \big( \Mc^c(x^\dagger) \big)$ and $\Mc^c (x^\dagger_{\Sc}) = \Mc^c(x^\dagger)$. In addition, $b$ and $\theta$ are chosen such that $x^\dagger_{\Sc}$ belongs to $\Sc$.
\end{definition}

\begin{definition}[Background removal]	\label{example:background}
To assess the confidence in low-intensity structures appearing in the background (e.g., determine if they exist or if they are artefacts due to the reconstruction process), we use \eqref{defS:eqn} with $L=2$ given by
\begin{equation}	\label{def:S:background}
\begin{cases}
\Sc_1 = [0,+\infty[^N,	\\
\Sc_2 = \big\{ x \in \R^N \, \mid \,
				     \Mc(x) \in [\underline{\tau}, \overline{\tau}]^{N_{\Mc}} \big\}\, ,
\end{cases}
\end{equation}
where $\overline{\tau} >\underline{\tau} \geq 0$ are tolerance parameters on the reference background level $b \in \R^{N_{\Mc}}$ for $\Mc(x)$. 
In this case, we define $x^\dagger_{\Sc}$ such that $\Mc(x^\dagger_{\Sc}) = b $ and $\Mc^c (x^\dagger_{\Sc}) = \Mc^c(x^\dagger)$. 
\end{definition}

To conclude, we now discuss our approach for using the distance between $\Sc$ and $\widetilde{\Cc}_\alpha$ to bound the intensity of the structure of interest. Recall that $\dist(\widetilde{\Cc}_\alpha, \Sc) = \| x_{\widetilde{\Cc}_\alpha}^\ddagger - x_{\Sc}^\ddagger \|$ is a by-product of Algorithm~\ref{algo:POCS_gen}. 
To relate this quantity to the structure's intensity we define the normalised intensity of the structure as the ratio between $\dist(\widetilde{\Cc}_\alpha, \Sc)$ and the intensity of the structure present in the MAP, given by $\| x^\dagger - x^\dagger_{\Sc} \|$:
\begin{equation}	\label{def:rho_alpha}
\rho_\alpha 
=	\frac{\| x^\ddagger_{\Sc} - x^\ddagger_{\widetilde{\Cc}_\alpha}\|}{\| x^\dagger - x^\dagger_{\Sc} \|}
=	\frac{\dist(\Sc, \widetilde{\Cc}_\alpha)}{\| x^\dagger - x^\dagger_{\Sc} \|}.
\end{equation}
Notice that $\rho_\alpha = 0$ is equivalent to $\dist(\Sc, \widetilde{\Cc}_\alpha) = 0$. Consequently when $\rho_\alpha = 0$, we conclude that $\Sc \cap \widetilde{\Cc}_\alpha \neq \emp$ and $H_0$ is not rejected. In the case when $\rho_\alpha >0$, we conclude that $H_0$ is rejected and the value of $\rho_\alpha$ corresponds to the energy percentage of the structure that is confirmed in the MAP estimate.

\subsection{Illustration example}
\label{ssec:Illust}

In this section, we provide a simulation example to illustrate the application of the proposed approach for uncertainty quantification. We consider the hypothesis test described in Section~\ref{Ssec:Bayes_quant}, with significance $\alpha = 1\%$.
We focus on the example in radio-astronomical imaging described in Fig.~\ref{Fig:exampleUQ}. We propose to quantify the uncertainty of the spatially localized structure appearing on the left of the image. 
More precisely, the compact source of interest is highlighted in red on the MAP estimate $x^\dagger$, in the top-left image of Fig.~\ref{Fig:Illustration_high}. 
Mathematical details for the definition of the set $\Sc$ are given in Section~\ref{Sssec:RI}.

The two resulting images $x^\ddagger_{\widetilde{\Cc}_\alpha}$ and $x^\ddagger_{\Sc}$ generated by Algorithm~\ref{algo:POCS_gen} are provided on the top-center and top-right images of Fig.~\ref{Fig:Illustration_high}, respectively. 
On the one hand, it can be visually observed that $x^\ddagger_{\widetilde{\Cc}_\alpha}$ and $x^\ddagger_{\Sc}$ are very similar. The structure is neither visible in $x^\ddagger_{\widetilde{\Cc}_\alpha}$, nor in $x^\ddagger_{\Sc}$. This similarity is highlighted on the bottom-center and right images of Fig.~\ref{Fig:Illustration_high}, corresponding to the images $x^\ddagger_{\widetilde{\Cc}_\alpha}$ and $x^\ddagger_{\Sc}$ zoomed in the area of interest. 
On the other hand, this visual observation is confirmed by the value of $\rho_\alpha$. For this example, the structure's confirmed intensity percentage is equal to $\rho_\alpha = 2.52\%$. Even if this value is not zero, we consider that we have $\rho_\alpha \approx 0\%$ due to the numerical approximations involved (see Section~\ref{sec:results} for details). 
Consequently, we conclude that $\widetilde{\Cc}_\alpha \cap \Sc \neq \emp$, $H_0$ is not rejected, and the evidence for the structure is weak. 

The uncertainty quantification conclusions drawn above are characterized by the simulation parameters $(M/N, \sigma^2) = (0.5, 0.01)$. It is reasonable to assume that the uncertainty should decrease if either $M/N$ increases, or $\sigma$ decreases. For the sake of the illustration, we now investigate the Bayesian uncertainty of the same compact source, but we consider the case when $(M/N, \sigma^2) = (1, 0.01)$. Note that other cases will be provided in Section~\ref{sec:results}. 
Results for this second case are provided in Fig.~\ref{Fig:Illustration_low}.
The images $x^\ddagger_{\widetilde{\Cc}_\alpha}$ and $x^\ddagger_{\Sc}$ generated by Algorithm~\ref{algo:POCS_gen} are provided on the top-center and top-right images of Fig.~\ref{Fig:Illustration_low}, respectively. 
It can be visually observed that $x^\ddagger_{\widetilde{\Cc}_\alpha}$ and $x^\ddagger_{\Sc}$ are different: the structure is visible in $x^\ddagger_{\widetilde{\Cc}_\alpha}$, while it is not visible in $x^\ddagger_{\Sc}$. This difference is highlighted on the bottom-center and right images of Fig.~\ref{Fig:Illustration_low}, corresponding to the images $x^\ddagger_{\widetilde{\Cc}_\alpha}$ and $x^\ddagger_{\Sc}$ zoomed in the area of interest. 
The visual observation is confirmed by the value of $\rho_\alpha= 18.76\%$.  
Consequently, we can conclude that $\widetilde{\Cc}_\alpha \cap \Sc = \emp$, and, according to Theorem~\ref{Thm:proba_feas}, $H_0$ is rejected with significance $\alpha= 1\%$ (recall that rejecting $H_0$ is equivalent to stating that the structure considered is real, not an artefact).

\begin{figure}[t!]
\begin{center}
\begin{tabular}{@{}c@{}c@{}c@{}}
	\includegraphics[height=3.8cm]{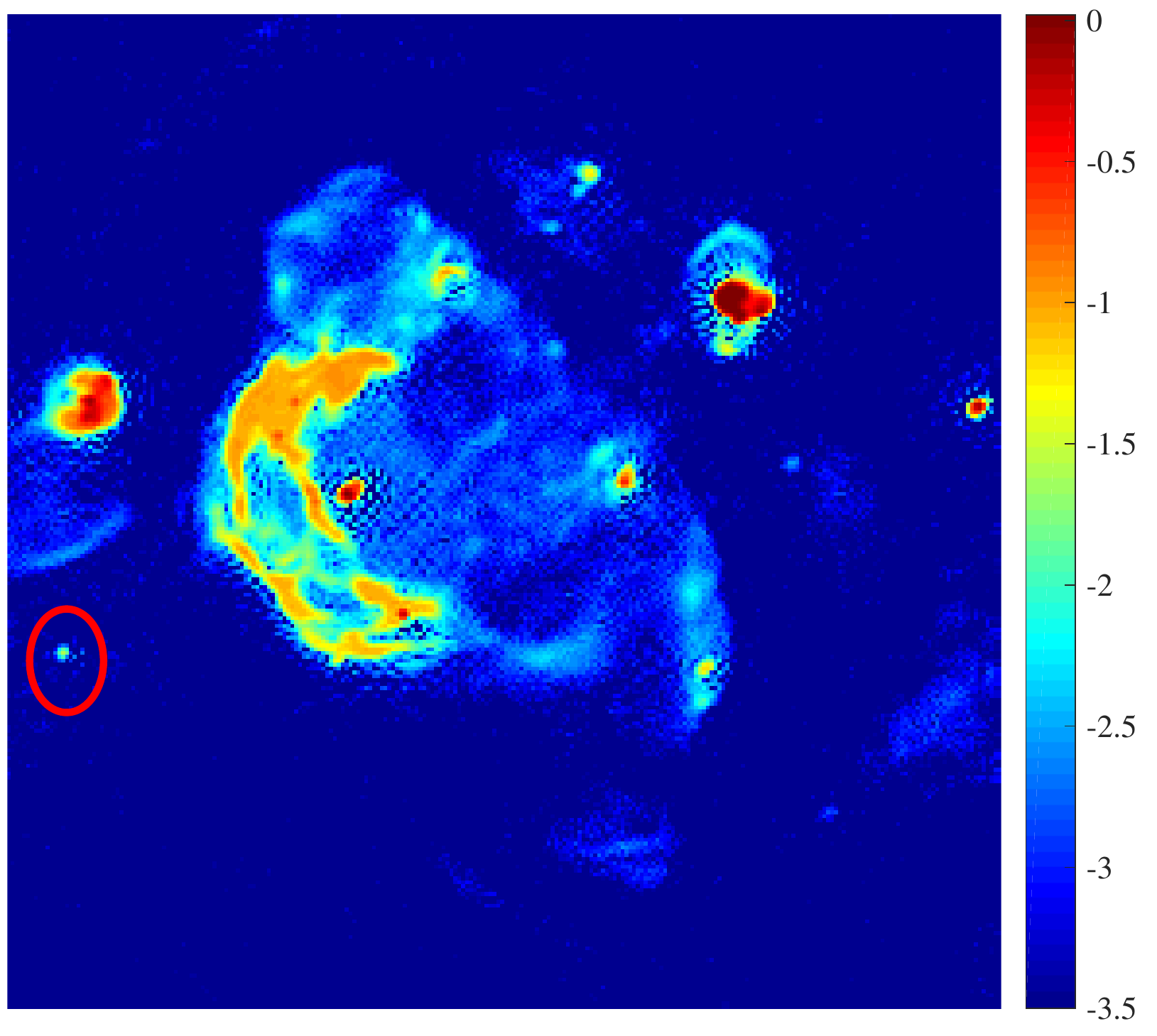}
&	\includegraphics[height=3.8cm]{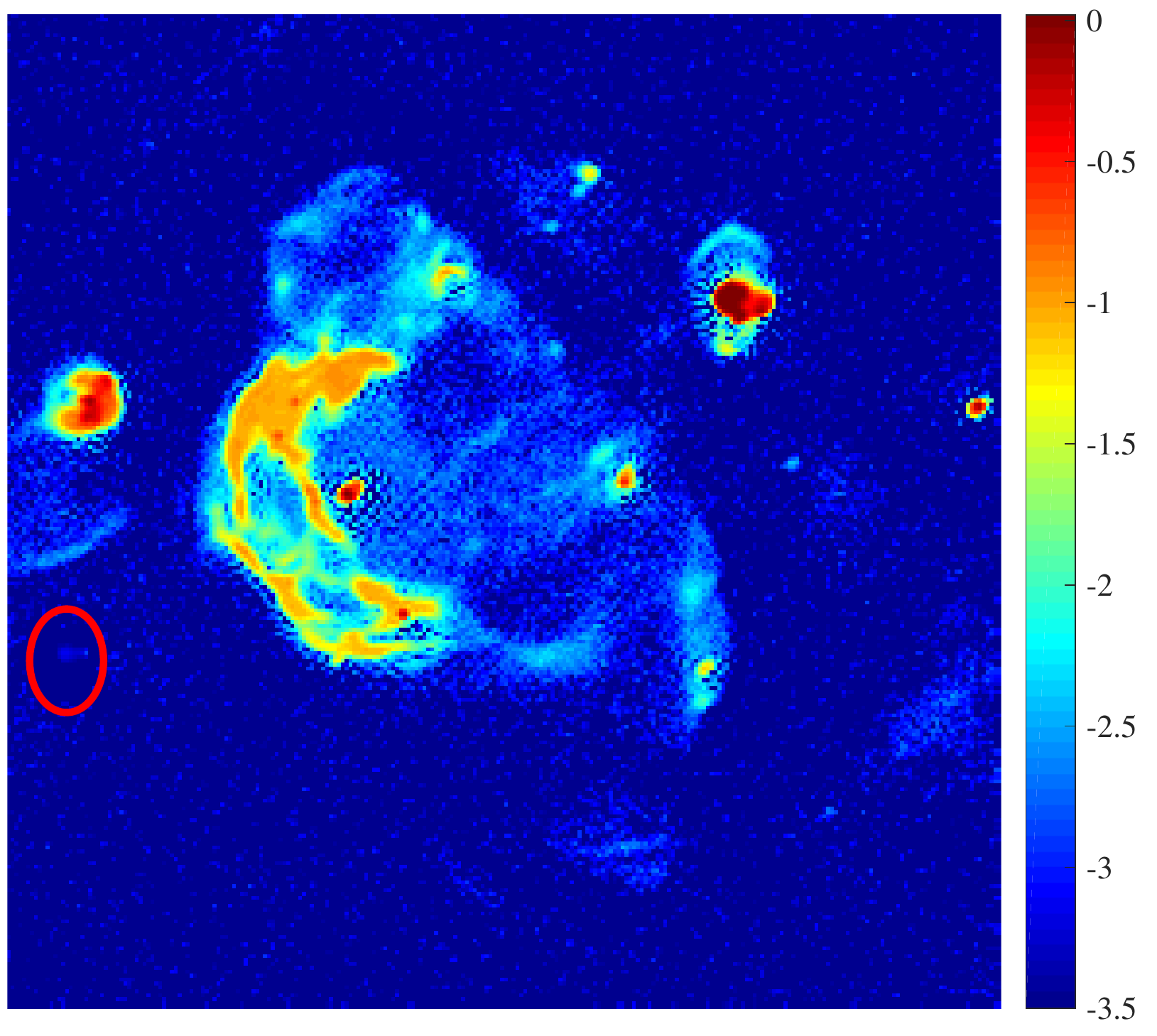}
&	\includegraphics[height=3.8cm]{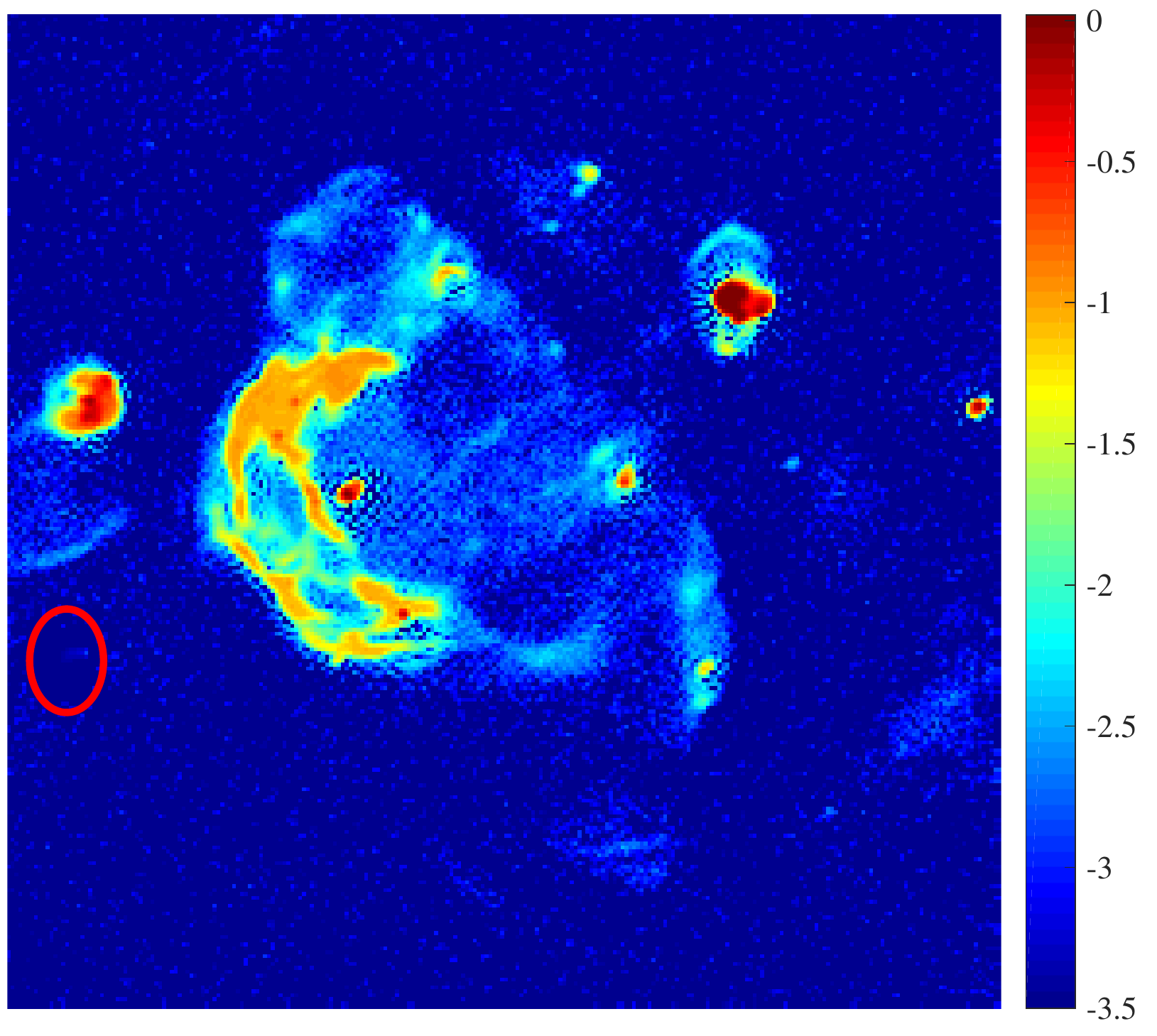}	\\[-0.2cm]
	$x^\dagger$	&	$x^\ddagger_{\widetilde{\Cc}_\alpha}$	&	$x^\ddagger_{\Sc}$	\\[0.1cm]
	\includegraphics[height=3.8cm]{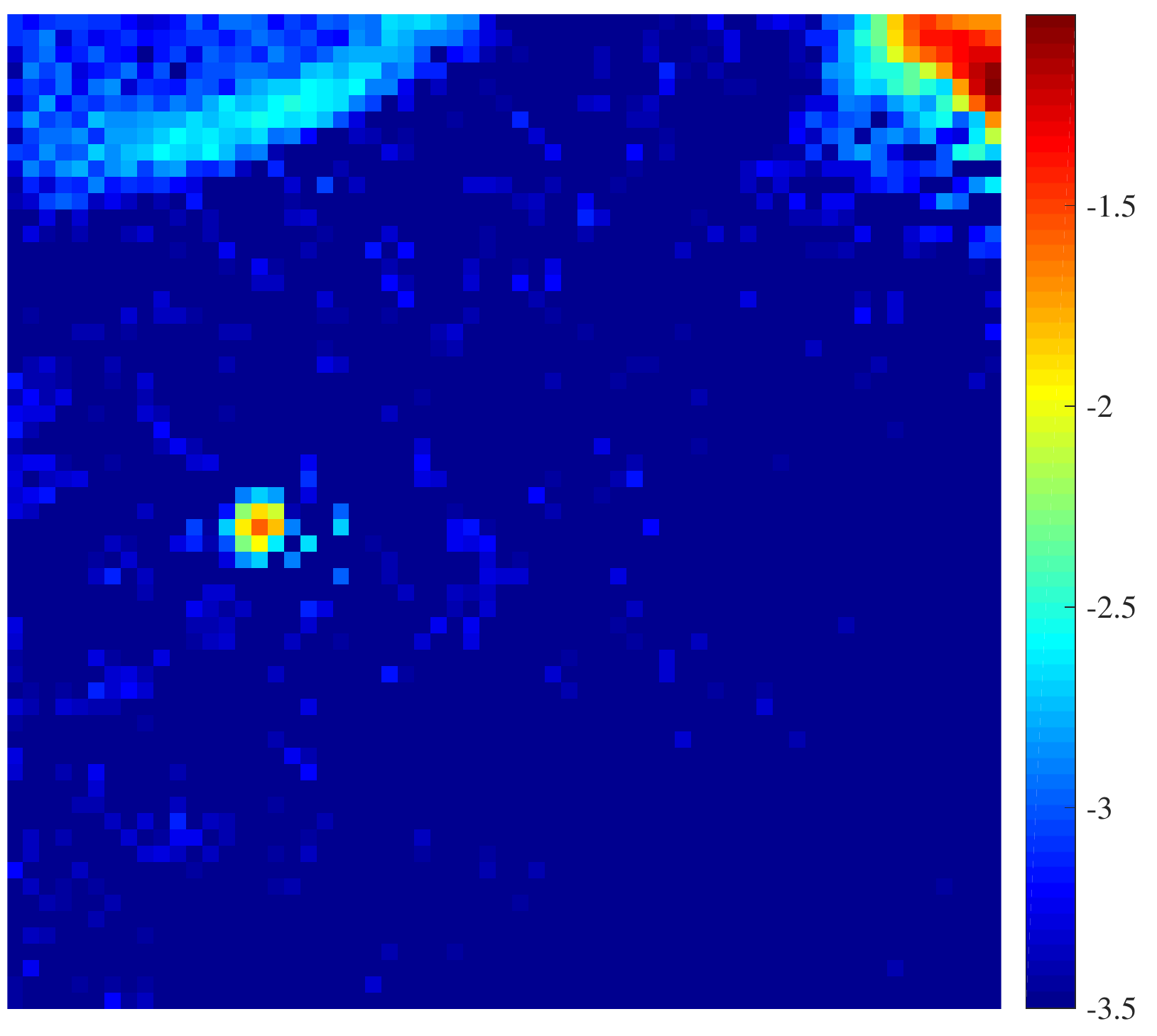}
&	\includegraphics[height=3.8cm]{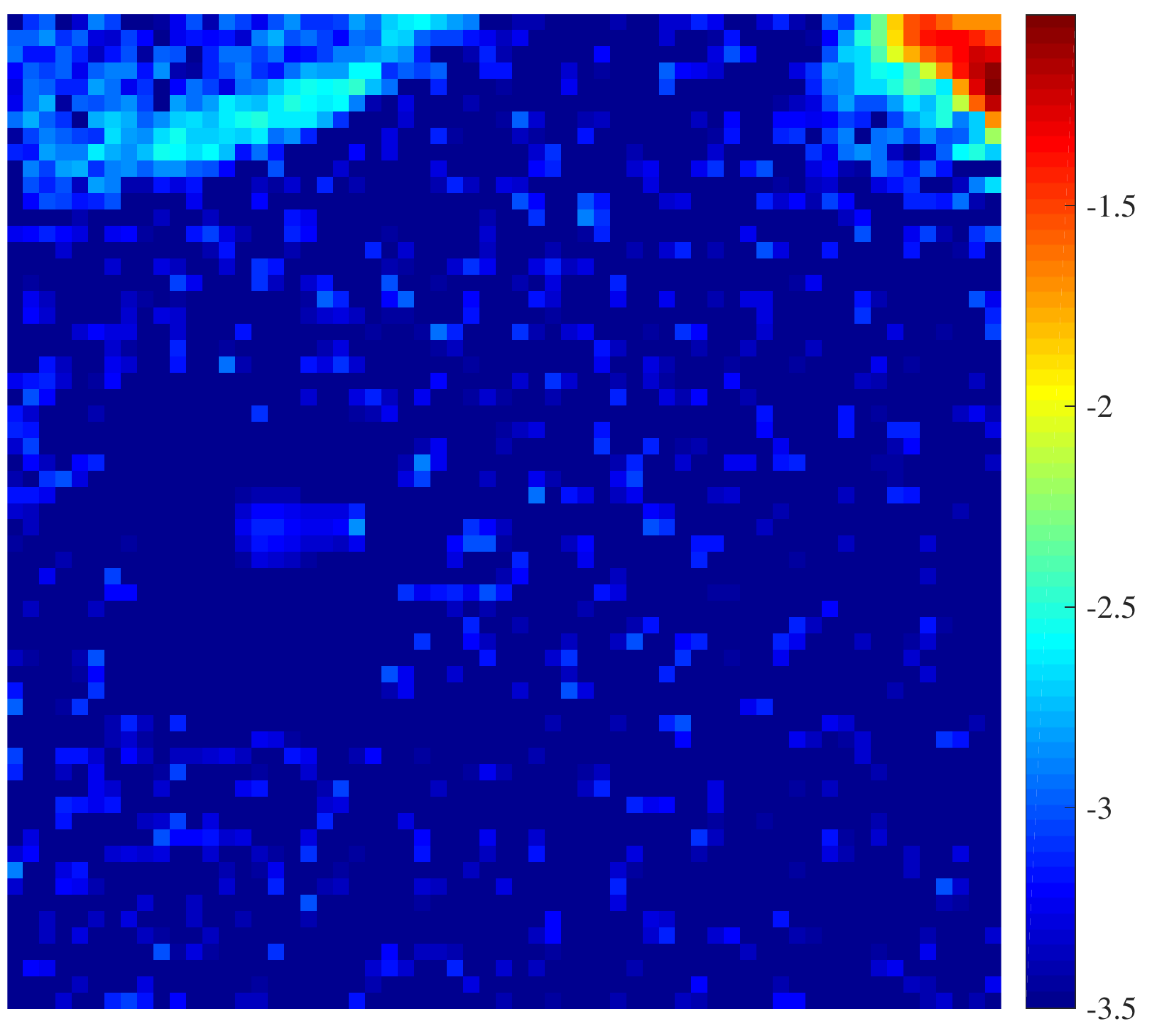}	
&	\includegraphics[height=3.8cm]{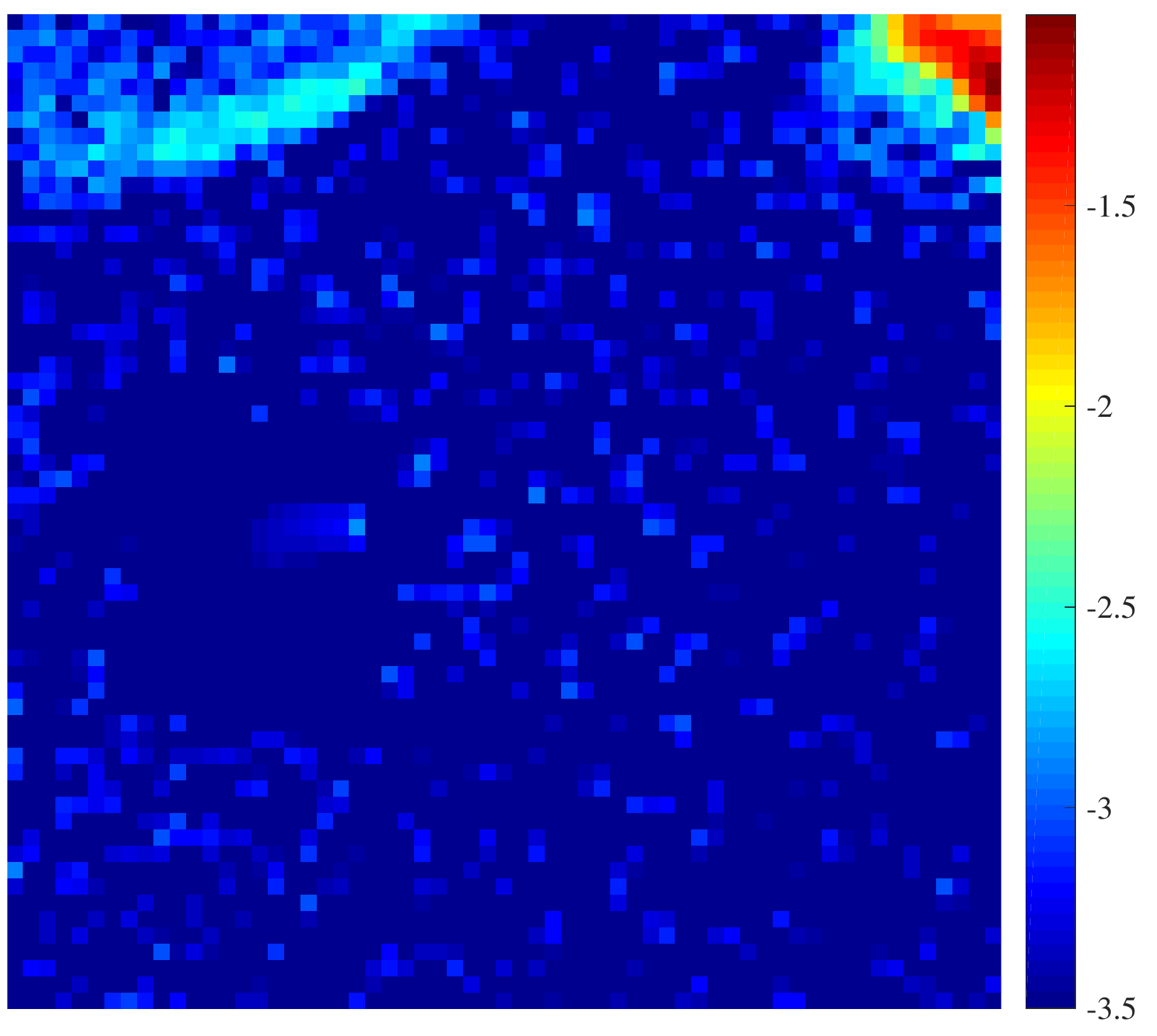}	
\end{tabular}
\end{center}

\vspace*{-0.3cm}

\caption{\label{Fig:Illustration_high}\small
Illustration of the proposed uncertainty quantification method, in the context of RI imaging. Simulation results obtained considering random Gaussian Fourier samplings, $N = 256 \times 256$, $M = N/2$ and $\sigma^2 =0.01$. 
From left to right, the top row shows the MAP estimate $x^\dagger$, and the two resulting images from the POCS algorithm, $x^\ddagger_{\widetilde{\Cc}_\alpha}$	and	$x^\ddagger_{\Sc}$. 
The Bayesian uncertainty quantification is performed on the compact source highlighted in red on the top row MAP estimate $x^\dagger$. 
The bottom row shows images zoomed in the corresponding area of interest. 
For this example, we have $\rho_\alpha \approx 0\%$. 
}
\end{figure}
\begin{figure}[h!]
\begin{center}
\begin{tabular}{@{}c@{}c@{}c@{}}
	\includegraphics[height=3.8cm]{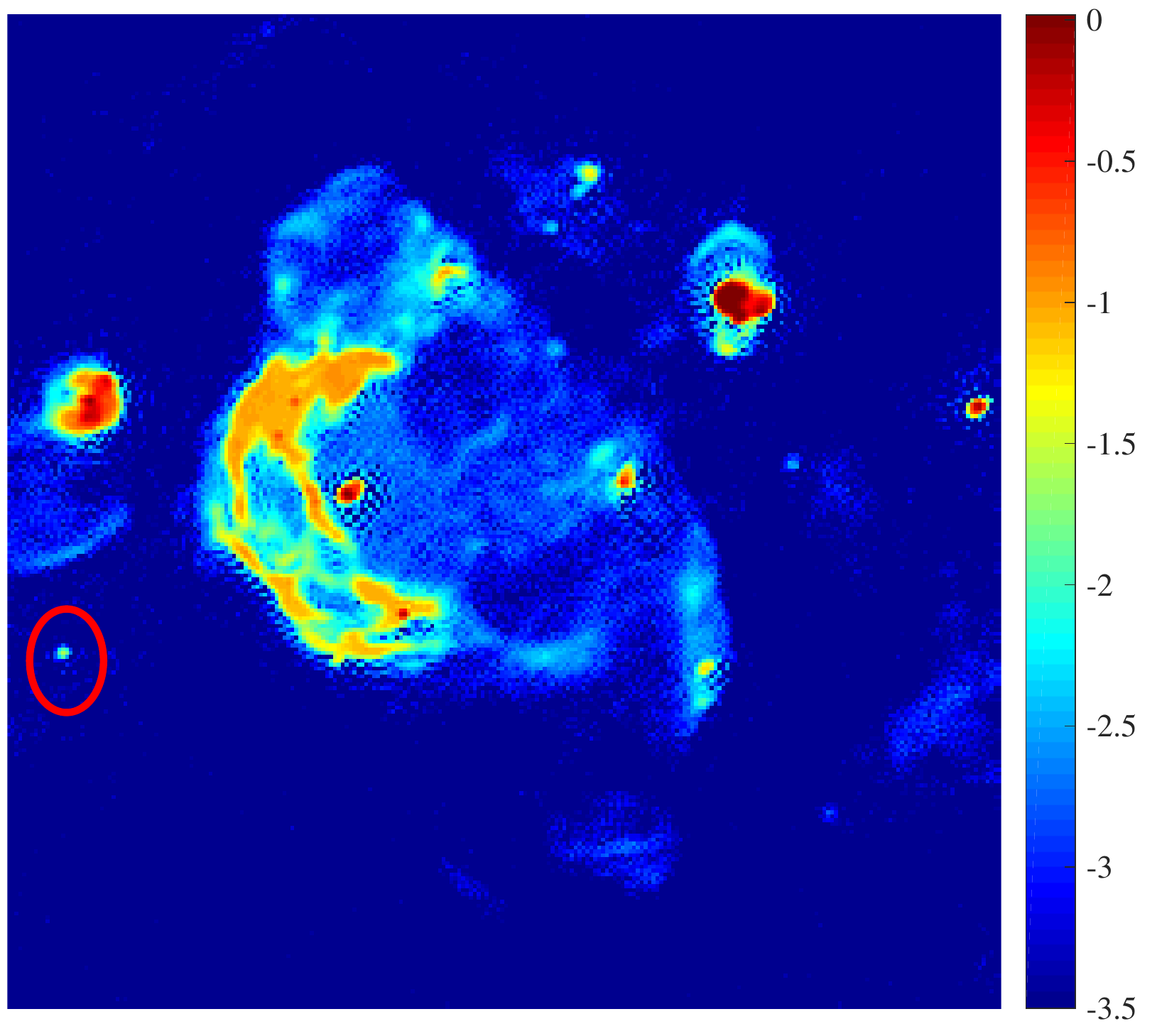}
&	\includegraphics[height=3.8cm]{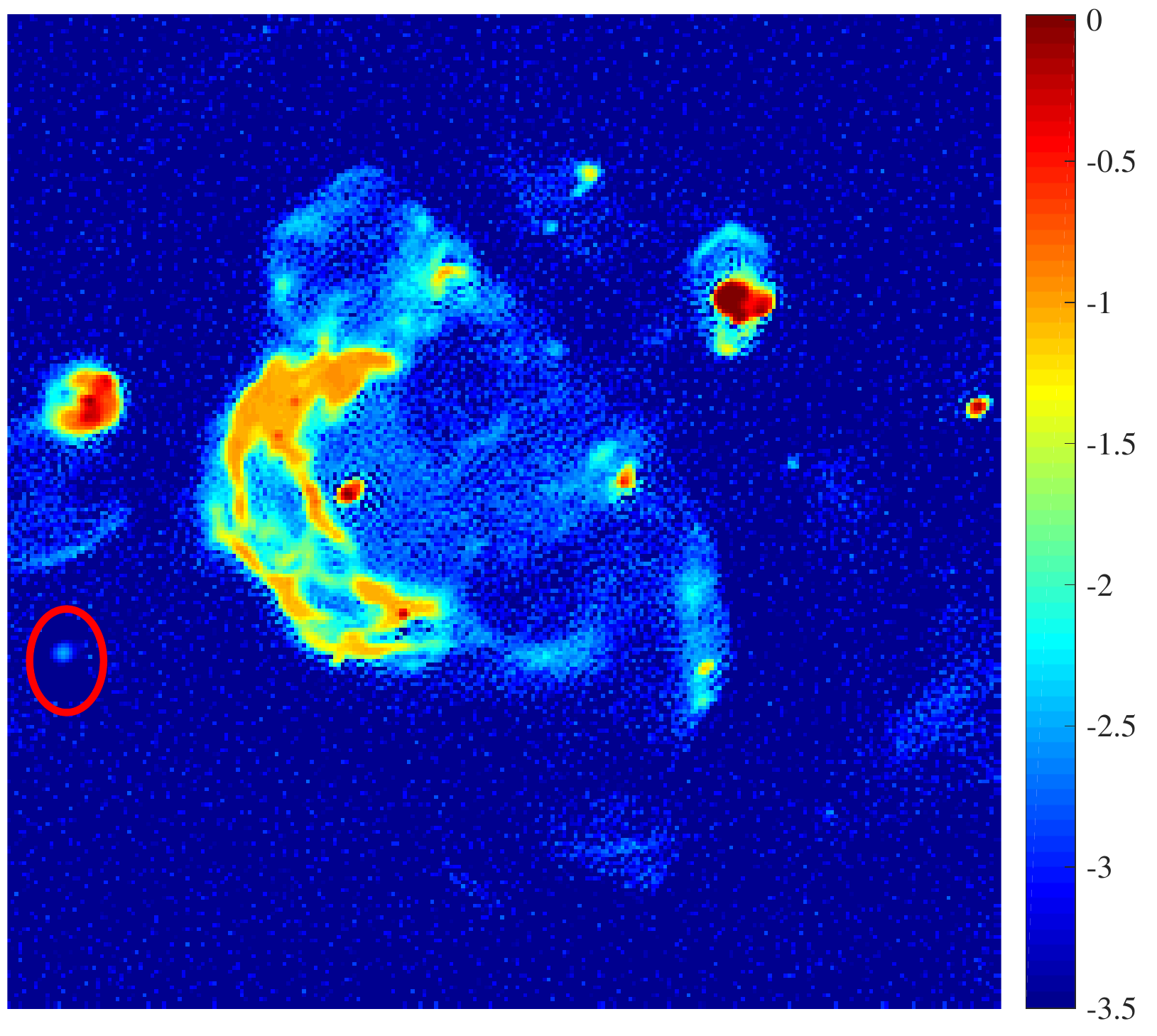}
&	\includegraphics[height=3.8cm]{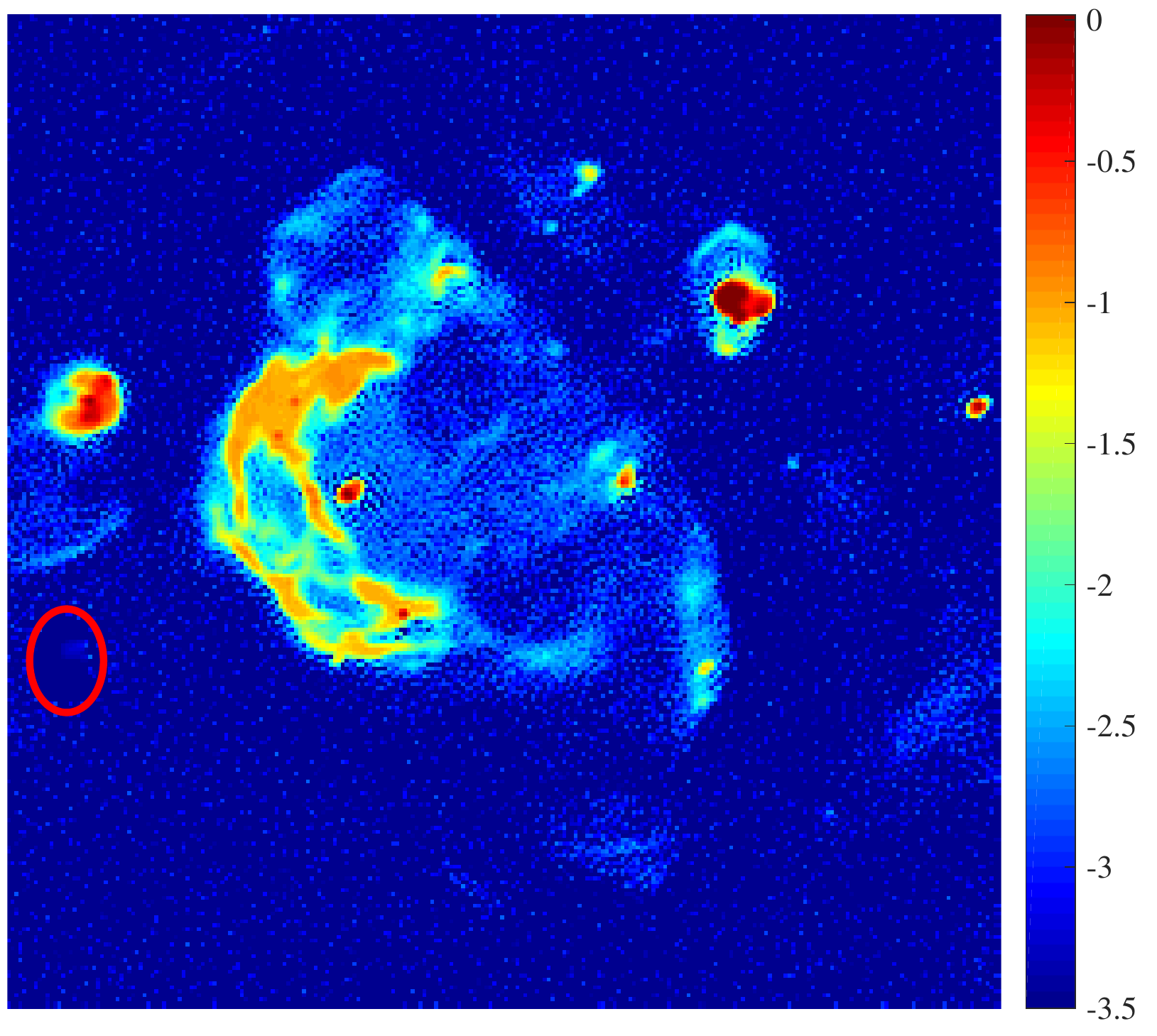}	\\[-0.2cm]
	$x^\dagger$	&	$x^\ddagger_{\widetilde{\Cc}_\alpha}$	&	$x^\ddagger_{\Sc}$	\\[0.1cm]
	\includegraphics[height=3.8cm]{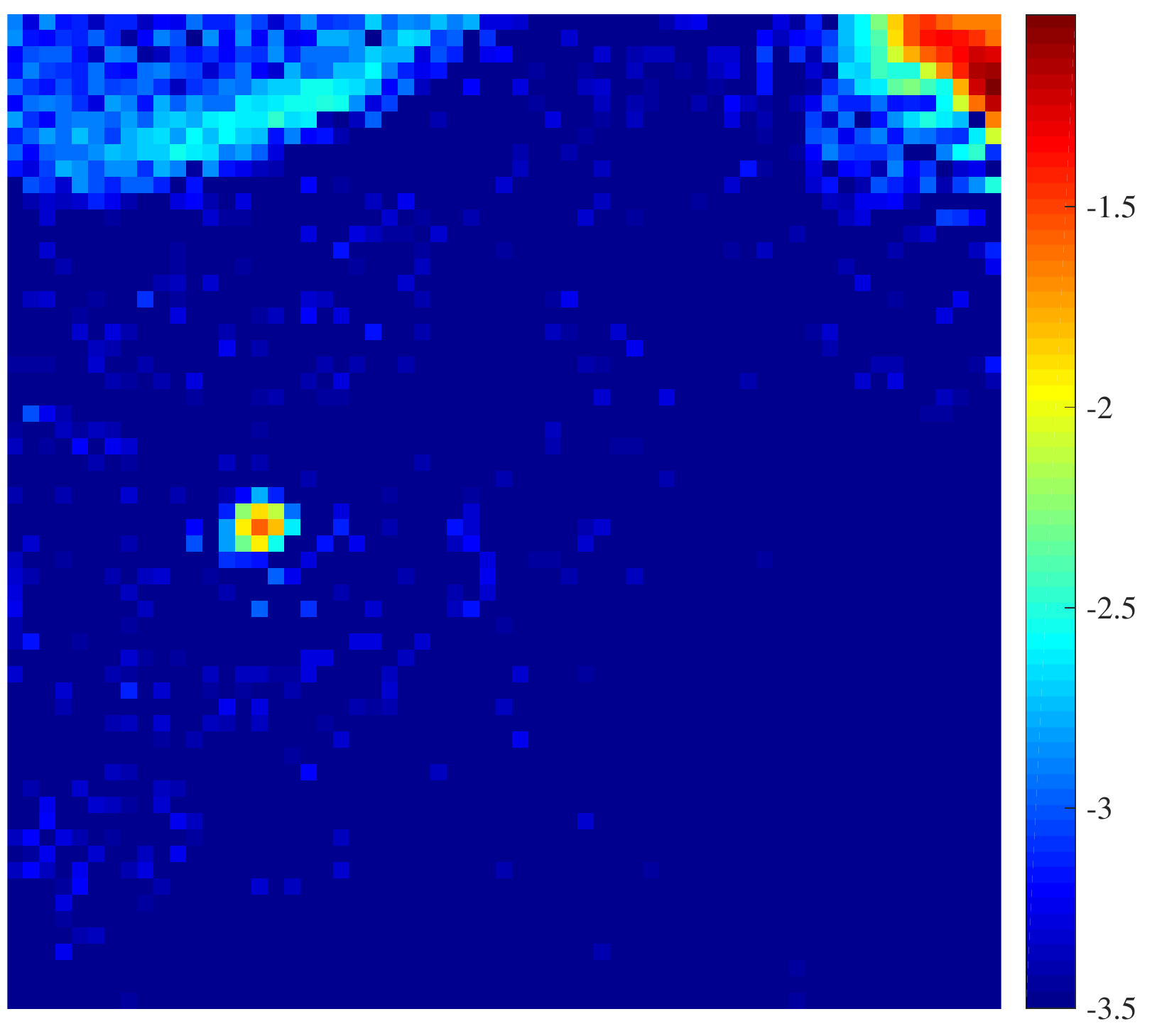}
&	\includegraphics[height=3.8cm]{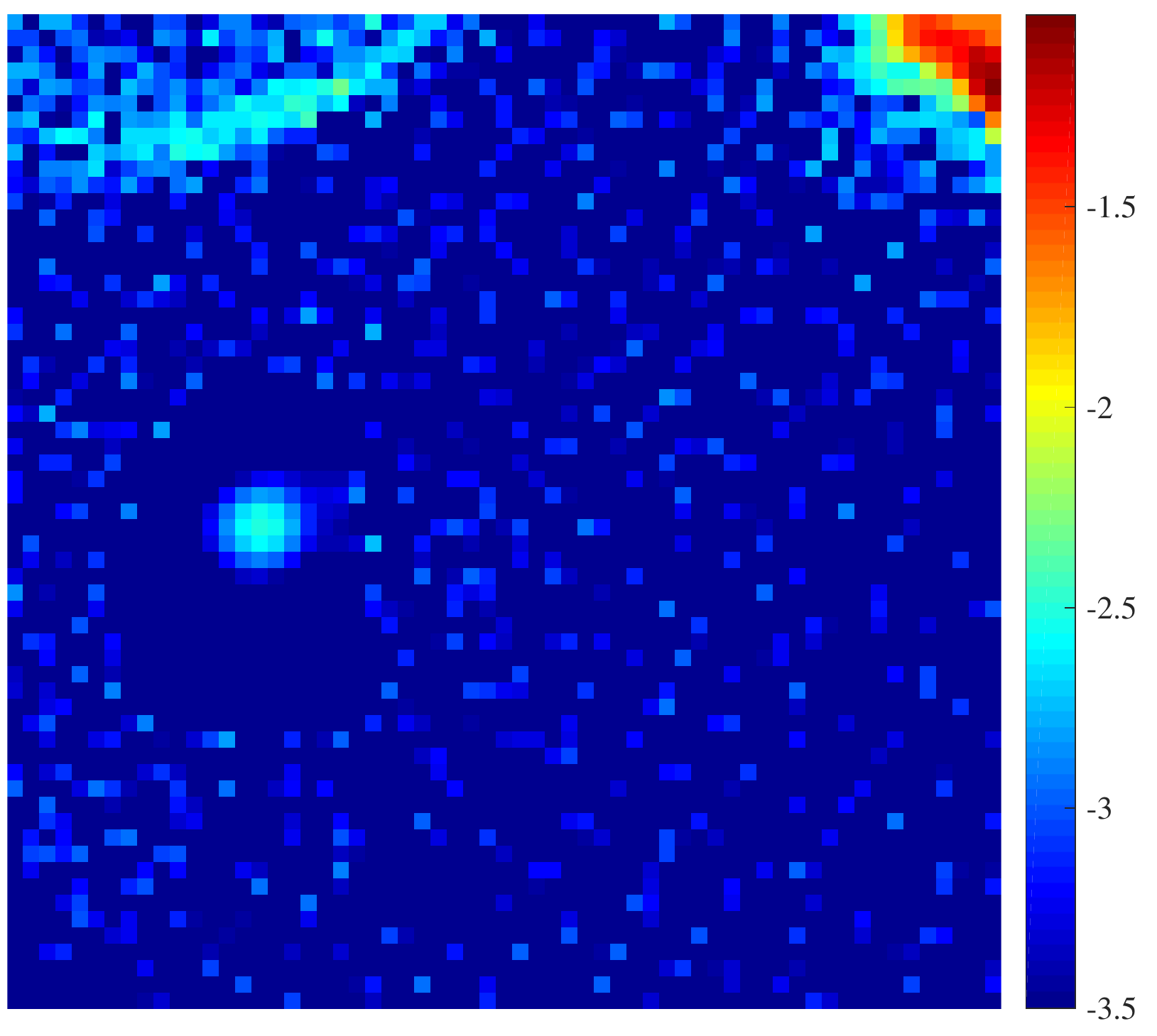}	
&	\includegraphics[height=3.8cm]{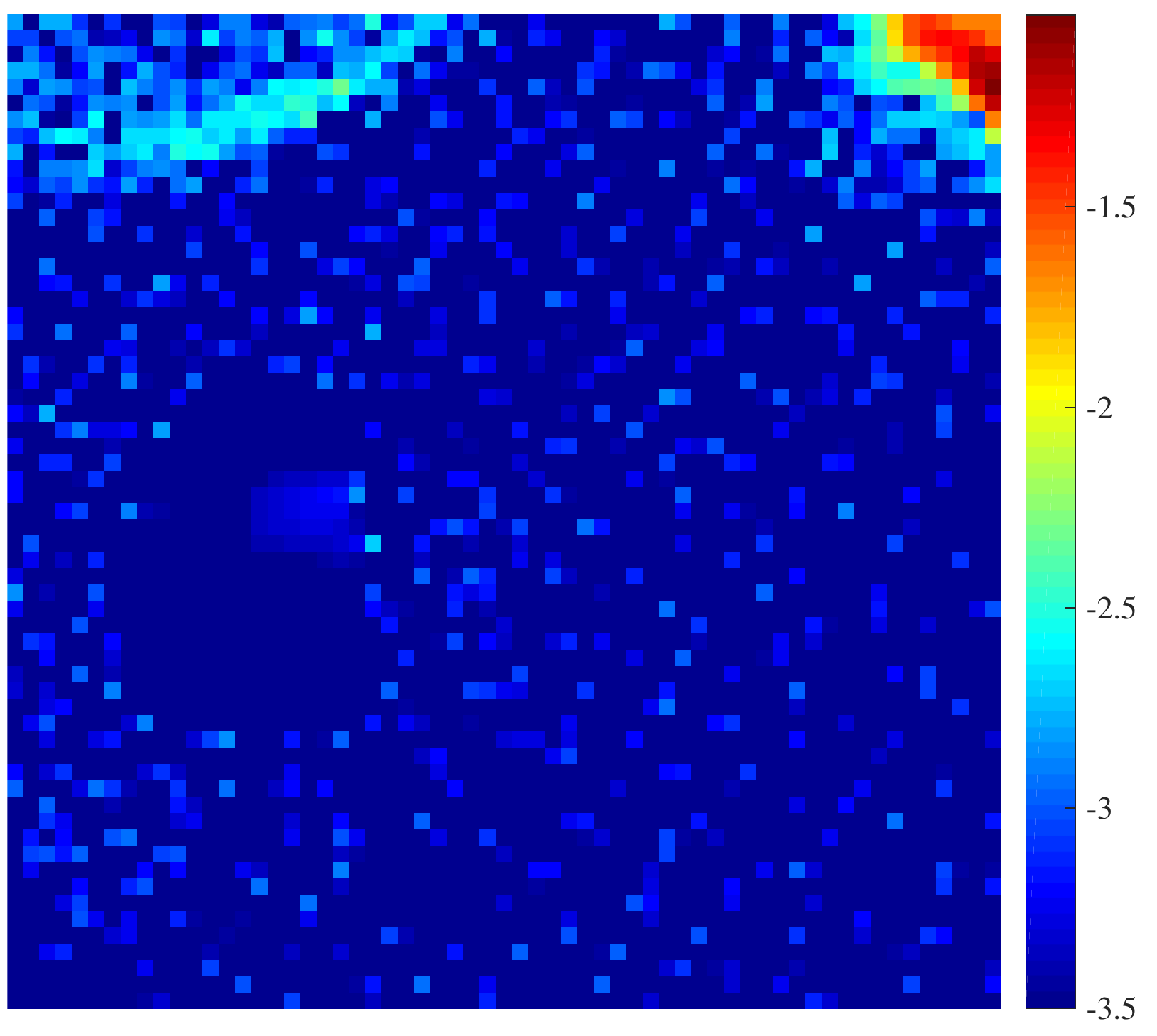}	
\end{tabular}
\end{center}

\vspace*{-0.3cm}

\caption{\label{Fig:Illustration_low}\small
Illustration of the proposed uncertainty quantification method, in the context of RI imaging. Simulation results obtained considering random Gaussian Fourier samplings, $N = 256 \times 256$, $M = N$ and $\sigma^2=0.01$. 
From left to right, the top row shows the MAP estimate $x^\dagger$, and the two resulting images from the POCS algorithm, $x^\ddagger_{\widetilde{\Cc}_\alpha}$	and	$x^\ddagger_{\Sc}$. 
The Bayesian uncertainty quantification is performed on the compact source highlighted in red on the top row MAP estimate $x^\dagger$. 
The bottom row shows images zoomed in the corresponding area of interest.  
For this example, we have $\rho_\alpha = 18.76\%$. 
}
\end{figure}

\subsection{Implementation details}
\label{sec:part_ex}

We now discuss implementation strategies for the proposed methodology. In particular, Algorithm~\ref{algo:POCS_gen} requires computing the projections onto $\widetilde{\Cc}_\alpha$ and $\Sc$, which may need to be sub-iterative depending on the structure of these sets. While there are several methods to compute the projection onto convex sets, here we choose to use primal-dual approaches (see e.g. \cite{beck2009fast, boyd2011distributed, Combettes_Vu_2011, Combettes_P_2012_j-svva_pri_dsa, condat2013primal, komodakis2015playing, vu2013splitting}).

\subsubsection{Projection onto $\widetilde{\Cc}_\alpha$}
\label{Sssec:proj_HPD}
We focus on the case where $g_2$ is the hybrid regularization given in \eqref{def:reg_gen}, where $f$ and $\Cc$ are chosen such that the projections onto the sets $\operatorname{lev}_{\le \beta}(f) = \{u \in \R^N \mid f(x) \le  \beta \}$, for any $\beta>0$, and $\Cc$ have closed form expressions. 
Accordingly, the MAP estimator is given by
\begin{equation}	\label{eq:pb_min_MAP_glob}
x^\dagger \in \Argmind{x \in \R^N} \lambda f( \Psi x ) + \iota_{\Cc}(x) + \iota_{\Bc_2(y,\varepsilon)}(\Phi x).
\end{equation}
In this case, the approximated confidence region is given by
\begin{equation}
\widetilde{\Cc}_\alpha
= 	\Big\{ x \in \Cc \, \mid \,
	\Phi x \in \Bc_2(y,\varepsilon) 
	\quad	\text{and}	\quad
	\lambda f( \Psi x ) \le \widetilde{\eta}_\alpha 
	\Big\},	\label{def:HPD:glob_balls}
\end{equation}
where $\widetilde{\eta}_\alpha =  \lambda f( \Psi x^\dagger ) + N(\tau_\alpha +1)$. 
In this particular case, at iteration $k \in \N$ in Algorithm~\ref{algo:POCS_gen}, the projection step~\ref{algo:POCS_gen:projC} onto the set $\widetilde{\Cc}_\alpha$  reads
\begin{equation}	\label{eq:proj:HPD}
x^{(k+\frac12)} 
= \proj_{\widetilde{\Cc}_\alpha}(x^{(k)})
= \argmind{x \in \R^N} F(x, x^{(k)}),
\end{equation}
with, for every $\big( x, x^{(k)} \big) \in \R^N \times \R^N$,
\begin{equation}	\label{eq:proj:HPDdef}
 F(x, x^{(k)}) = \iota_{\Bc_2(y,\varepsilon)}(\Phi x) + \iota_{\operatorname{lev}_{\le \widetilde{\eta}_\alpha/ \lambda}(f)}(\Psi x) + \iota_{[0, +\infty[^N}(x) + \frac12 \| x - x^{(k)} \|^2.
\end{equation}
The minimization problem described in~\eqref{eq:proj:HPD}-\eqref{eq:proj:HPDdef} involves sophisticated constraints with linear operators. These constraints can be handled efficiently using primal-dual methods such as those developed in \cite{Alotaibi_A_2014_Solving_ccm, Combettes_Vu_2011, Combettes_P_2012_j-svva_pri_dsa, condat2013primal, vu2013splitting}. For an overview on primal-dual approaches, we refer the reader to \cite{komodakis2015playing}. In particular, we propose to solve problem \eqref{eq:proj:HPD}-\eqref{eq:proj:HPDdef} using the primal-dual forward-backward algorithm developed in~\cite{condat2013primal, vu2013splitting} given in Algorithm~\ref{algo:PD_HPD} below.

\begin{algorithm}[ht!]
	\caption{Primal-dual forward-backward algorithm to solve \eqref{eq:proj:HPD}-\eqref{eq:proj:HPDdef}.}
	\begin{algorithmic}[1]
        \State 	
        \textbf{Initialization:} 
        Let $u^{(0)} \in \R^N$,  
        $v_1^{(0)} \in \R^N$, and $v_2^{(0)} \in \C^M$. 
        Let $\sigma>0$ and $\gamma>0$ such that $\displaystyle \sigma \big( \frac12 + \gamma ( \| \Psi \|^2 + \| \Phi \|^2 ) \big) <1$.
        \newline
        \State	\vspace{0.5em}
        \textbf{For} $i = 0, 1, \ldots$
		\State	\vspace{0.3em}	\label{algo:PD_HPD:primal}
		$\displaystyle \quad\quad
		u^{(i+1)} = 
		\proj_{\Cc} \Big( u^{(i)} - \sigma \big(  u^{(i)} - x^{(k)} + \Psi^\dagger v_1^{(i)} + \Phi^\dagger v_2^{(i)} \big) \Big)$
		\State	\vspace{0.3em}	\label{algo:PD_HPD:dual11}
		$ \displaystyle \quad\quad
		\widetilde{v}_1^{(i)} = 
		v_1^{(i)} + \gamma \Psi (2 u^{(i+1)} - u^{(i)})	$
		\State	\vspace{0.3em}	\label{algo:PD_HPD:dual12}
		$ \displaystyle \quad\quad
		v_1^{(i+1)} = 
		\widetilde{v}_1^{(i)} - \gamma \proj_{\operatorname{lev}_{\le \widetilde{\eta}_\alpha/ \lambda}(f)} \big( \gamma^{-1} \widetilde{v}_1^{(i)} \big)	$
		\State	\vspace{0.3em}	\label{algo:PD_HPD:dual21}
		$ \displaystyle \quad\quad
		\widetilde{v}_2^{(i)} = 
		v_2^{(i)} + \gamma \Phi (2 u^{(i+1)} - u^{(i)})	$
		\State	\vspace{0.3em}	\label{algo:PD_HPD:dual22}
		$ \displaystyle \quad\quad
		v_2^{(i+1)} = 
		\widetilde{v}_2^{(i)} - \gamma \proj_{\Bc_2(y, \varepsilon)} \big( \gamma^{-1} \widetilde{v}_2^{(i)} \big)	$
        \State	\vspace{0.5em}
        \textbf{end for}
	\end{algorithmic}
    \label{algo:PD_HPD}
\end{algorithm}
%

Since $\widetilde{\Cc}_\alpha$ is convex, for every $x^{(k)} \in \R^N$, the function $F( \cdot, x^{(k)})$ is strictly convex. Therefore, according to \cite{condat2013primal, vu2013splitting}, the sequence $(u^{(i)})_{i \in \N}$ generated by Algorithm~\ref{algo:PD_HPD} is ensured to converge to the unique minimizer of $F( \cdot, x^{(k)})$ (i.e. the point of $\widetilde{\Cc}_\alpha$ the closest to $x^{(k)}$).

As particular cases for the function $f$, we can mention the $\ell_1$-norm, the $\ell_2$-norm, the $\ell_\infty$-norm, the $\ell_{1,2}$-norm, and the negative logarithm function. The projections onto the lower level sets of these functions can be found on the online proximity operator repository \cite{Chierchia_prox_rep}.

\subsubsection{Projection onto $\Sc$: Definition~\ref{example:structure}}
\label{Sssec:proj_S_struct}

Consider the set $\Sc$ defined by \eqref{def:S:structure}. 
For every iteration $k\in \N$, the projection step~\ref{algo:POCS_gen:projS} onto $\Sc$ in Algorithm~\ref{algo:POCS_gen} is given by
\begin{align}
x^{(k+1)} 
&=	\proj_{\Sc}(x^{(k+\frac12)} ) \nonumber	\\
&= 	\argmind{x \in \R^N} \iota_{[0,+\infty[^N}(x) 
		+ \iota_{\Sc_2}(x) 
		+ \iota_{\Bc_2(b,\theta)}\big(\Mc(x)\big) 
		+ \frac12 \| x - x^{(k+\frac12)} \|^2	 .	
	\label{eq:proj_D_ex2}
\end{align}
Since 
\begin{equation}
x \in \Sc_2
\quad \Leftrightarrow \quad
	\Mc(x) - \Lc \big(\Mc^c(x)\big)	\in [\underline{\tau}, \overline{\tau}]^{N_{\Mc}}	
\quad \Leftrightarrow \quad
	\overline{\Lc} (x) \in [\underline{\tau}, \overline{\tau}]^{N_{\Mc}}	,		
\end{equation}
where $\overline{\Lc } = \Mc - \Lc \circ \Mc^c$, we have
\begin{equation}
x^{(k+1)} 
=	\argmind{x \in \R^N} \iota_{[0,+\infty[^N}(x) 
		+ \iota_{[\underline{\tau}, \overline{\tau}]^{N_{\Mc}}} \big( \overline{\Lc} (x) \big)
		+ \iota_{\Bc_2(b,\theta)}\big(\Mc(x)\big) 
		+ \frac12 \| x - x^{(k+\frac12)} \|^2	 .
\end{equation}
This problem does not have a closed form solution, and hence needs to be solved by computing sub-iterations. Again, here we use a primal-dual forward-backward algorithm \cite{condat2013primal, vu2013splitting}. The resulting method is described in Algorithm~\ref{algo:PD_projS_struct}.

\begin{algorithm}[htbp!]
	\caption{Primal-dual forward-backward algorithm to solve \eqref{eq:proj_D_ex2}.}
	\begin{algorithmic}[1]
        \State 	
        \textbf{Initialization:} 
        Let $a^{(0)} \in [0,+\infty[^{N}$, $p_1^{(0)} \in \R^{N_{\Mc}}$, and $p_2^{(0)} \in \R^{N_{\Mc}}$. 
        Let $\kappa>0$ and $\nu>0$ such that $\displaystyle \kappa \big( \frac12 + \nu( \| \Lc \|^2 + 1 ) \big) <1$.
        \newline
        \State	\vspace{0.5em}
        \textbf{For} $j = 0, 1, \ldots$
		\State	\vspace{0.3em}	\label{algo:PD_projS_struct:primal1}
		$\displaystyle \quad\quad
		\widetilde{a}^{(j)} = 
		\proj_{[0,+\infty[^{N}}
		\Big(	a^{(j)} 
				- \kappa \big( a^{(j)} - x^{(k+\frac12)} 
				-  \overline{\Lc}^\dagger p_1^{(j)} + \Mc^\dagger(p_2^{(j)}) \big) \Big) $
		\State	\vspace{0.3em}	\label{algo:PD_projS_struct:dual1}
		$ \displaystyle \quad\quad
		\widetilde{p}_1^{(j)} = 
		p_1^{(j)} + \nu \, \overline{\Lc} (2 a^{(j+1)} - a^{(j)})	$
		\State	\vspace{0.3em}	\label{algo:PD_projS_struct:dual2}
		$ \displaystyle \quad\quad
		p_1^{(j+1)} = 
		\widetilde{p}_1^{(j)} - \nu \, \proj_{[\underline{\tau}, \overline{\tau}]^{N_{\Mc}}} \big( \nu^{-1} \widetilde{p}_1^{(j)} \big)	$
		\State	\vspace{0.3em}	\label{algo:PD_projS_struct:dual3}
		$ \displaystyle \quad\quad
		\widetilde{p}_2^{(j)} = 
		p_2^{(j)} + \nu \, \Mc (2 a^{(j+1)} - a^{(j)})	$
		\State	\vspace{0.3em}	\label{algo:PD_projS_struct:dual4}
		$ \displaystyle \quad\quad
		p_2^{(j+1)} = 
		\widetilde{p}_2^{(j)} - \nu \, \proj_{\Bc_2(b, \theta)} \big( \nu^{-1} \widetilde{p}_2^{(j)} \big)	$
        \State	\vspace{0.5em}
        \textbf{end for}
	\end{algorithmic}
    \label{algo:PD_projS_struct}
\end{algorithm}
%

Note that, for every $x^{(k+\frac12)} \in \R^N$, problem~\eqref{eq:proj_D_ex2} is strictly convex. Therefore, according to \cite{condat2013primal, vu2013splitting}, the sequence $(a^{(j)})_{j \in \N}$ generated by Algorithm~\ref{algo:PD_projS_struct} is ensured to converge to the unique solution to problem~\eqref{eq:proj_D_ex2} (i.e. the point of $\Sc$ the closest to $x^{(k+\frac12)}$).

\subsubsection{Projection onto $\Sc$: Definition~\ref{example:background}}
\label{Sssec:proj_S_back}

Consider the set $\Sc$ defined by \eqref{def:S:background}. In this case the projection onto $\Sc$ has an explicit formula. In particular, at every iteration $k\in \N$, the projection step~\ref{algo:POCS_gen:projS} onto $\Sc$ in Algorithm~\ref{algo:POCS_gen} is given by
\begin{align}
x^{(k+1)} 
&= \proj_{\Sc}(x^{(k+\frac12)} \nonumber	\\
&= \argmind{x \in \R^N} \iota_{[0,+\infty[^N}(x) + \iota_{\Sc_2}(x) + \frac12 \| x - x^{(k+\frac12)} \|^2 .		
\end{align}
Then, we have 
\begin{equation}
\begin{cases}
\Mc(x^{(k+1)} ) = 0,	\\
\Mc^c(x^{(k+1)} ) = \min \Big\{ \overline{\tau}, \max \big\{ \underline{\tau}, \Mc^c(x^{(k+1)} ) \big\} \Big\}.
\end{cases}
\end{equation}

\subsection{Scalable and approximated alternating projection methods}
\label{Ssec:discuss:fast}

The use of the POCS method given in Algorithm~\ref{algo:POCS_gen} to solve problem~\eqref{pb:feas_gen} is important to illustrate the proposed uncertainty quantification approach. However, it is worth mentioning that the convergence of this algorithm can be slow in practice and the convergence results (see Theorem~\ref{thm:POCS:cvg}) hold only if the projections are computed exactly. 

There are multiple (possibly accelerated) methods in the literature to solve convex feasibility problems such as~\eqref{pb:feas_gen} (see \cite{Bauschke_Borwein_1996, Deutsch_1992_MAP, Escalande_book_2011} for details). 
However, our method not only requires to solve~\eqref{pb:feas_gen}, but also necessitate to determine if this problem is feasible or not, i.e. if the intersection between $\widetilde{\Cc}_\alpha $ and $ \Sc$ is empty or not. Due to that particular subtlety, accelerated POCS methods cannot be used in our approach, since they all assume that the problem of interest must be feasible.

Because $\widetilde{\Cc}_\alpha \cap \Sc = \emp$ holds if and only if $\dist \big( \widetilde{\Cc}_\alpha, \Sc \big) >0$, we could also formulate \eqref{pb:feas_emp} as follows:
\begin{equation}	\label{pb:feas_gen_dist}
\text{find }
\big( x^\ddagger_{\widetilde{\Cc}_\alpha}, x^\ddagger_{\Sc} \big) 
= 
\argmind{\big( x_{\widetilde{\Cc}_\alpha}, x_{\Sc} \big) \in \R^{2N}} \| x_{\widetilde{\Cc}_\alpha} - x_{\Sc} \|^2
\text{ s.t. }
\big( x_{\widetilde{\Cc}_\alpha}, x_{\Sc} \big) \in \widetilde{\Cc}_\alpha \times \Sc.
\end{equation}
This problem is strictly convex on $\big( x_{\widetilde{\Cc}_\alpha}, x_{\Sc} \big)$ and can be solved using recent convex optimization techniques, e.g. the forward-backward (FB) algorithm \cite{Tseng_P_2000_j-siam-control-optim_Modified_fbs, combettes2005signal, Attouch_Bolte_2011} or its accelerated versions (e.g. \cite{beck2009fast, Chouzenoux13, Pock2013}). Applied to problem~\eqref{pb:feas_gen_dist}, the classical FB method can be seen as an alternating projection approach, and takes the form of Algorithm~\ref{algo:FB_dist}.
\begin{algorithm}[h!]
	\caption{FB algorithm to solve problem~\eqref{pb:feas_gen_dist}.}
	\begin{algorithmic}[1]
        \State 	
        \textbf{Initialization:} 
        Let $x_{\widetilde{\Cc}_\alpha}^{(0)} \in \widetilde{\Cc}_\alpha$ and $x_{\Sc}^{(0)} \in \Sc$. 
        Let $\gamma \in ]0, 1[$.
        \newline
        \State	\vspace{0.5em}
        \textbf{For} $k = 0, 1, \ldots$
		\State	\vspace{0.3em}	\label{algo:FB_dist:projC}
		$\displaystyle \quad\quad
		x_{\widetilde{\Cc}_\alpha}^{(k+1)} = 
		\proj_{\widetilde{\Cc}_\alpha} \Big( (1-\gamma) x_{\widetilde{\Cc}_\alpha}^{(k)} + \gamma  x^{(k)}_{\Sc} \Big)$
		\State	\vspace{0.3em}	\label{algo:FB_dist:projS}
		$ \displaystyle \quad\quad
		x_{\Sc}^{(k+1)} = 
		\proj_{\Sc} \Big( (1-\gamma) x^{(k)}_{\Sc} + \gamma  x^{(k)}_{\widetilde{\Cc}_\alpha} \Big)	$
        \State	\vspace{0.5em}
        \textbf{end for}
	\end{algorithmic}
    \label{algo:FB_dist}
\end{algorithm}
The sequence $\Big( x_{\widetilde{\Cc}_\alpha}^{(k)}, x^{(k)}_{\Sc} \Big)_{k \in \N}$ generated by Algorithm~\ref{algo:FB_dist} converges to the unique solution to problem~\eqref{pb:feas_gen_dist}. 
The convergence of this algorithm is also guaranteed when projections are computed approximately (with additive errors \cite{combettes2005signal} or relative errors \cite{Attouch_Bolte_2011}). 
Notice that the POCS method given in Algorithm~\ref{algo:POCS_gen} is recovered in the limit case when $\gamma = 1$ in Algorithm~\ref{algo:FB_dist}, which actually provides some notion of robustness of Algorithm~\ref{algo:POCS_gen} to approximation errors.

As for the POCS method given in Algorithm~\ref{algo:POCS_gen}, the projections onto the sets $\widetilde{\Cc}_\alpha$ and $\Sc$, appearing in steps~\ref{algo:FB_dist:projC} and \ref{algo:FB_dist:projS} respectively, may require sub-iterations (see Section~\ref{sec:part_ex} for implementation details). To avoid these sub-iterations, it is possible to use more advanced techniques such as the primal-dual algorithm used in Section~\ref{sec:part_ex} (see for example \cite{Combettes_P_2012_j-svva_pri_dsa, condat2013primal, komodakis2015playing, pesquet2014class, vu2013splitting}).

\subsection{BUQO in a nutshell}

In this section, we summarize the principle of the proposed method. BUQO for computational imaging consists of four main steps, described below:
\begin{enumerate}
\item
Compute the MAP estimate $x^\dagger$ by solving problem~\eqref{hpd2}.
\begin{itemize}
\item
Use $x^\dagger$ to deduce $\widetilde{\Cc}_\alpha$ using equation~\eqref{eq:hpdapp1}.
\end{itemize}

\item
Identify the structure of interest in $x^\dagger$.
\begin{itemize}
\item
Define the hypothesis test, by postulating the null hypothesis $H_0$, i.e. the structure of interest is absent in the true image (see Section~\ref{Ssec:Bayes_quant} for the details).
\item
Define the associated set $\Sc$ (see Section~\ref{Ssec:def_S}): If the structure is spatially localized, use Definition~\ref{def:S:structure}; If the structure corresponds to the background, use Definition~\ref{def:S:background}.
\end{itemize}

\item
Determine if $\Sc \cap \widetilde{\Cc}_\alpha = \emp$ using Algorithm~\ref{algo:POCS_gen}. In this algorithm, 
\begin{itemize}
\item
Step~3 corresponds to the projection onto $\widetilde{\Cc}_\alpha$. The computation of this projection is detailed in Section~\ref{Sssec:proj_HPD}.
\item
Step~4 corresponds to the projection onto $\Sc$. The computation of this projection is detailed in Section~\ref{Sssec:proj_S_struct} for a spatially localized structure, and in Section~\ref{Sssec:proj_S_back} for background removal.
\end{itemize}

\item
Deduce if $H_0$ is rejected using Theorem~\ref{Thm:proba_feas}.
\begin{itemize}
\item
If $\Sc \cap \widetilde{\Cc}_\alpha = \emp$, then $H_0$ is rejected with significance $\alpha$, and the structure of interest is present in the true image with probability $1-\alpha$.
\item
If $\Sc \cap \widetilde{\Cc}_\alpha \neq \emp$, then $H_0$ cannot be rejected, and the presence of the structure of interest in the true image is uncertain.
\end{itemize}
\end{enumerate}

\section{Simulation results}
\label{sec:results}

In this section we apply the proposed uncertainty quantification approach to Fourier imaging applications in radio astronomy (Section~\ref{ssec:RI_astro}) and magnetic resonance in medicine (Section~\ref{ssec:MRmed}). 
We refer the reader to Section~\ref{ssec:Illust} for an illustration example, where a step-by-step explanation is given for the practical application of the proposed uncertainty quantification approach.

Before giving the uncertainty quantification results obtained using the proposed approach, we describe in Section~\ref{Ssec:simul:evalH0} the common simulation settings.

\subsection{Simulation settings}
\label{Ssec:simul:evalH0}

In both the two considered applications, the MAP estimate $x^\dagger$ is obtained from problem~\eqref{eq:pb_min_MAP_glob}, where $\Phi$ is the measurement operator associated with each problem (defined in Sections~\ref{Sssec:RI} and \ref{Sssec:MRI}), $f = \|.\|_1$, and $\Psi$ corresponds to the Daubechies wavelet Db8. 
As far as the additive noise is considered, it is generated as i.i.d. Gaussian noise with variance $\sigma^2$. We recall that our approach assumes no explicit knowledge of the noise distribution, other than the fact that it has bounded energy with bound $\epsilon$. For Gaussian noise, such a bound can be computed analytically based on the fact that $\|w\|$  follows a $\chi^2$ distribution with $2M$ degrees of freedom. Due the the concentration of measure in high dimension, the $\chi^2$ is extremely peaked around its mean value. In practice, to ensure a bound satisfied with high probability, we choose $\varepsilon = \sigma \big( 2M + 2\sqrt{4M} \big)^{1/2} $ corresponding to a value 2 standard deviations above the mean of the $\chi^2$ distribution.

We consider the definition of $\widetilde{\Cc}_\alpha$ given in equation~\eqref{def:HPD:glob_balls}, with $\alpha = 1\%$.
To choose $\lambda>0$, we assume that $\Psi \overline{x}$ follows an i.i.d. Laplace distribution, and we propose to choose the maximum likelihood of $\lambda$ based on the MAP estimate $x^\dagger$, i.e.:
\begin{equation}
\lambda = \dfrac{N}{\| \Psi x^\dagger \|_1}.
\end{equation}

In our simulations, we consider that Algorithm~\ref{algo:POCS_gen} has converged if one of the following stopping criteria is fulfilled:
\begin{equation}	\label{stop_crit:norm}
\begin{cases}
\| x^{(k+1)} - x^{(k)} \| < 10^{-5} \| x^{(k+1)} \| \, ,	\\
\| x^{(k+\frac12)} - x^{(k-\frac12)} \| < 10^{-5} \|x^{(k+\frac12)} \| \, ,
\end{cases}
\end{equation}
or
\begin{equation}	\label{stop_crit:dist}
| \delta^{(k+1)} - \delta^{(k)} | < 10^{-5} \delta^{(k+1)},
\end{equation}
where $\delta^{(k+1)} = \| x^{(k+\frac12)} - x^{(k+1)}\|$. 
In other worlds, the first criterion \eqref{stop_crit:norm} verifies the relative variations of the convergent sequences $\big( x^{(k)} \big)_{k \in \N}$ and $\big( x^{(k+\frac12)} \big)_{k \in \N}$. 
The second criterion \eqref{stop_crit:dist} verifies the relative variations of $\big( \delta^{(k)} \big)_{n \in \N}$ which, according to Theorem~\ref{thm:POCS:cvg}, converges to $\dist(\Sc, \widetilde{\Cc}_\alpha)$.

Note that due to the considered stopping criteria \eqref{stop_crit:norm} and \eqref{stop_crit:dist}, the algorithm cannot reach exactly $\| x^\ddagger_{\Sc} - x^\ddagger_{\widetilde{\Cc}_\alpha} \| = 0$. Consequently, the parameter $\rho_\alpha$ introduced in \eqref{def:rho_alpha} cannot be equal to $0$. 
To take into account this numerical approximation, we consider that when $\rho_\alpha > \eta $, with $\eta\approx 0$, then $H_0$ is rejected with significance $\alpha = 1\%$. For instance, in our simulations, we will choose $\eta = 3\%$.

\subsection{Radio-astronomical imaging}
\label{ssec:RI_astro}

\subsubsection{Problem description}
\label{Sssec:RI}

\begin{figure}[h]
\begin{center}
\begin{tabular}{@{}c@{}c@{}}
	\includegraphics[height=5.5cm]{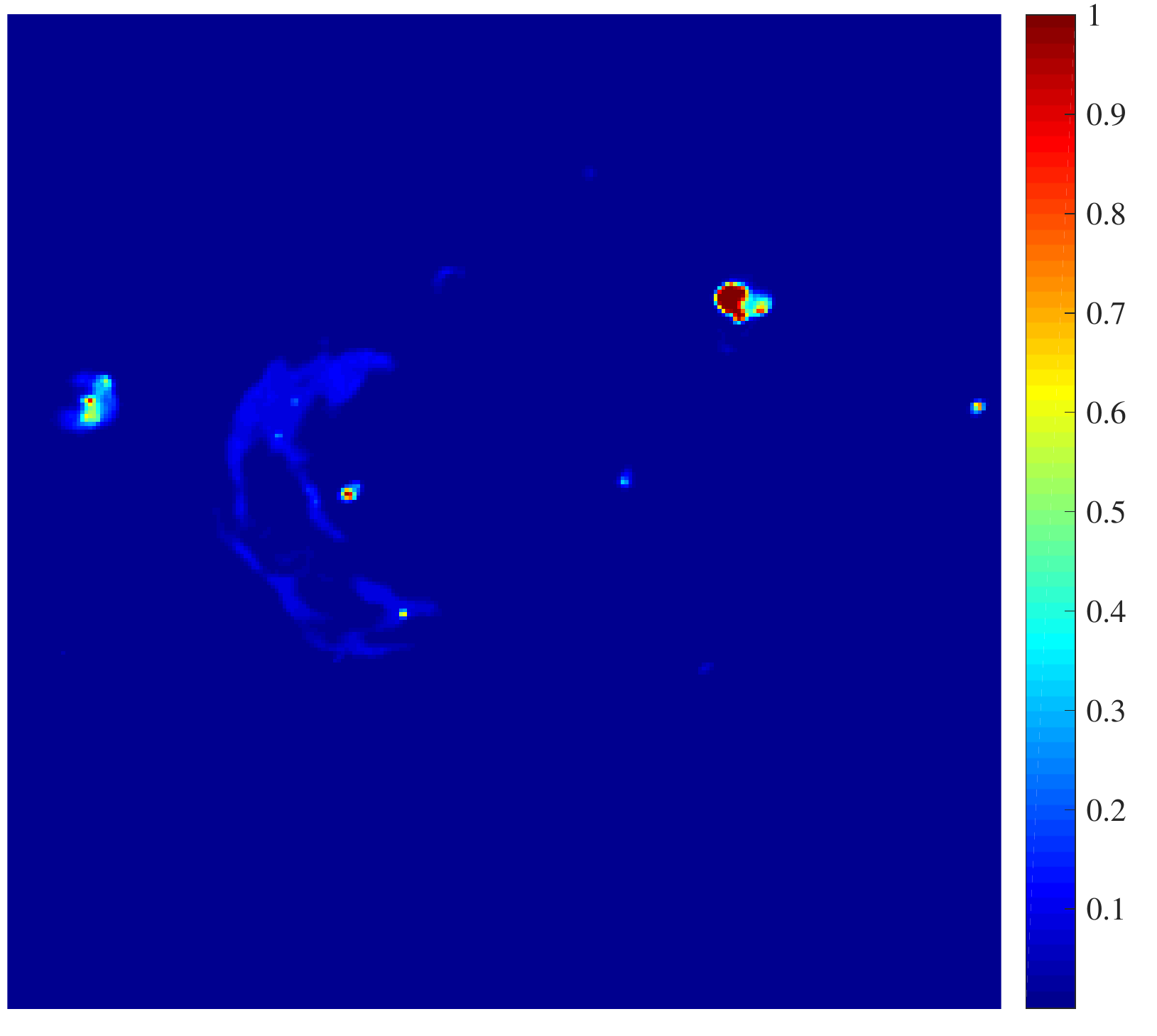}
&	\includegraphics[height=5.5cm]{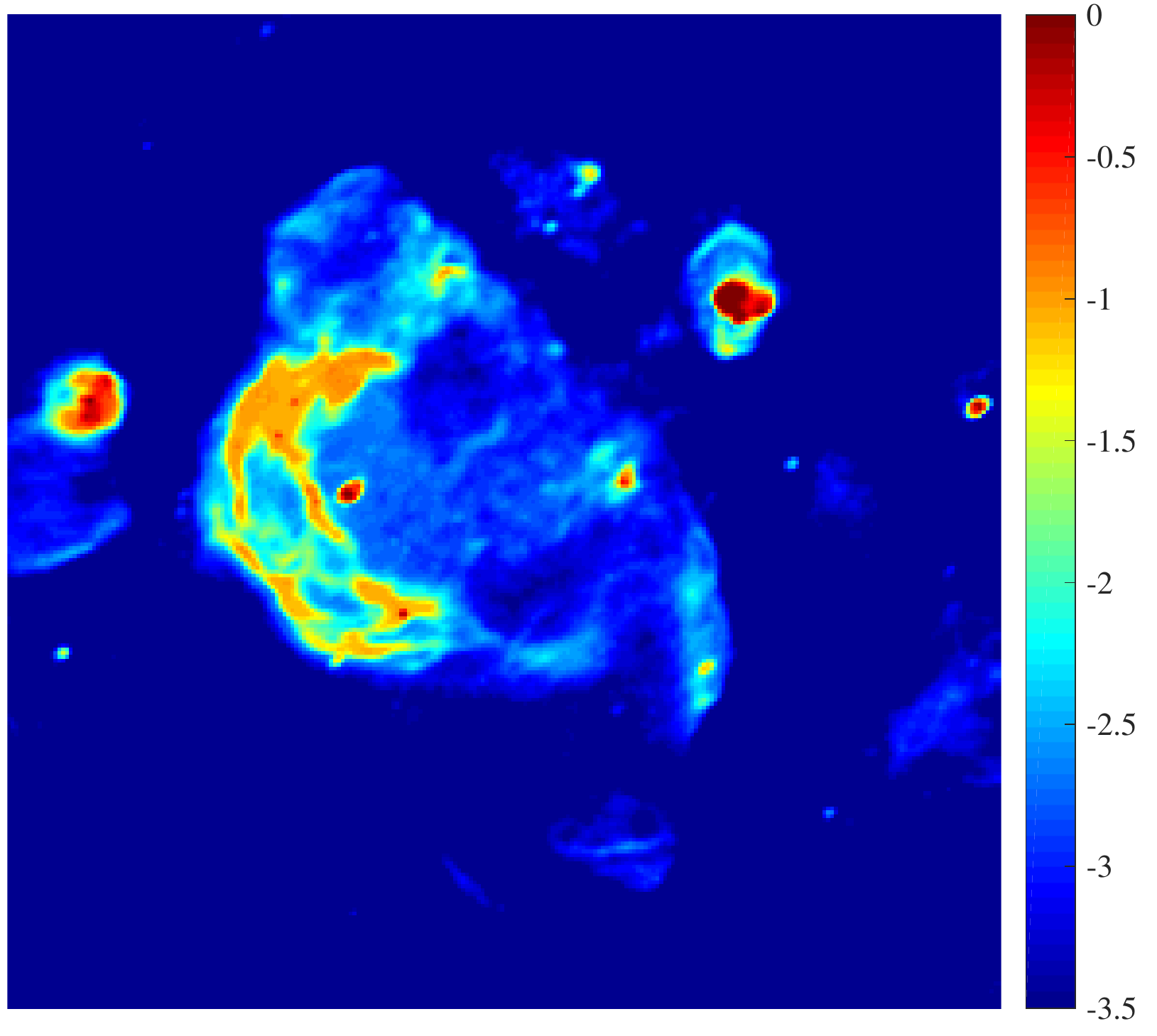}	\\
	(a)	&	(b)	
\end{tabular}\\
\begin{tabular}{@{}c@{}}
	\includegraphics[height=5.6cm]{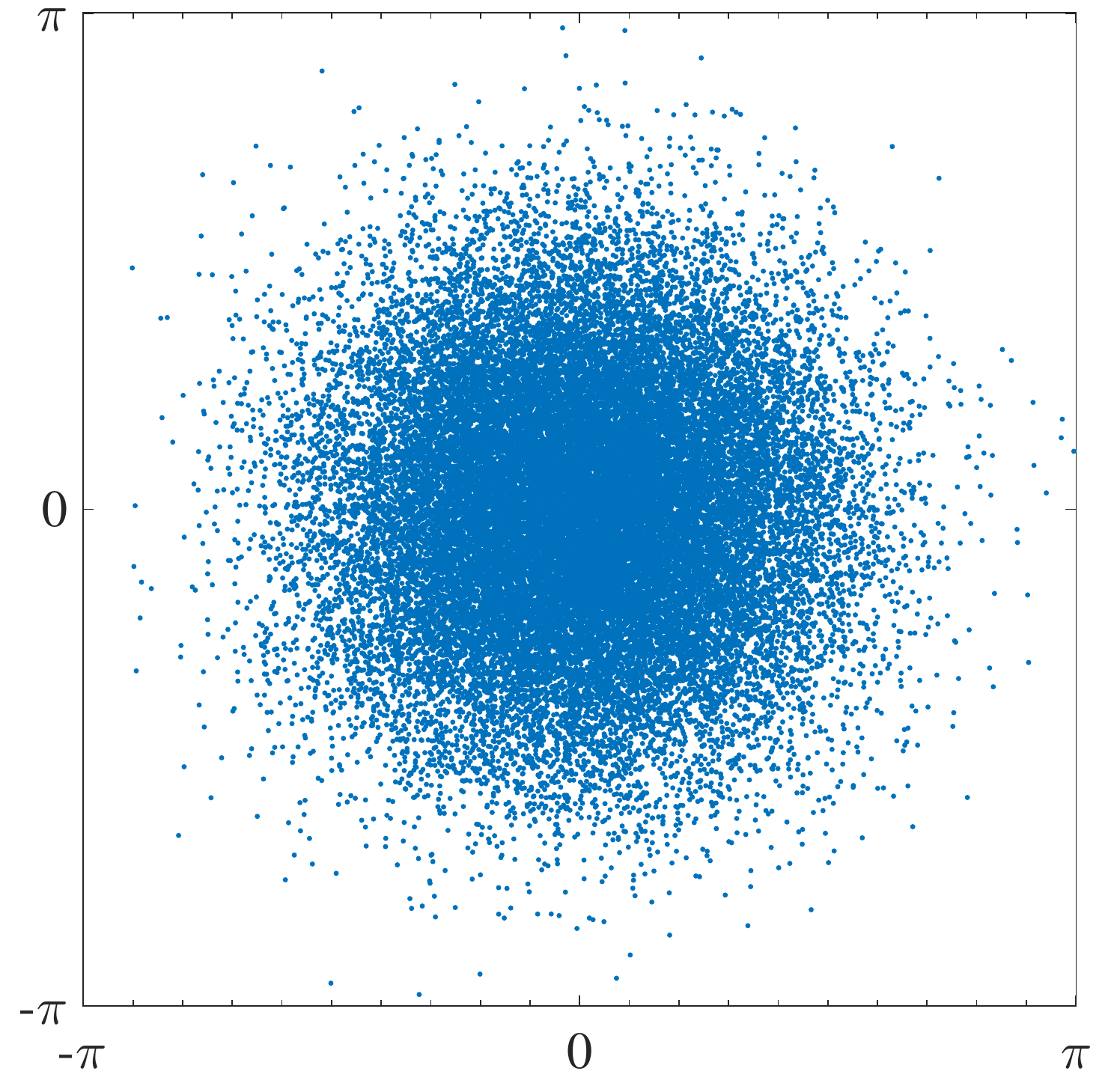}	\\
	(c)	
\end{tabular}
\end{center}
\caption{\label{Fig:RI:description} \small
Radio-astronomical imaging problem. 
(a)~Original image of W28 in linear scale. 
(b)~Original image of W28 in log scale. 
(c)~Normalized continuous Fourier space showing the frequencies selected to obtain $y$, using a random samplings with $M/N=0.5$.
}
\end{figure}

Radio astronomy aims to observe the sky at high angular resolution through an array of antennas. 
New radio telescopes, such as the future flagship Square Kilometre Array (SKA) are intended to provide images at unprecedented resolutions and sensitivities, and on a wide frequency band. Data rate estimates, for the first phase of development of the telescope only, are around few terabytes per second. The massive amounts of data to be acquired will represent a great challenge for the infrastructure and signal processing, and the methods solving the inverse problem associated with the image reconstruction need to be fast and to scale well with the data volumes and the expected image sizes (gigapixel sizes for monochromatic imaging). In this context, not only image estimation but also associated uncertainty quantification methodologies, key to the scientific interpretation of the data, must scale to extreme dimension.

Formally, we are interested in estimating the original sky brightness distribution $\overline{x} \in \R^N$ from $M$ measurements $y \in \C^M$. The measurement operator $\Phi \in \C^{M \times N}$, which in the simplest setting, consists in a non-uniform Fourier sampling operator, and $w \in \C^M$ is a realization of an additive complex i.i.d. Gaussian noise with zero mean and variance equal to $\sigma^2 \in \{0.01, 0.02, 0.03\}$, for both the real and imaginary parts of the noise.  
This model defines an ill-posed inverse problem for the recovery  of the radio sky $\overline{x}$. 
An  intensity image, representing W28 supernova with $N=256 \times 256$, is shown in Figure~\ref{Fig:RI:description}(a). Radio sky images are particularly difficult to reconstruct due to their important dynamic range. This dynamic range can be observed in the log-scaled image of W28 displayed in Figure~\ref{Fig:RI:description}(b). 
In our simulations we consider random Fourier samplings. This allows us to investigate the performance of the proposed uncertainty quantification approach with different Fourier samplings, considering several sampling ratio values $M/N \in \{0.5, 0.75, 1\}$. 
More precisely, we use Fourier samplings generated randomly through a Gaussian distribution, with zero mean and variance of 0.25 of the maximum frequency, creating a concentration of data at low frequencies. 
An example of Fourier samplings for the ratio $M/N=0.5$ is given in Figure~\ref{Fig:RI:description}(c).

In our simulations, we will perform Bayesian uncertainty quantification on three different spatially localized structures and on the background of the MAP estimate, defined mathematically in Definitions~\ref{example:structure} and \ref{example:background}, respectively. 

On the one hand, we investigate the uncertainty associated with the structures, denoted by Structure~1 and Structure~2, highlighted in red in the first columns of Figures~\ref{Fig:RI:struct1} and \ref{Fig:RI:struct2} respectively. We consider as well the structure presented in Section~\ref{ssec:Illust} for illustration of the method, namely Structure~3, highlighted in red in Figure~\ref{Fig:Illustration_high}. 
These three structures consist of compact or slightly extended sources corresponding to the definition of $\Sc$ given by Definition~\ref{example:structure}. 
This set $\Sc$ is characterized by $\Lc$, chosen such that $\Lc = \frac{1}{3}( \Lc_{3 \times 3} + \Lc_{7 \times 7} + \Lc_{11 \times 11})$, where $\Lc_{3 \times 3}$ (resp. $\Lc_{7 \times 7}$ and $\Lc_{11 \times 11}$) are built to model a 2D normalized convolution between the image (filled with zeros inside the structure) and  2D Gaussian convolution kernels of size $3 \times 3$ (resp. $7 \times 7$ and $11 \times 11$).
In addition, we choose$\tau = \operatorname{std}\big( \Mc(x^\dagger) - \Lc(\Mc(x^\dagger)) \big)$ for the set $\Sc_2$ and $b = 0$ and $\theta = \big\| \Lc \big( \Mc \big( x^\dagger \big)  \big) \big\|_2$ for the set $\Sc_3$.

On the other hand, we investigate the uncertainty associated with the background including all the weak intensity structures of the MAP estimates. 
The backgrounds of the MAP solutions obtained when considering $(M/N, \sigma^2) = (1, 0.01)$ and $(M/N, \sigma^2) = (0.5, 0.03)$ can be seen in the first column of Figure~\ref{Fig:RI:background}, where the log scale has been chosen to show values ranging from $10^{-4.2}$ to $\max_{1 \le n \le N}x^\dagger_n=1$. More precisely, the first two rows correspond to the case $(M/N, \sigma^2) = (1, 0.01)$, with the MAP estimate shown in the first row and zoomed images in the second rows. Similarly, the last two rows correspond to the case $(M/N, \sigma^2) = (0.5, 0.03)$. 
Mathematically, the set considered for the uncertainty quantification of the background is described in Example~\ref{example:background}, where $\underline{\tau}=0$ and $\overline{\tau} = \vartheta \| \Mc(x^\dagger) \|_2 / N_{\Mc}$ (for instance, $\vartheta=10^{-2}$).
In practice, for each MAP estimate $x^\dagger$, the background, represented by the operator $\Mc$ selecting its support, is determined through its complement, which is built in 2 steps. 
Firstly, we identify the structures of the image by selecting the elements of $x^\dagger$ with values larger than ${10^{-3}\times \max_{1 \le n \le N}x^\dagger_n}$. Then, the selected elements are dilated with disks of radius of size 7 pixels.

\subsubsection{Uncertainty quantification in radio astronomy}
\label{ssec:UQ_RI}


\begin{figure}[h!]
\begin{center}
\begin{tabular}{@{}c@{}c@{}c@{}}
	\includegraphics[height=3.8cm]{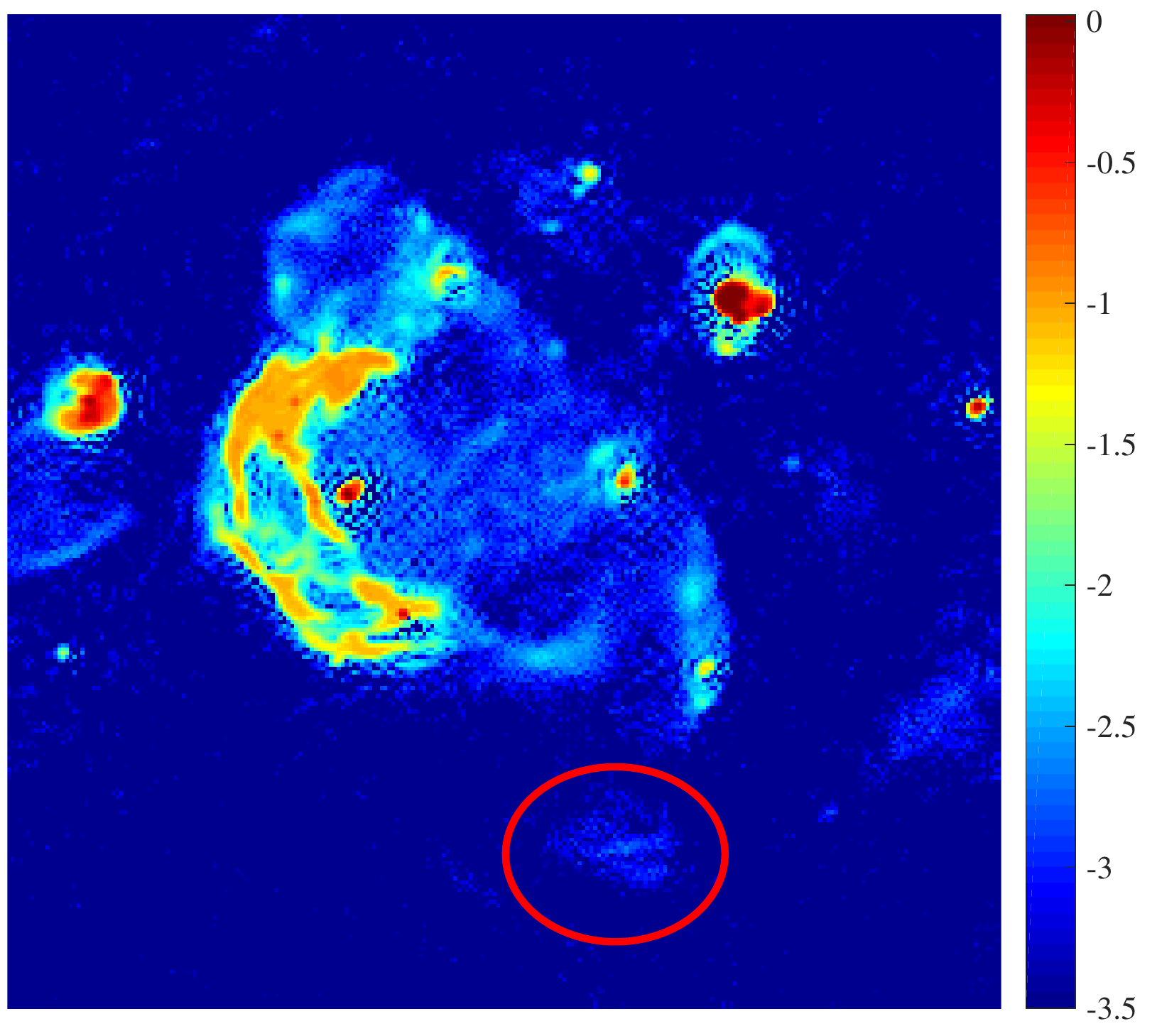}
&	\includegraphics[height=3.8cm]{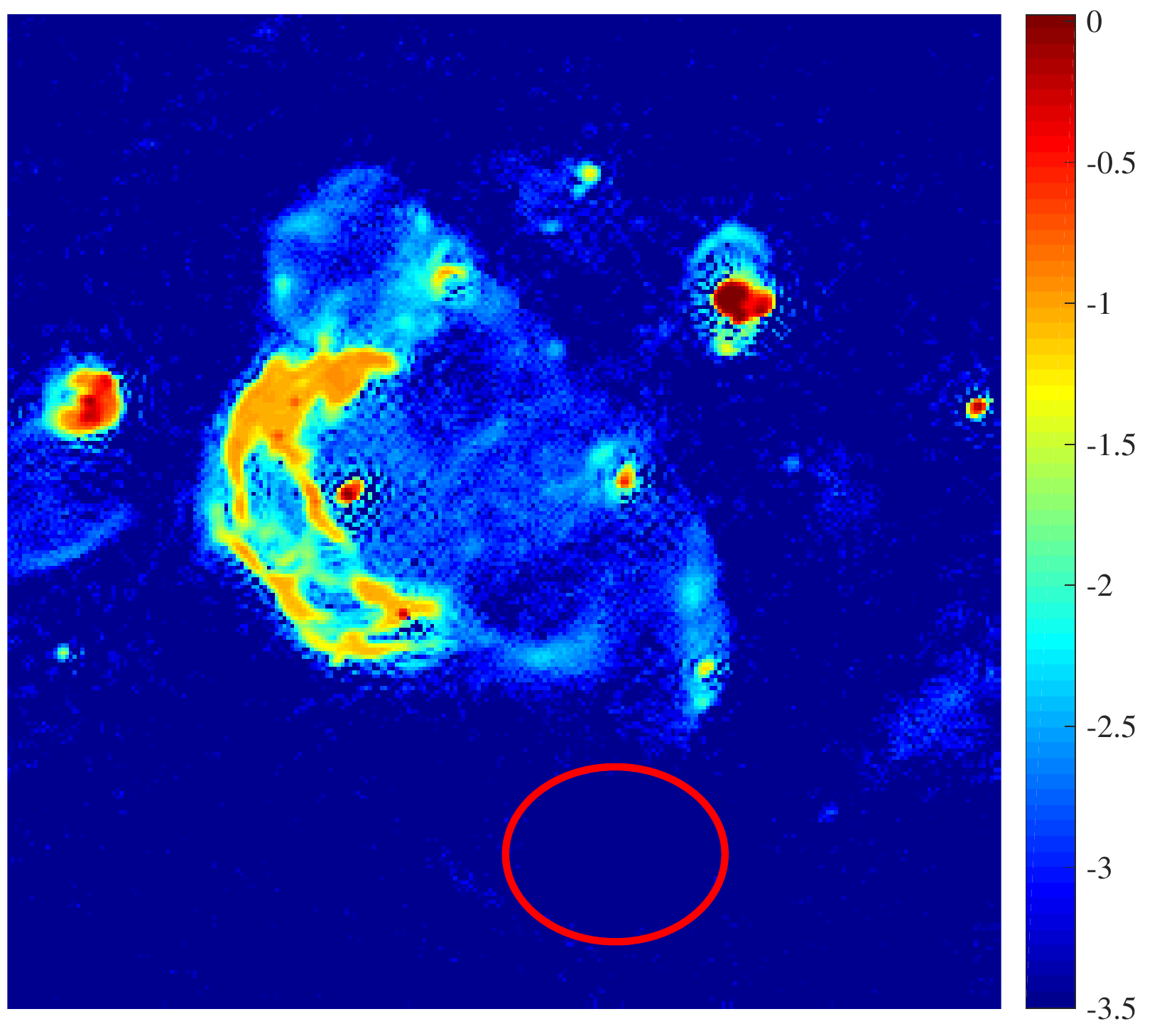}
&	\includegraphics[height=3.8cm]{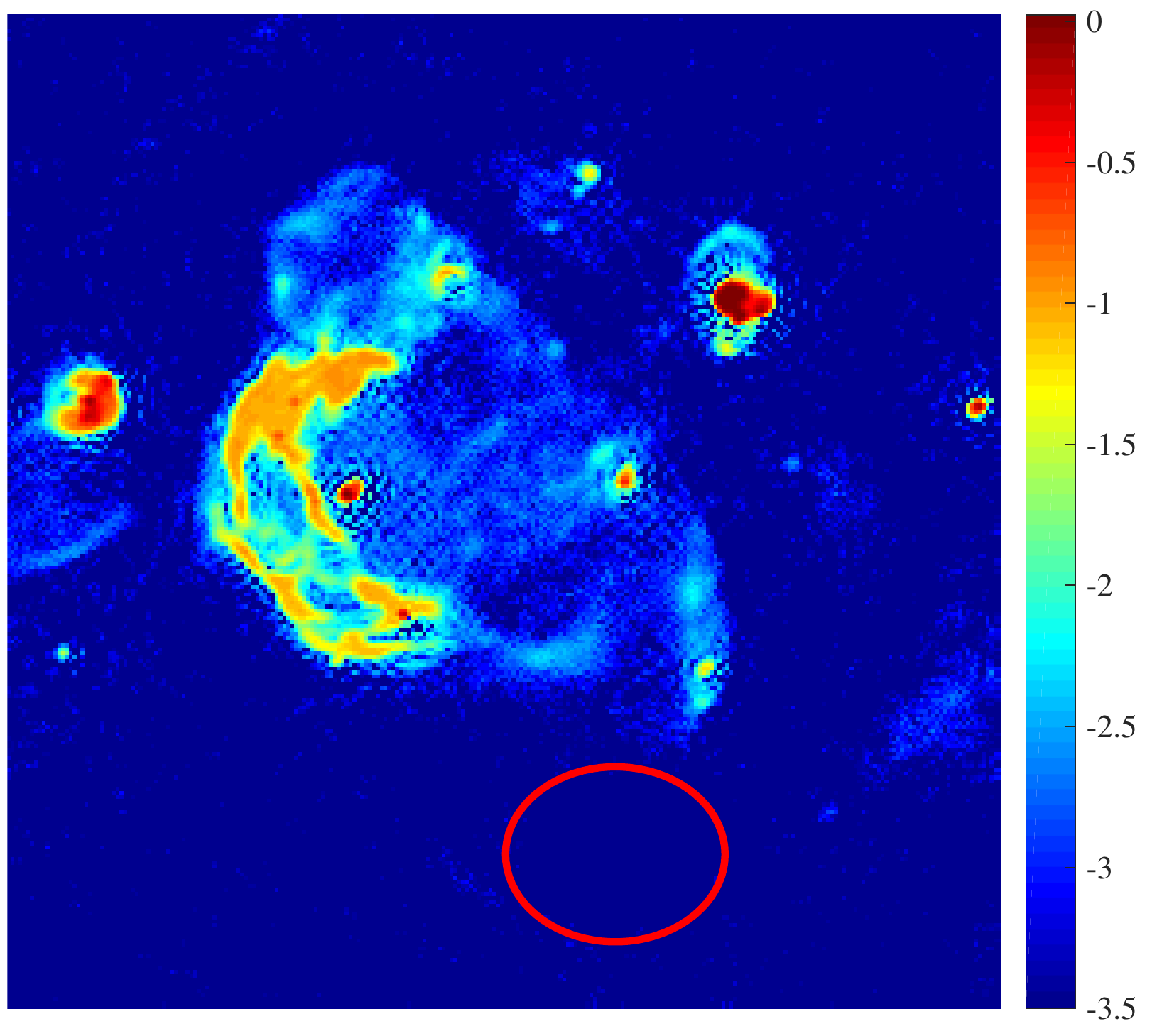}	\\
	\includegraphics[height=3.8cm]{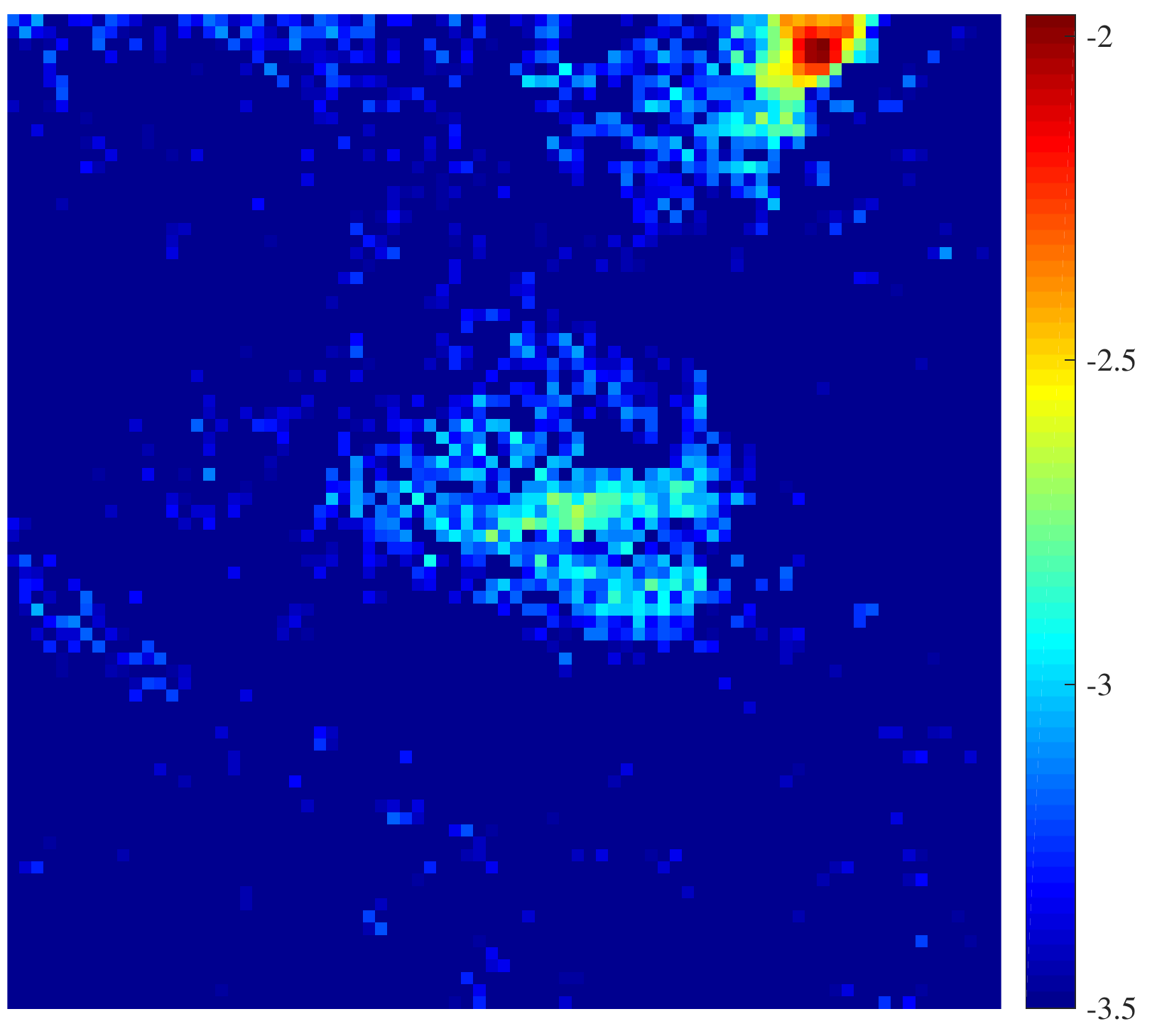}
&	\includegraphics[height=3.8cm]{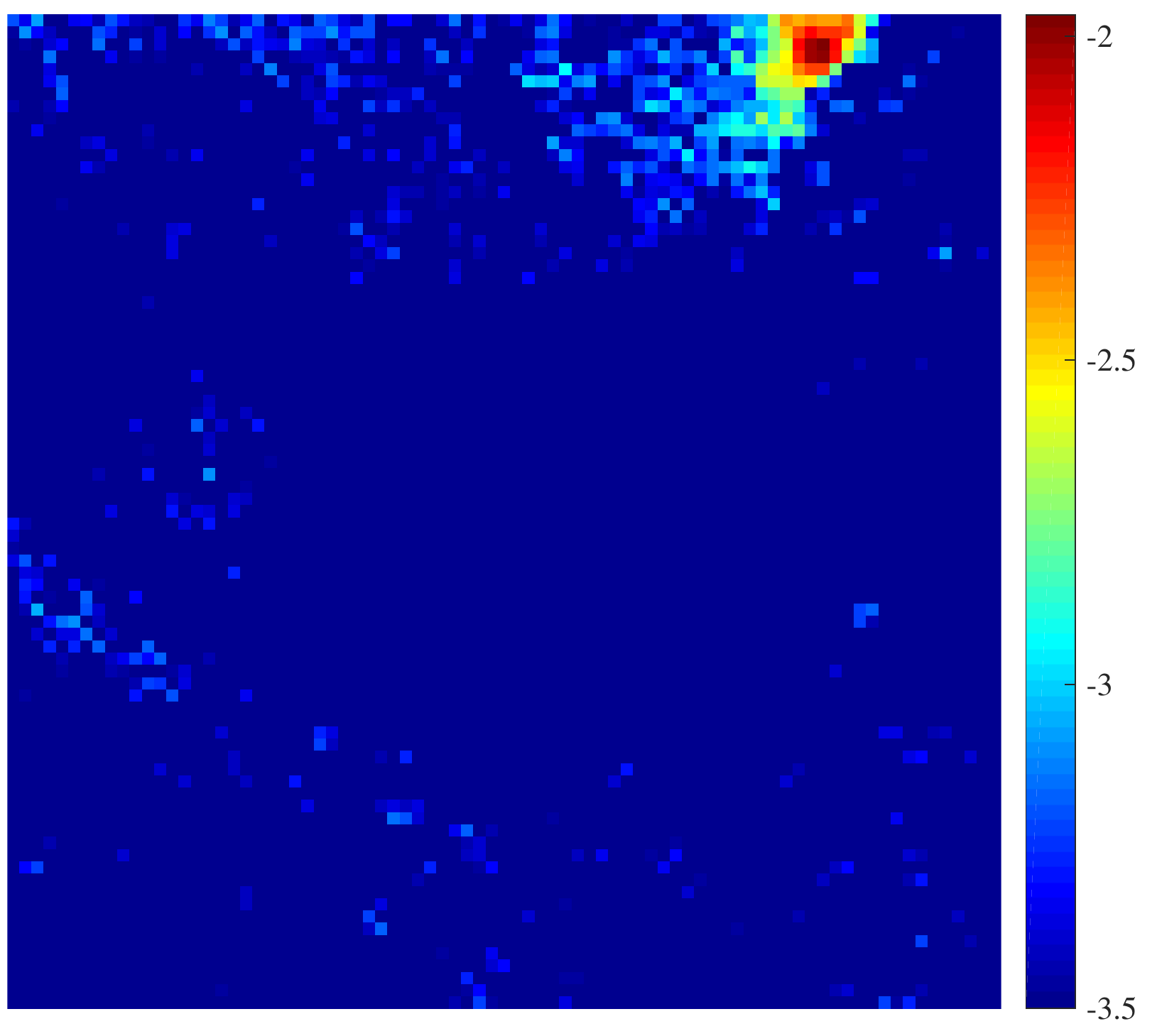}	
&	\includegraphics[height=3.8cm]{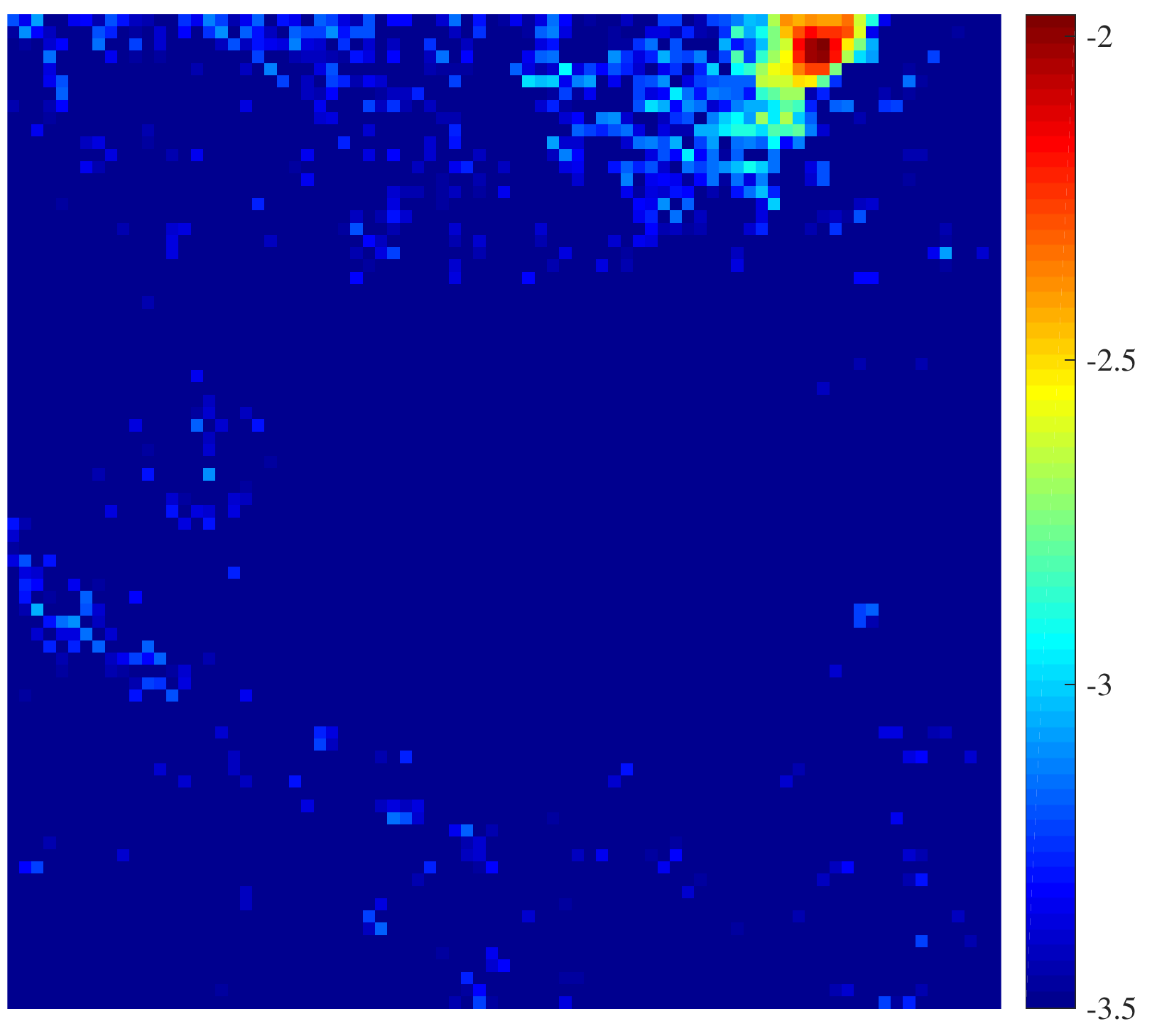}	
\end{tabular}
\end{center}
\caption{\label{Fig:RI:struct1} \small
Simulation results for the radio-astronomical imaging problem. Uncertainty quantification of Structure~1, in the case when $M/N = 0.5$ and $\sigma^2=0.03$. In this context, $\rho_\alpha = 0.07\%$ and $H_0$ cannot be rejected.
Top row: images in log scale with Structure~1 highlighted in red with, from left to right, 
$x^\dagger$, $x^\ddagger_{\widetilde{\Cc}_\alpha}$, and $x^\ddagger_{\Sc}$. 
Bottom row: zoomed images in log scale on the area of Structure~1, corresponding to the images displayed in first row. The log scale in the zoomed images is adapted to better emphasize Structure~1.
}
\end{figure}

\begin{figure}[h!]
\begin{center}
\begin{tabular}{@{}c@{}c@{}c@{}}
	\includegraphics[height=3.8cm]{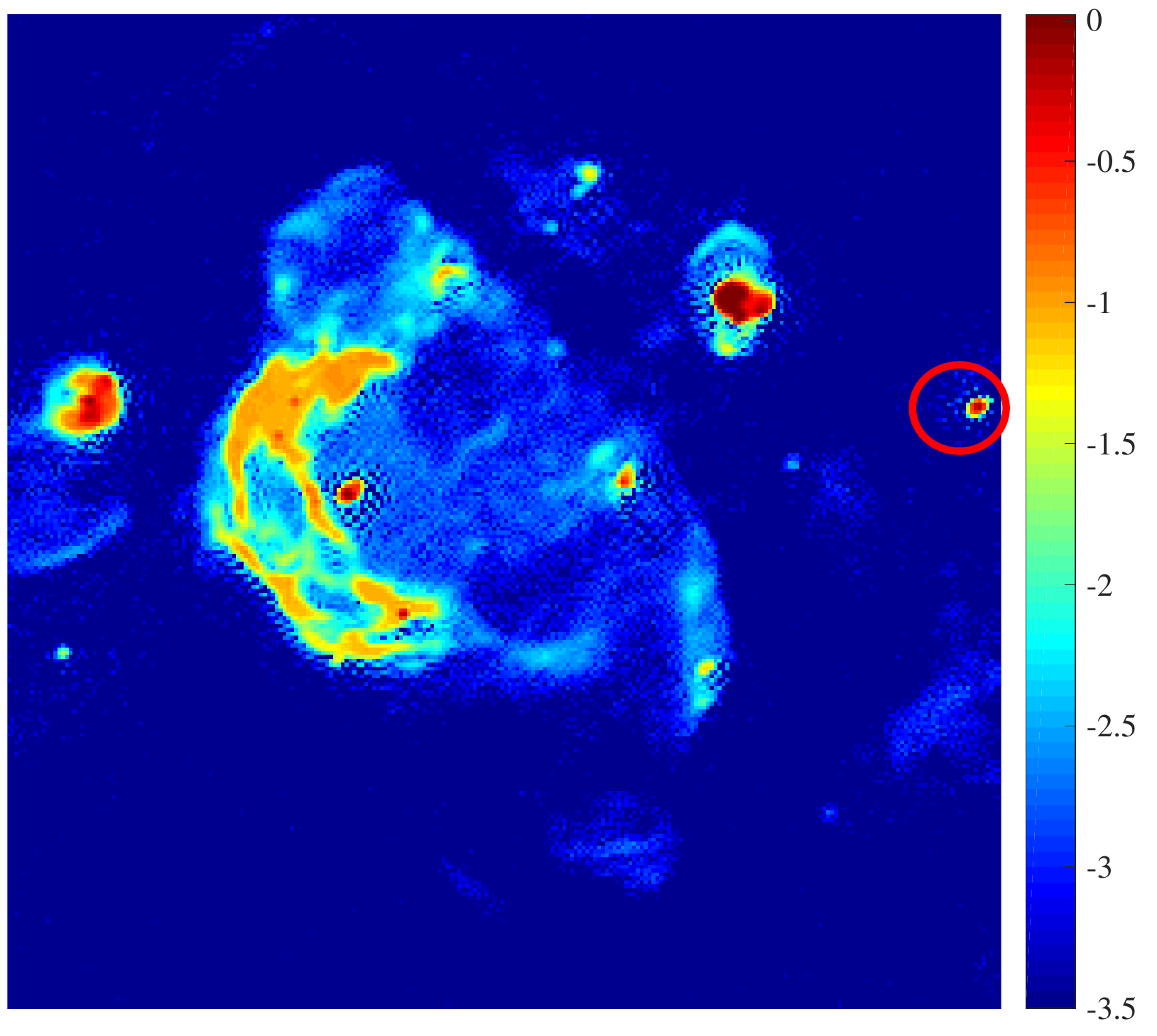}
&	\includegraphics[height=3.8cm]{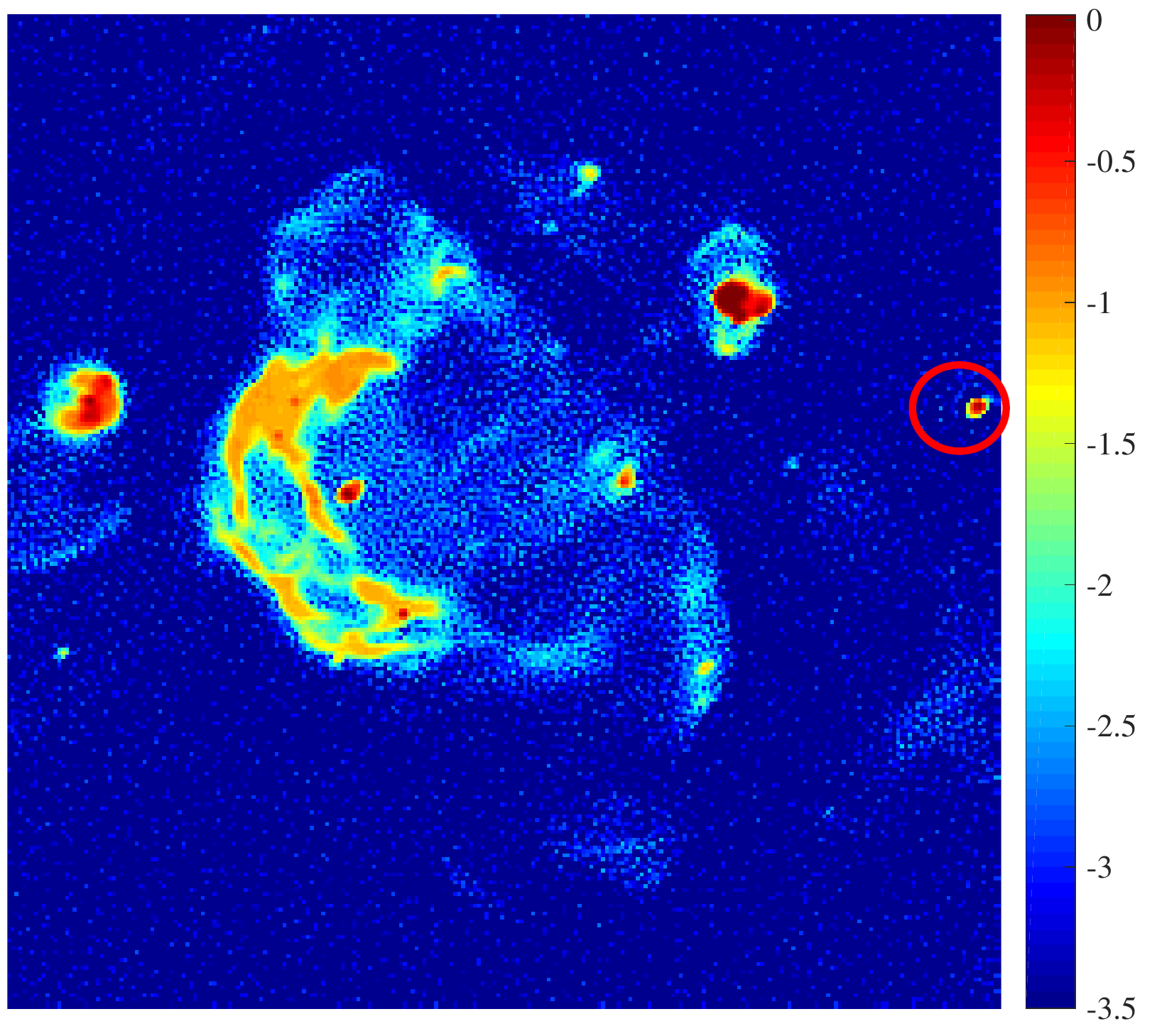}
&	\includegraphics[height=3.8cm]{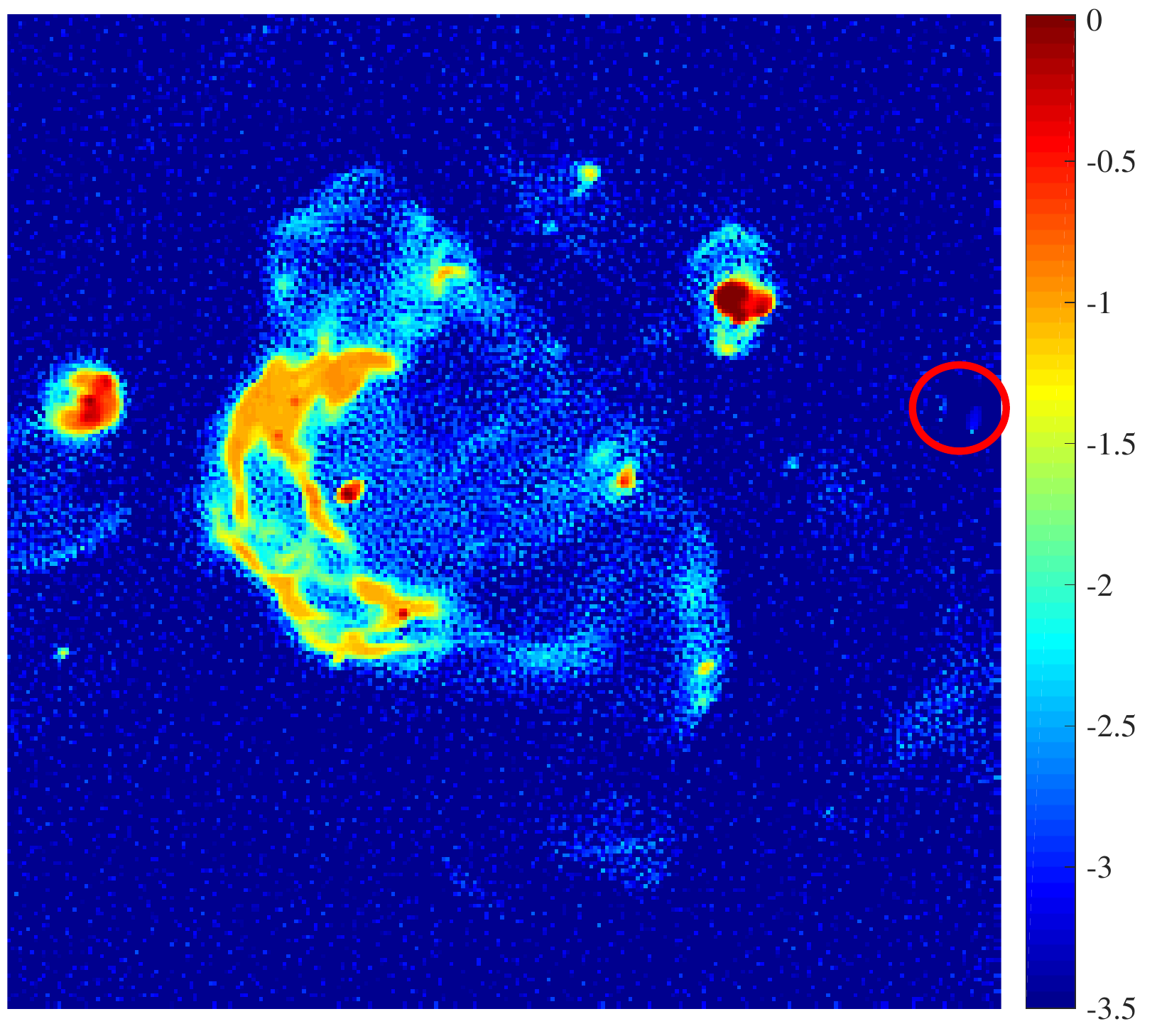}	\\
	\includegraphics[height=3.8cm]{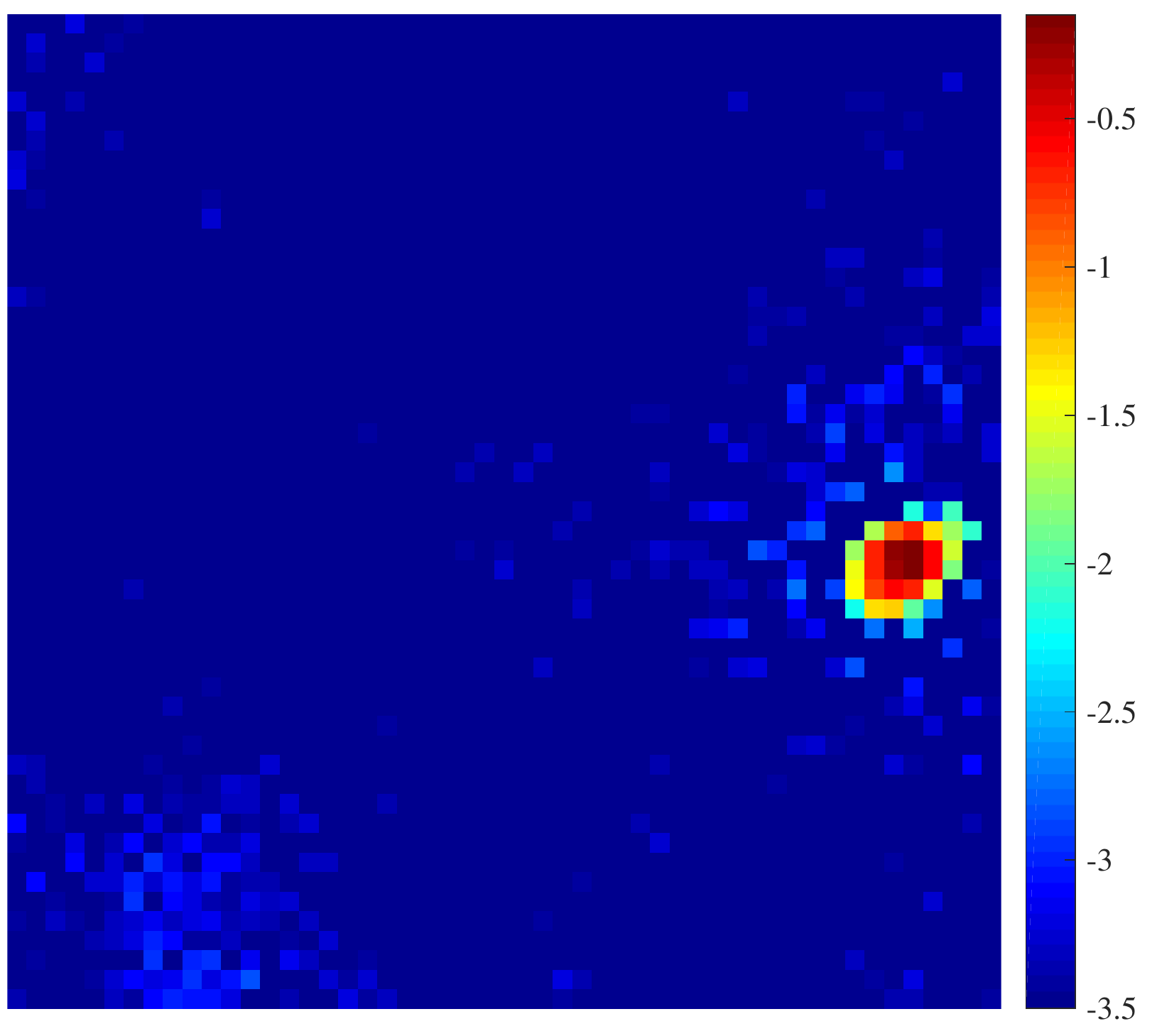}
&	\includegraphics[height=3.8cm]{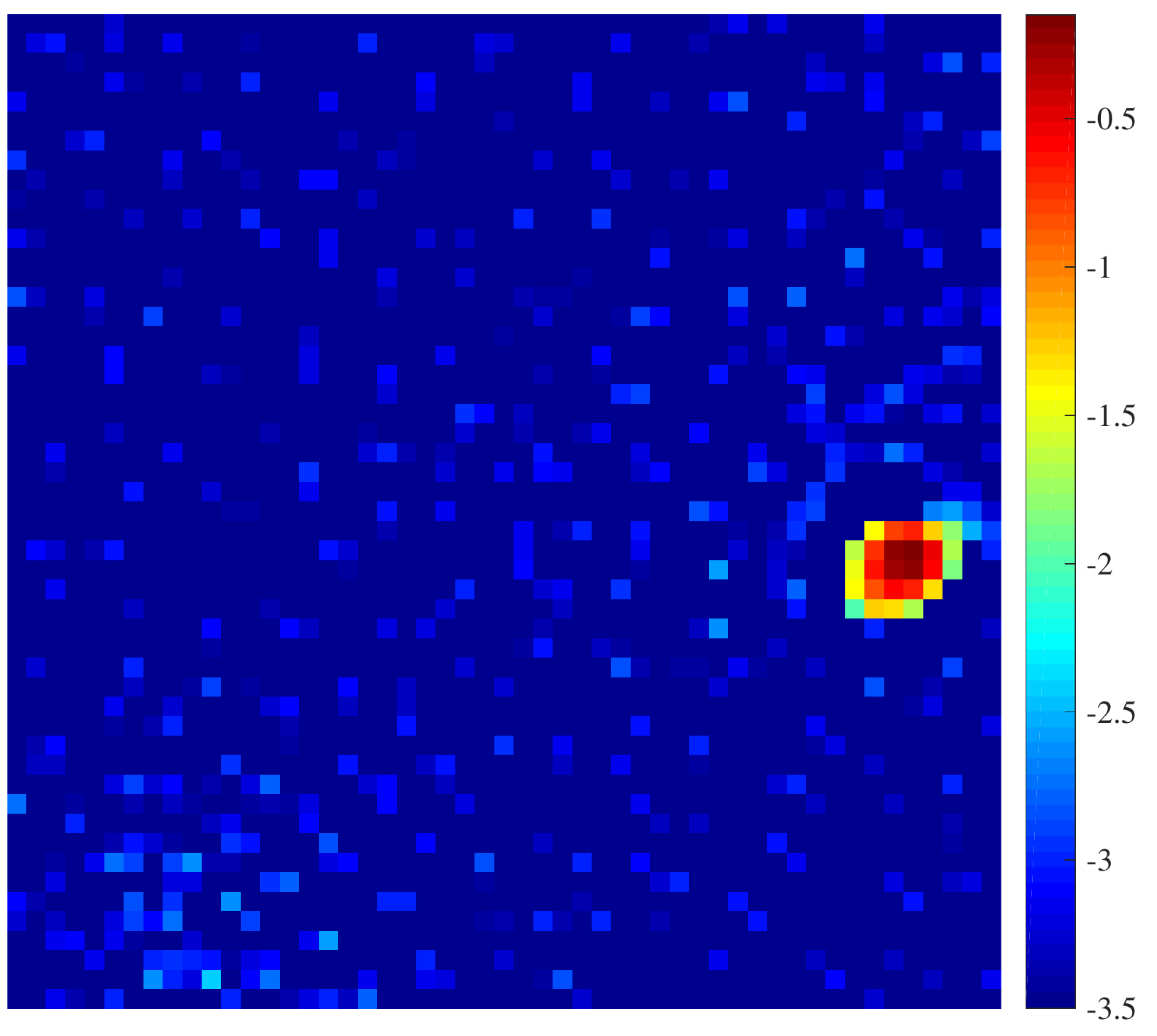}	
&	\includegraphics[height=3.8cm]{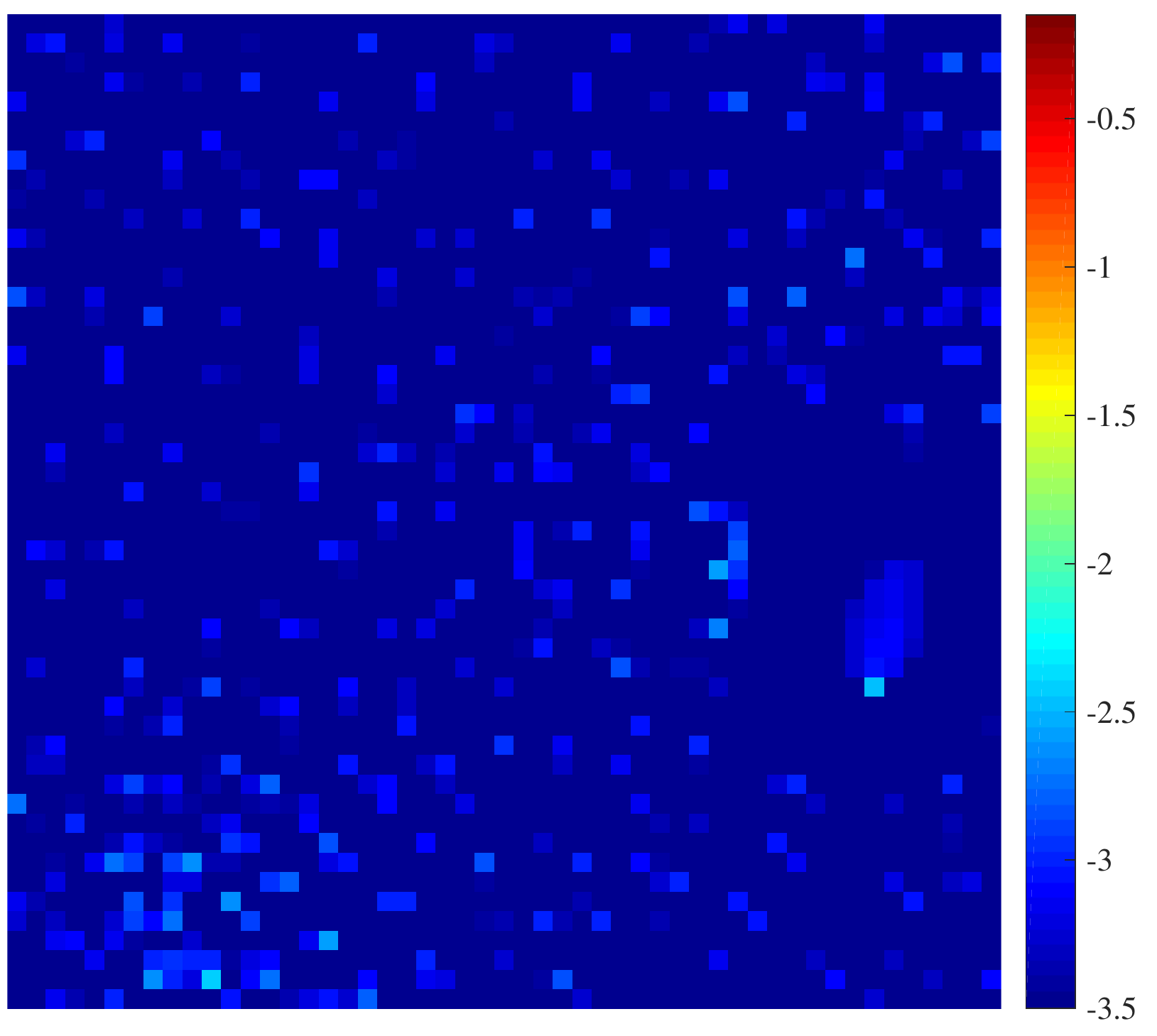}	
\end{tabular}
\end{center}
\caption{\label{Fig:RI:struct2}\small
Simulation results for the radio-astronomical imaging problem. Uncertainty quantification of Structure~2, in the case when $M/N = 1$ and $\sigma^2=0.01$. In this context, $\rho_\alpha = 96.58\%$ of the intensity's structure is confirmed at $99\%$, and $H_0$ is rejected with significance $1\%$.
Top row: images in log scale with Structure~2 highlighted in red with, from left to right, 
$x^\dagger$, 
$x^\ddagger_{\widetilde{\Cc}_\alpha}$, 
and $x^\ddagger_{\Sc}$. 
Bottom row: zoomed images in log scale on the area of Structure~2, corresponding to the images displayed in first row. The log scale in the zoomed images is adapted to better emphasize Structure~2.
}
\end{figure}

\begin{figure}[h!]
\begin{center}
\begin{tabular}{@{}c@{}c@{}c@{}}
	\includegraphics[height=3.8cm]{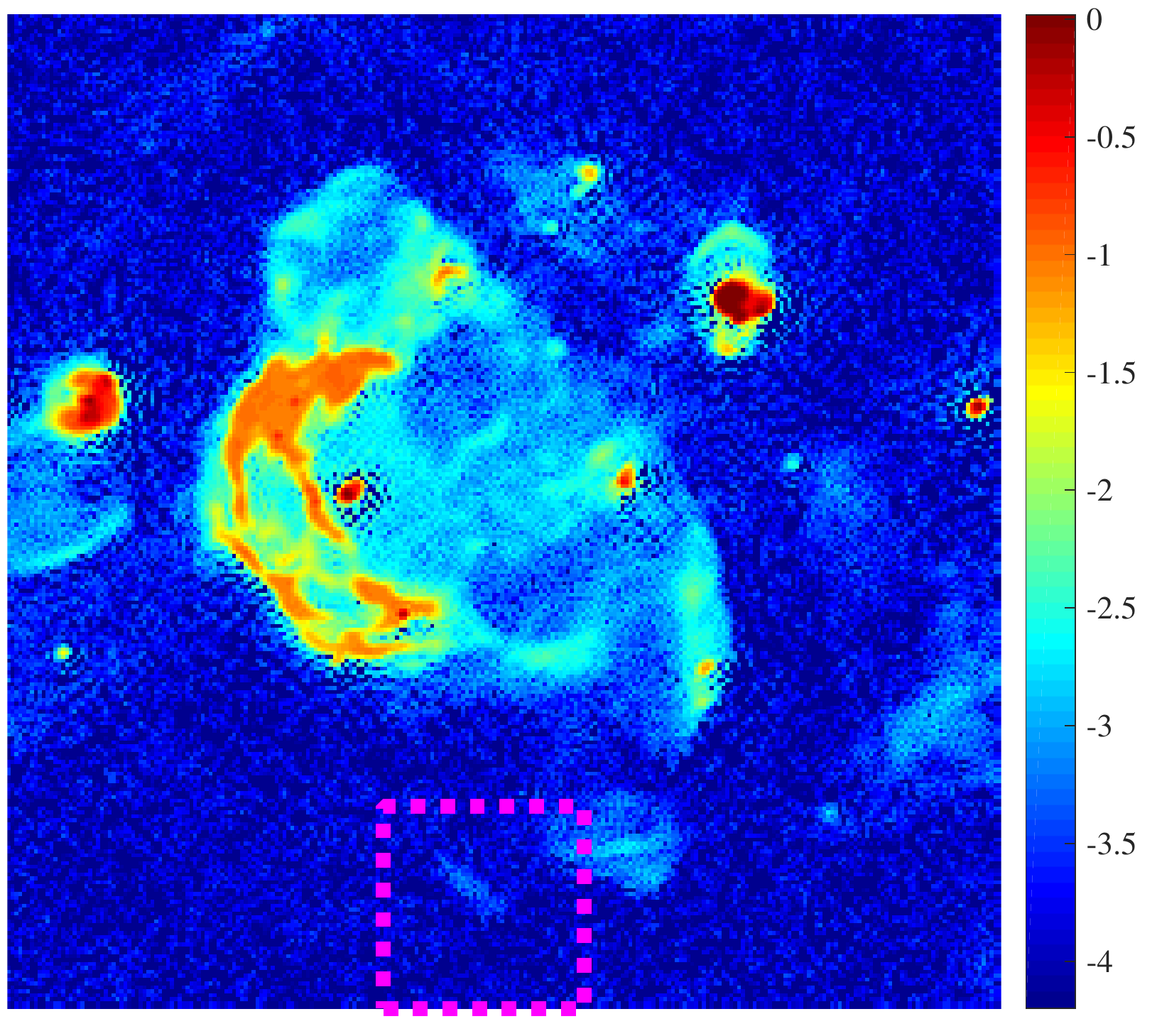}
&	\includegraphics[height=3.8cm]{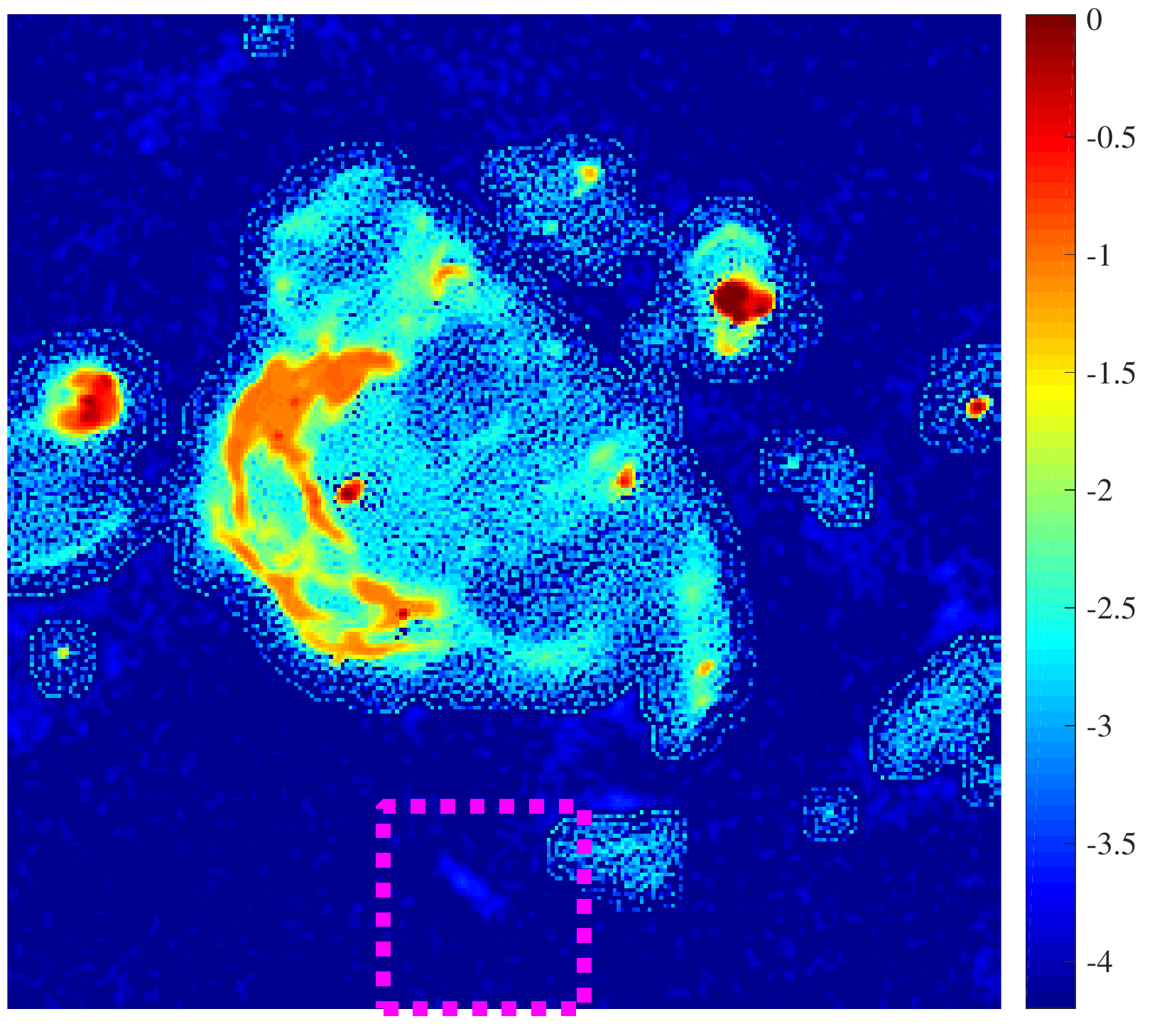}
&	\includegraphics[height=3.8cm]{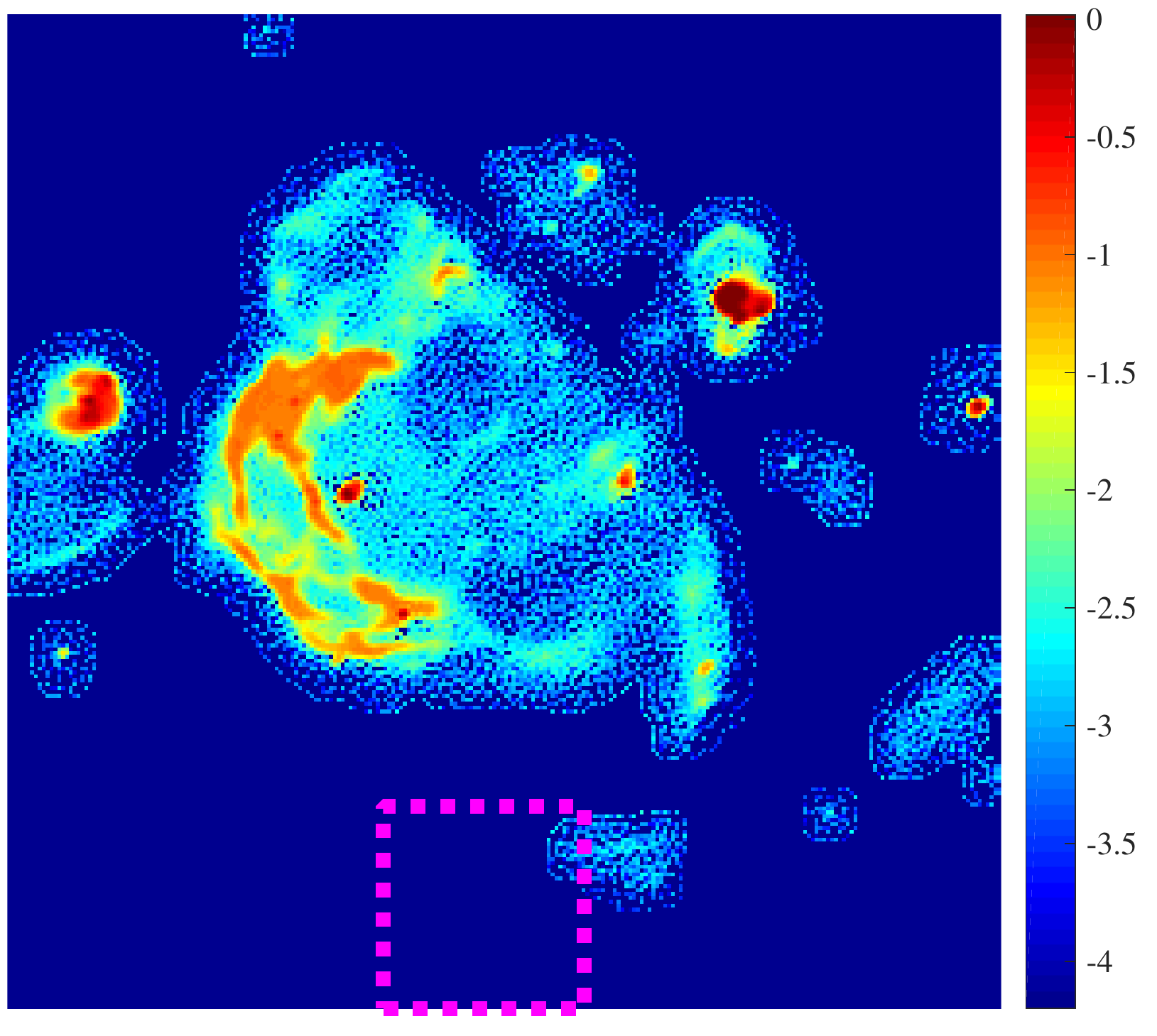}	\\
	\includegraphics[height=3.8cm]{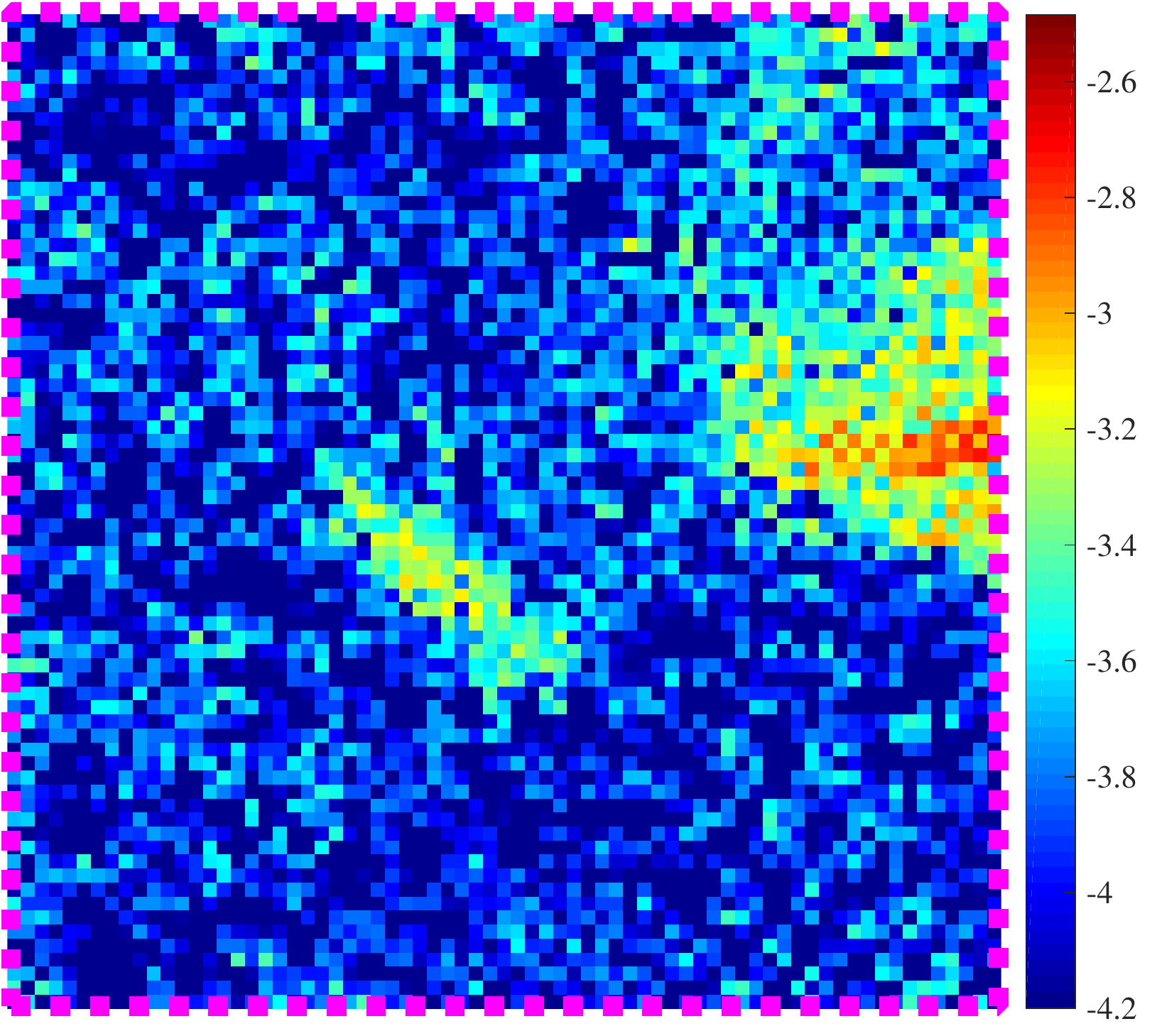}
&	\includegraphics[height=3.8cm]{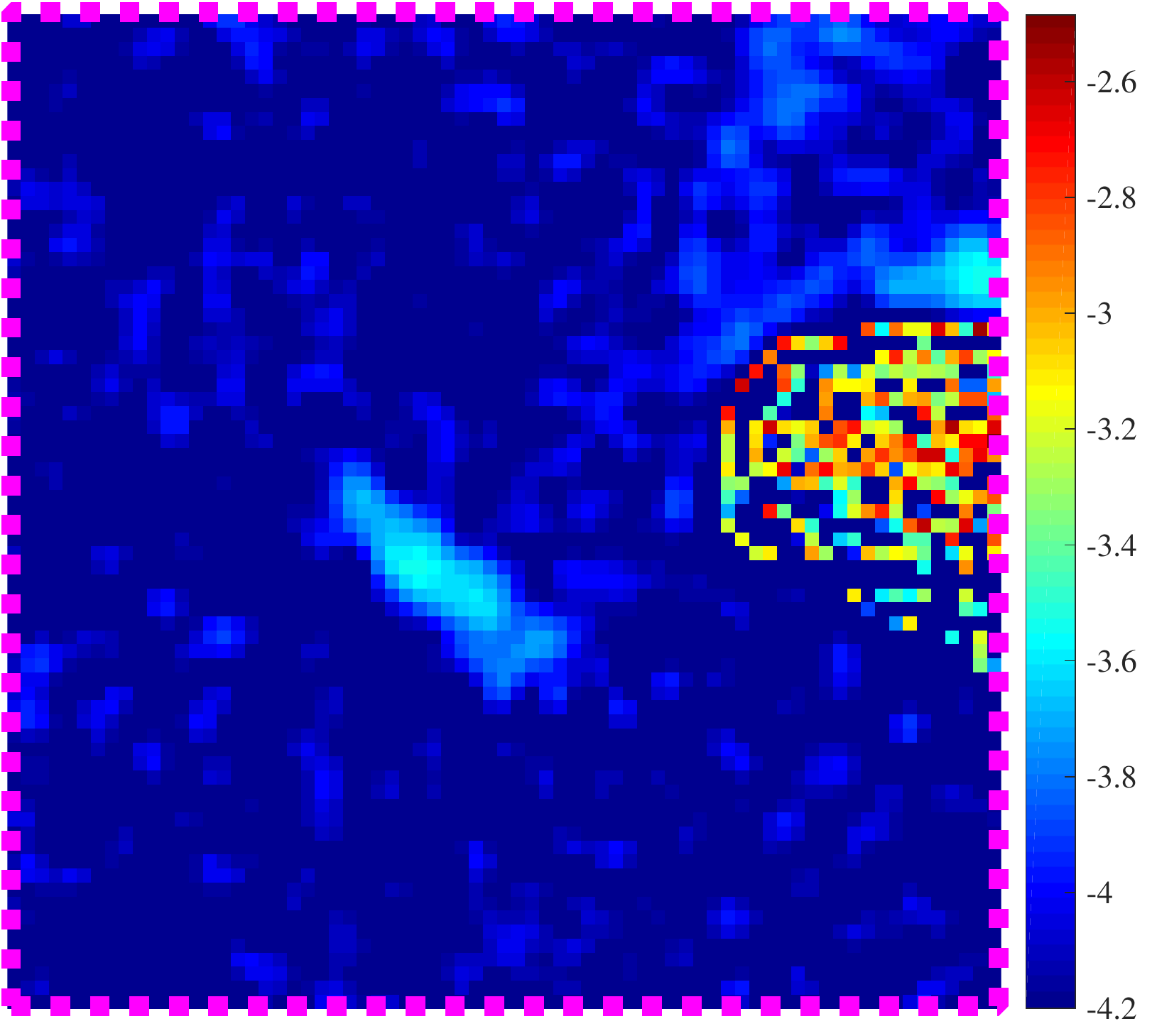}	
&	\includegraphics[height=3.8cm]{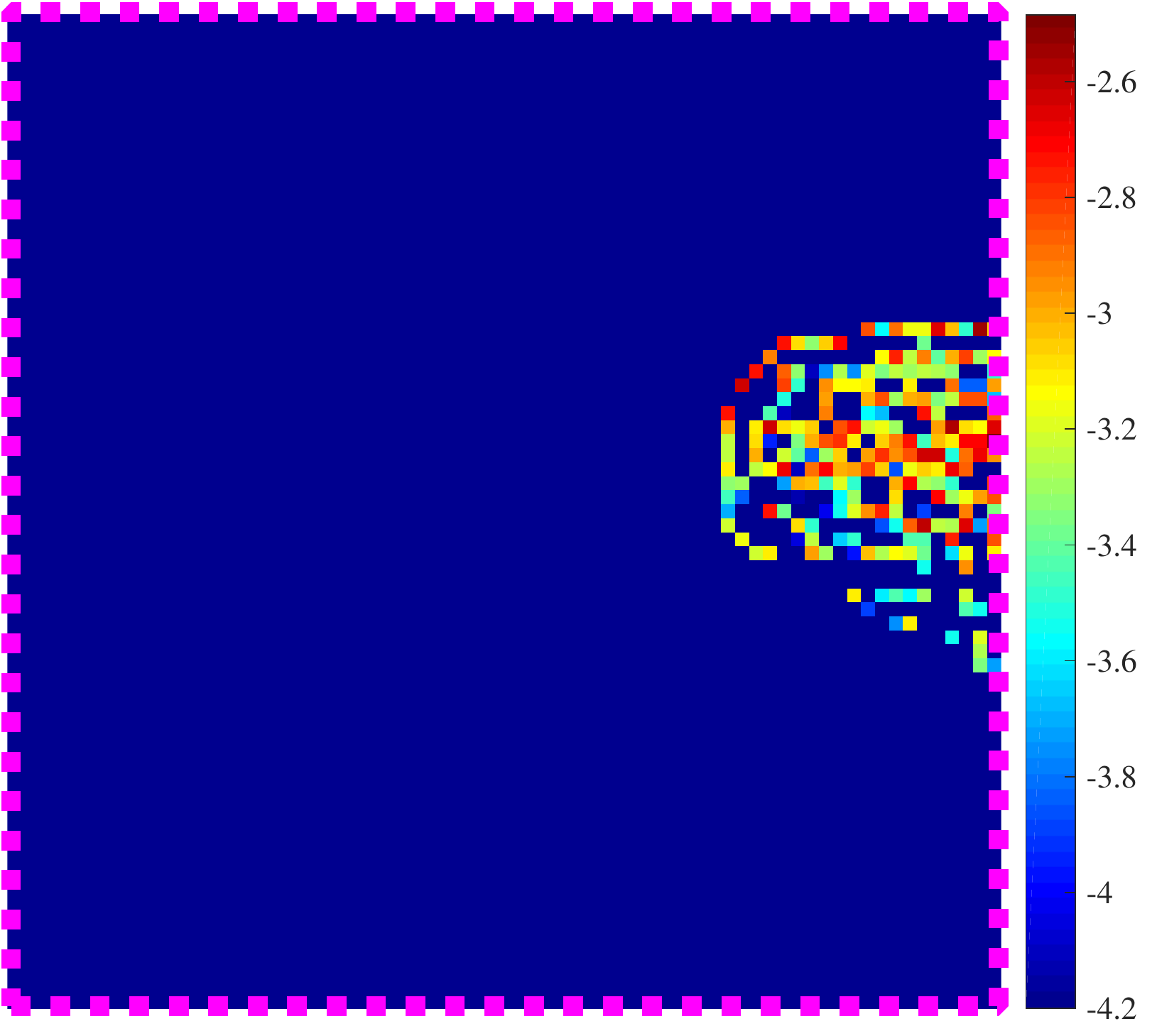}	\\[0.2cm]
	\includegraphics[height=3.8cm]{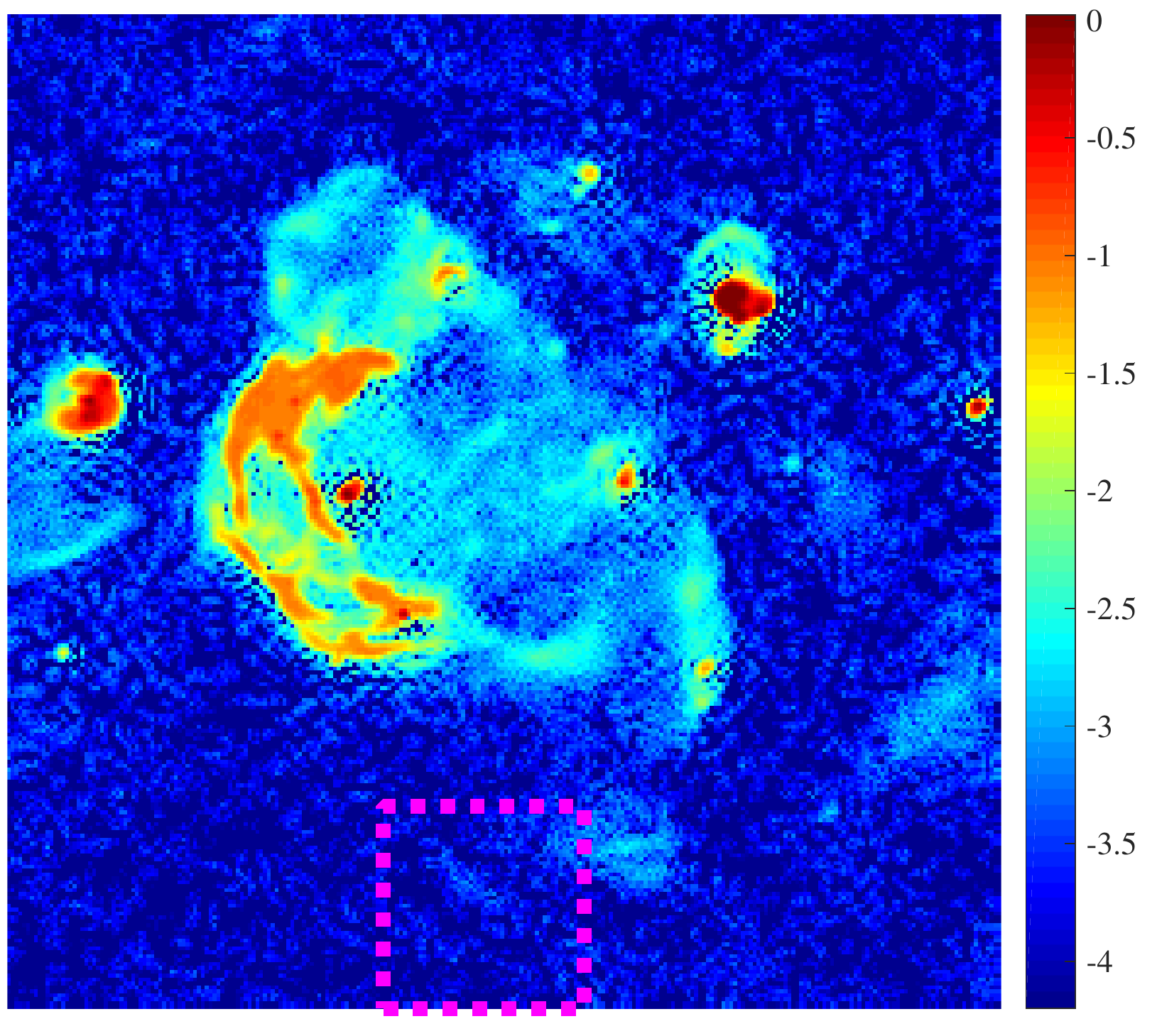}
&	\includegraphics[height=3.8cm]{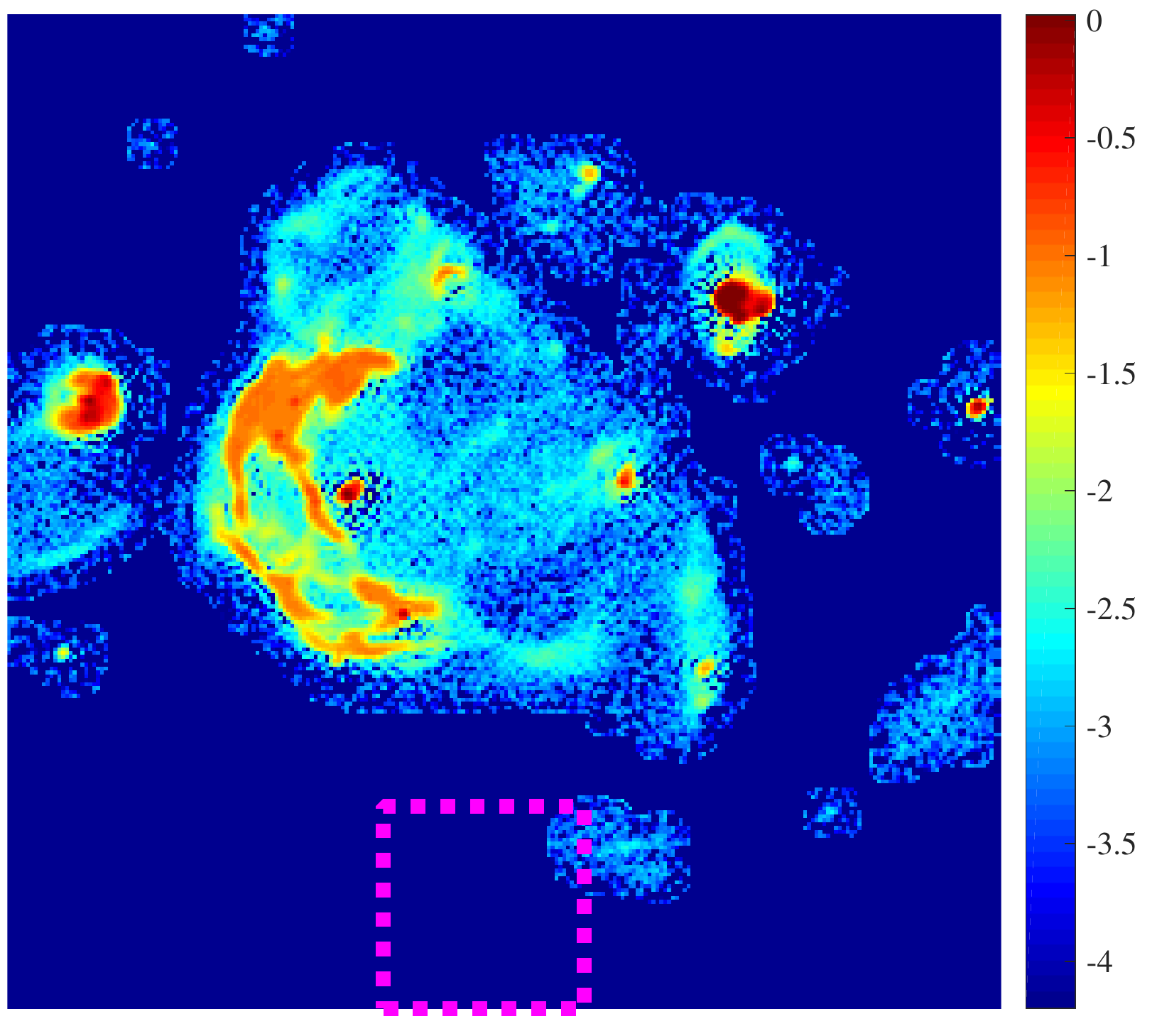}
&	\includegraphics[height=3.8cm]{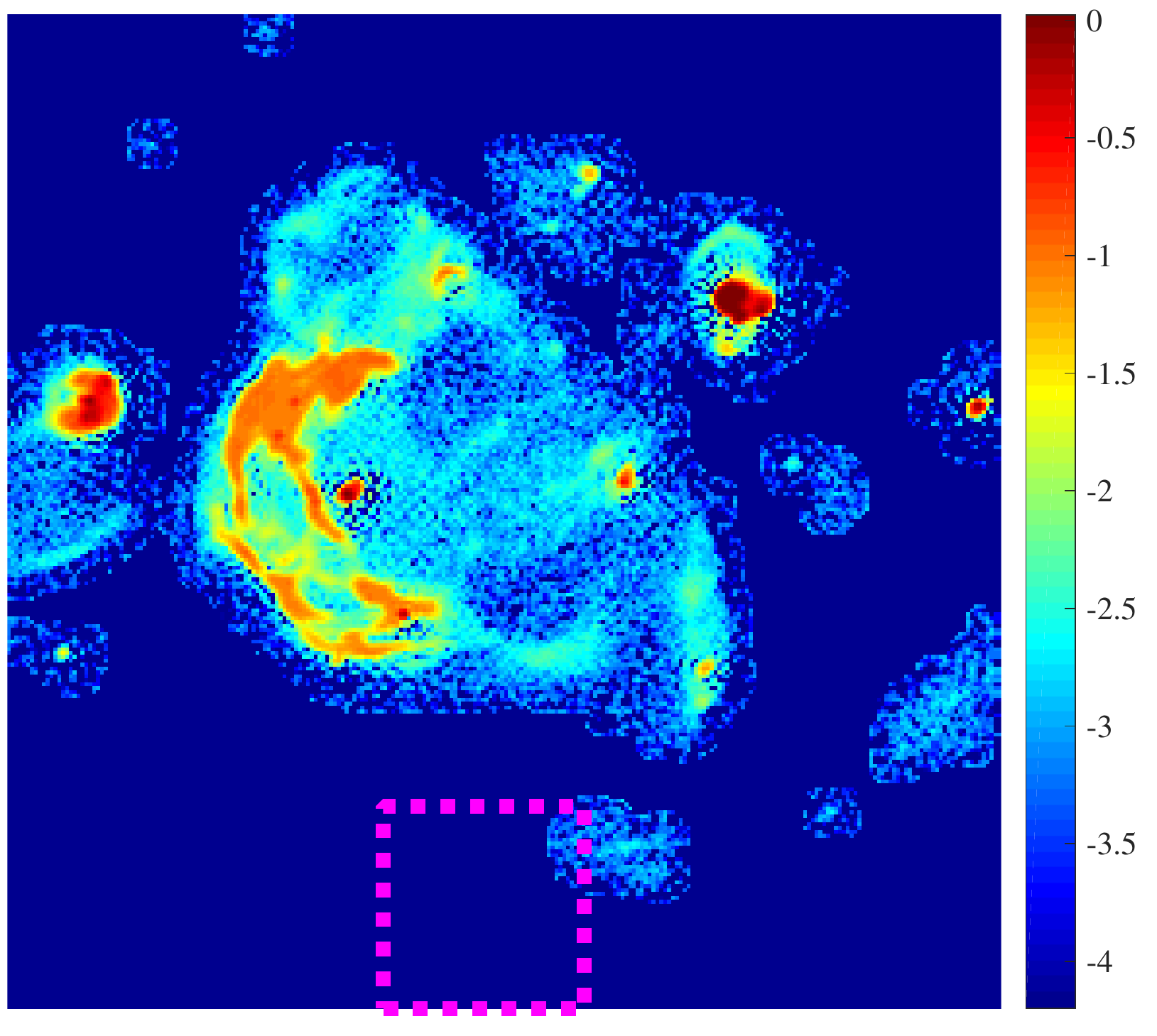}	\\
	\includegraphics[height=3.8cm]{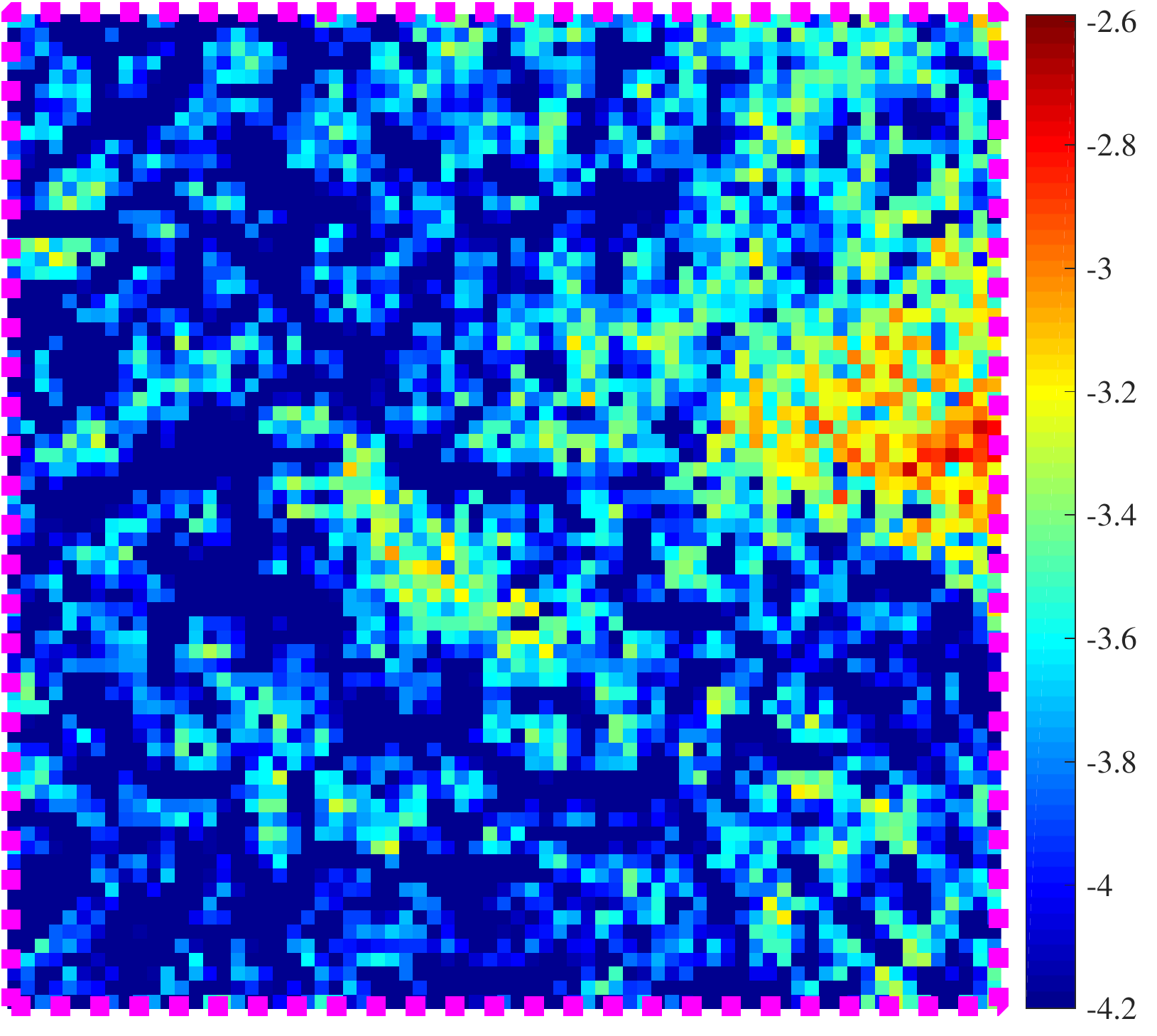}
&	\includegraphics[height=3.8cm]{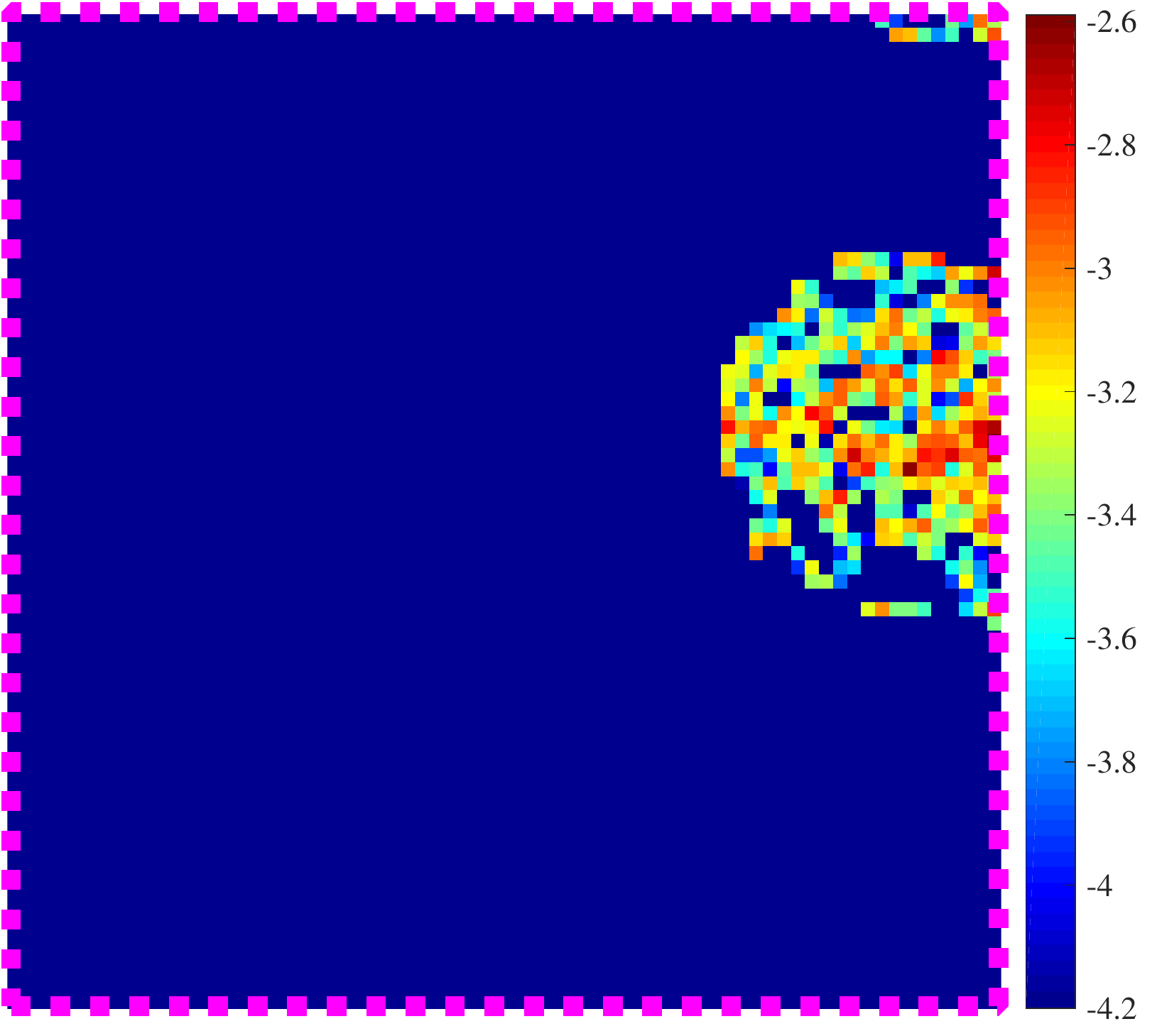}	
&	\includegraphics[height=3.8cm]{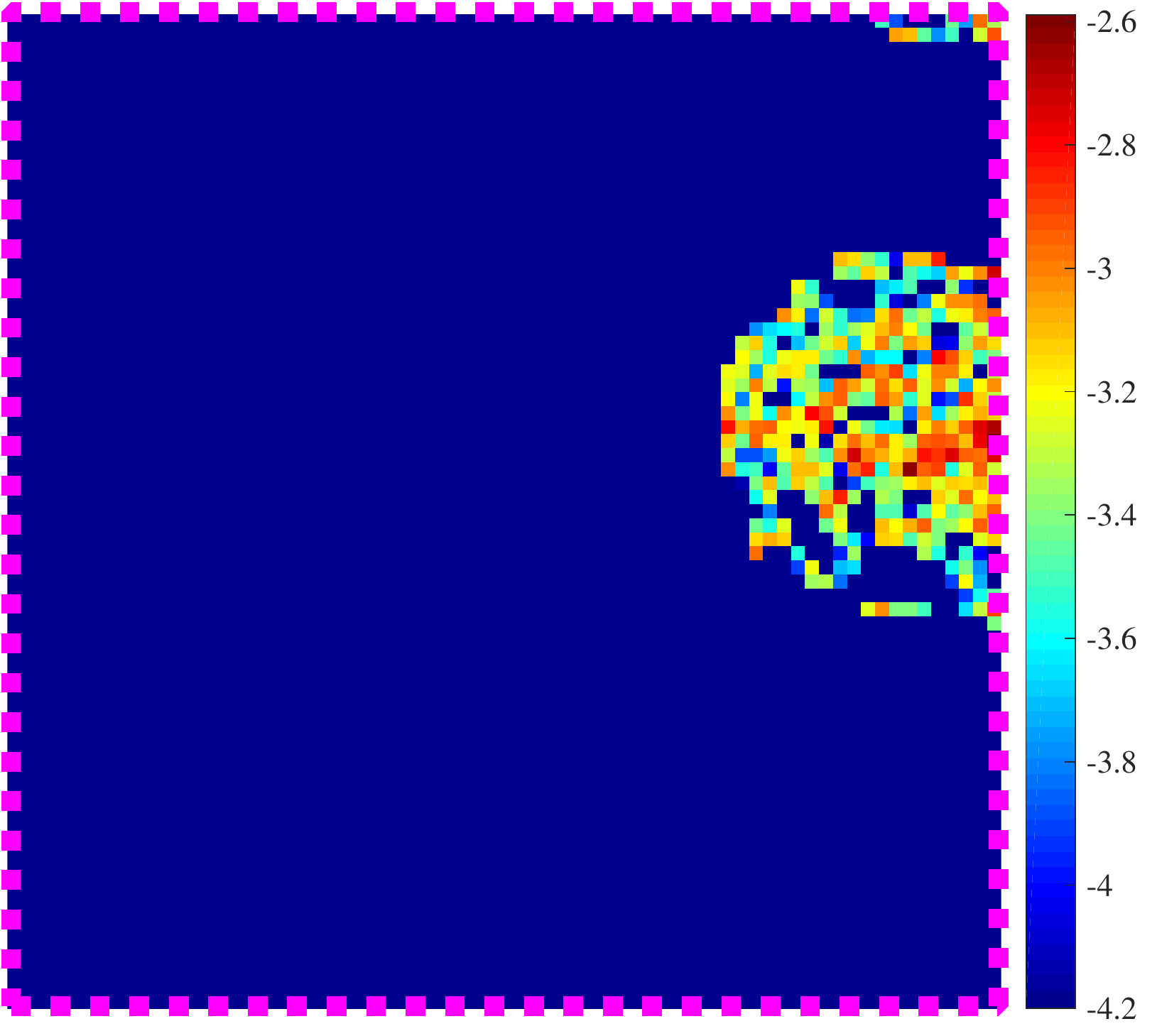}	
\end{tabular}
\end{center}


\caption{\label{Fig:RI:background} \small
Simulation results for the radio-astronomical imaging problem. Uncertainty quantification of the background of the MAP estimate. 
The first two rows correspond to the case when $M/N = 1$ and $\sigma^2=0.01$. In this context, $\rho_\alpha = 40.07\%$ of the intensity's structure is confirmed at $99\%$, and $H_0$ is rejected with significance $1\%$.
The last two rows correspond to the case when $M/N = 0.5$ and $\sigma^2=0.03$. In this context, $\rho_\alpha = 1.38\%$ and $H_0$ cannot be rejected.
In rows 1 and 3 are shown the images in log scale with a highlighted area corresponding to the zoomed images displayed in rows 2 and 4, respectively. From left to right: 
$x^\dagger$, 
$x^\ddagger_{\widetilde{\Cc}_\alpha}$, 
and $x^\ddagger_{\Sc}$. 
The log scale in the zoomed images is adapted to better emphasize the areas of interest.
}
\end{figure}

In this section we present our simulation results for the radio-astronomical imaging problem described in the previous section. 
For visual comparisons, we show the images obtained with the proposed uncertainty quantification approach, applied to the Structures~1 and 2, in Figs.~\ref{Fig:RI:struct1} and \ref{Fig:RI:struct2}, respectively. 

The top-row of Fig.~\ref{Fig:RI:struct1} shows, from left to right, the MAP estimate $x^\dagger$ obtained with $(M/N, \sigma^2) = (0.5, 0.03)$, and the two resulting images from Algorithm~\ref{algo:POCS_gen}: $x^\ddagger_{\widetilde{\Cc}_\alpha}$ and $x^\ddagger_{\Sc}$. In these images, Structure~1 is highlighted in red. 
The bottom-row of Fig.~\ref{Fig:RI:struct1} shows the images $x^\dagger$, $x^\ddagger_{\widetilde{\Cc}_\alpha}$ and $x^\ddagger_{\Sc}$, zoomed in the area of Structure~1. 
For this choice of $(M/N, \sigma^2)$, we have $\rho_\alpha = 0.07\% \approx 0\%$, and we conclude that $\widetilde{\Cc}_\alpha \cap \Sc \neq \emp$. Therefore, $H_0$ cannot be rejected (recall that failing to reject $H_0$ indicates that the structure considered is potentially not real, e.g. a reconstruction artefact).

Similarly, Structure~2 is highlighted in red in the top-row of Fig.~\ref{Fig:RI:struct2}, representing, from left to right, the MAP estimate $x^\dagger$ obtained with $(M/N, \sigma^2) = (1, 0.01)$, and the two corresponding images generated by Algorithm~\ref{algo:POCS_gen}: $x^\ddagger_{\widetilde{\Cc}_\alpha}$ and $x^\ddagger_{\Sc}$. 
The bottom-row of Fig.~\ref{Fig:RI:struct2} shows the images $x^\dagger$, $x^\ddagger_{\widetilde{\Cc}_\alpha}$ and $x^\ddagger_{\Sc}$, zoomed in the area of Structure~2. For this example, the structure's confirmed intensity percentage is equal to $\rho_\alpha = 96.58\%$. Consequently, we conclude that $\widetilde{\Cc}_\alpha \cap \Sc = \emp$, and $H_0$ is rejected with significance $\alpha = 1\%$ (recall that rejecting $H_0$ provides evidence to support that the structure considered is real, not an artefact). 

A complete description of the uncertainty quantification of Structure~3, in the cases when $(M/N, \sigma^2) = (0.5,0.01)$ and  $(M/N, \sigma^2) = (1, 0.01)$, is provided in Section~\ref{ssec:Illust}.

Results related to background removal are presented in Fig.~\ref{Fig:RI:background}. The top-row of this figure shows, from left to right, the MAP estimate $x^\dagger$ obtained with $(M/N, \sigma^2) = (1, 0.01)$, and the two images obtained using the proposed approach: $x^\ddagger_{\widetilde{\Cc}_\alpha}$ and $x^\ddagger_{\Sc}$. 
The second row of Fig.~\ref{Fig:RI:background} shows the images $x^\dagger$, $x^\ddagger_{\widetilde{\Cc}_\alpha}$ and $x^\ddagger_{\Sc}$, zoomed in the pink area. In this case, the structure's confirmed intensity percentage is equal to $\rho_\alpha = 40.07\%$. We can deduce then that $\widetilde{\Cc}_\alpha \cap \Sc = \emp$, and we conclude that $H_0$ is rejected with significance $\alpha = 1\%$. 
A second example is provided in Fig.~\ref{Fig:RI:background}, considering a smaller ratio $M/N$ and a higher noise level $\sigma^2$. We give in the third row of Fig.~\ref{Fig:RI:background}, from left to right, the MAP estimate $x^\dagger$ obtained with $(M/N, \sigma^2) = (1, 0.01)$, $x^\ddagger_{\widetilde{\Cc}_\alpha}$ and $x^\ddagger_{\Sc}$. The corresponding images, zoomed in the pink area, are provided in the fourth row of Fig.~\ref{Fig:RI:background}. For this second case, we have $\rho_\alpha = 1.38\% \approx 0\%$. Consequently, we conclude that $\widetilde{\Cc}_\alpha \cap \Sc \neq \emp$, and that $H_0$ cannot be rejected.

\begin{table}
\begin{center}
\begin{tabular}{c c}
{\renewcommand{\arraystretch}{1.5}
\begin{tabular}{|c|r|r|r|r|}
\multicolumn{5}{c}{}	\\
\cline{3-5}
\multicolumn{2}{c}{}
&	\multicolumn{3}{|c|}{$\sigma^2$}	\\
\cline{3-5}
\multicolumn{2}{c|}{}
&	0.01	&	0.02	&	0.03	\\
\hline
\multirow{3}{*}{$\dfrac{M}{N}$}
&	0.5
	&	\phantom{0}0.79	&	\phantom{0}0.24	&	\phantom{0}0.07	\\
\cline{2-5}
&	0.75
	&	1.26	&	0.41	&	0.34	\\
\cline{2-5}
&	1
	&	2.24	&	0.83	&	0.59	\\
\hline
\multicolumn{5}{c}{Structure~1}	
\end{tabular}}
&\hspace*{0.2cm}
{\renewcommand{\arraystretch}{1.5}
\begin{tabular}{|c|r|r|r|r|}
\multicolumn{5}{c}{}	\\
\cline{3-5}
\multicolumn{2}{c}{}
&	\multicolumn{3}{|c|}{$\sigma^2$}	\\
\cline{3-5}
\multicolumn{2}{c|}{}
&	0.01	&	0.02	&	0.03	\\
\hline
\multirow{3}{*}{$\dfrac{M}{N}$}
&	0.5
	&	\phantom{0}2.52	&	\phantom{0}0.57	&	\phantom{0}0.31	\\
\cline{2-5}
&	0.75
	&	11.1	&	1.33	&	0.64	\\
\cline{2-5}
&	1
	&	18.76	&	2.54	&	0.97	\\
\hline
\multicolumn{5}{c}{Structure~3}	
\end{tabular}}
\\
{\renewcommand{\arraystretch}{1.5}
\begin{tabular}{|c|r|r|r|r|}
\multicolumn{5}{c}{}	\\
\cline{3-5}
\multicolumn{2}{c}{}
&	\multicolumn{3}{|c|}{$\sigma^2$}	\\
\cline{3-5}
\multicolumn{2}{c|}{}
&	0.01	&	0.02	&	0.03	\\
\hline
\multirow{3}{*}{$\dfrac{M}{N}$}
&	0.5
	&	95.46	&	94.35	&	93.51	\\
\cline{2-5}
&	0.75
	&	96.2\phantom{0}	&	95.23	&	94.52	\\
\cline{2-5}
&	1
	&	96.58	&	95.72	&	95.15	\\
\hline
\multicolumn{5}{c}{Structure~2}	
\end{tabular}}
&\hspace*{0.2cm}
{\renewcommand{\arraystretch}{1.5}
\begin{tabular}{|c|r|r|r|r|}
\multicolumn{5}{c}{}	\\
\cline{3-5}
\multicolumn{2}{c}{}
&	\multicolumn{3}{|c|}{$\sigma^2$}	\\
\cline{3-5}
\multicolumn{2}{c|}{}
&	0.01	&	0.02	&	0.03	\\
\hline
\multirow{3}{*}{$\dfrac{M}{N}$}
&	0.5
	&	16.42	&	\phantom{0}2.93	&	\phantom{0}1.38	\\
\cline{2-5}
&	0.75
	&	33.85	&	17.32	&	6.84	\\
\cline{2-5}
&	1
	&	40.07	&	25.09	&	14.11	\\
\hline
\multicolumn{5}{c}{Background}	
\end{tabular}}
\end{tabular}
\end{center}
\caption{\label{Tab:RI:results} \small
Values of $\rho_\alpha$ in percentage ($\%$) for the four structures of interest in the radio-astronomical imaging problem.}
\end{table}

For all the three structures and the background, uncertainty quantification has been performed as well for other values of $M/N$ and $\sigma^2$. The values of $\rho_\alpha$ obtained for the considered cases are reported in Table~\ref{Tab:RI:results}. 
For all the experiments, it can be observed that $\rho_\alpha$ increases when $M/N$ increases or $\sigma^2$ decreases. In other words, larger is the number of measurements and higher is $\rho_\alpha$. On the contrary, higher is the noise level and lower is $\rho_\alpha$. This observation can be intuitively understood. Indeed, the MAP estimate is of lower quality when the observation data are not accurate (few noisy measurements). Consequently, in this case the uncertainty is higher. 

For Structure~1, the values of $\rho_\alpha$ range from $0.07\%$ for $(M/N,\sigma^2) = (0.5, 0.03)$ to $2.24\%$ for $(M/N, \sigma^2)=(1, 0.01)$. 
In all the considered cases, the value of $\rho_\alpha$ is low. As explained in Section~\ref{Ssec:simul:evalH0}, $\rho_\alpha$ cannot be equal to 0 due to the chosen stopping criteria. However, since in the worst case, the structure's confirmed intensity percentage is equal to $\rho_\alpha = 2.24\% \approx 0\%$, for all the considered cases we can conclude that $\widetilde{\Cc}_\alpha \cap \Sc \neq \emp$. As a consequence, we conclude that $H_0$ cannot be rejected and that Structure~1 is highly uncertain.

For Structure~2, the structure's confirmed intensity percentage is at least equal to $\rho_\alpha = 93.51\%$, corresponding to $(M/N, \sigma^2) = (0.5, 0.03)$. More precisely, the values of $\rho_\alpha$ are between $93.51\%$ for $(M/N,\sigma^2) = (0.5, 0.03)$ to $96.58\%$ for $(M/N, \sigma^2)=(1, 0.01)$. Thus, depending of the considered $(M/N,\sigma^2)$, we can conclude that between $93.51\%$ and $96.58\%$ of Structure~1 is confirmed at $99\%$. Consequently, for all the considered cases in this experiment, $\widetilde{\Cc}_\alpha \cap \Sc = \emp$ and we conclude that $H_0$ is rejected with significance $1\%$.

For Structure~3, the values of $\rho_\alpha$ range from $0.31\%$ to $18.76\%$, for $(M/N, \sigma^2) = (0.03, 0.5)$ and $(M/N, \sigma^2) = (0.01, 1)$ respectively. For this structure, different conclusions can be drawn. 
For $\sigma^2=0.03$ we have $\rho_\alpha \approx 0\%$ for all the considered under-sampling ratios, and we conclude that $\widetilde{\Cc}_\alpha \cap \Sc \neq \emp$ and that $H_0$ cannot be rejected. 
The other observations depend on the tolerance fixed by the user. For instance, if we consider that $\widetilde{\Cc}_\alpha \cap \Sc \neq \emp$ when $\rho_\alpha < 3\% $, the only cases when the structure is confirmed are $(M/N, \sigma^2) = (0.01, 1)$ and $(M/N, \sigma^2) = (0.01, 0.75)$, where the structure's confirmed intensity percentages are $\rho_\alpha = 18.76\%$ and $11.1\%$, respectively.

Concerning the uncertainty quantification of the background, the values of $\rho_\alpha$ range from $1.38\%$ to $40.07\%$, for $(M/N, \sigma^2) = (0.03, 0.5)$ and $(M/N, \sigma^2) = (0.01, 1)$ respectively. As for Structure~3, the conclusion for the different cases presented in Table~\ref{Tab:RI:results} depend on the tolerance fixed by the user. 
As previously, considering that $\widetilde{\Cc}_\alpha \cap \Sc \neq \emp$ when $\rho_\alpha < 3\% $, the only cases satisfying this condition are $(M/N, \sigma^2) = (0.03,0.5)$ and $(M/N, \sigma^2) = (0.02,0.5)$, with $\rho_\alpha = 1.38\% \approx 0\%$ and $\rho_\alpha = 2.93\% \approx 0\%$ respectively.

\begin{figure}[t]
\begin{center}
\begin{tabular}{@{}c@{}c@{}}
	\includegraphics[height=5.5cm]{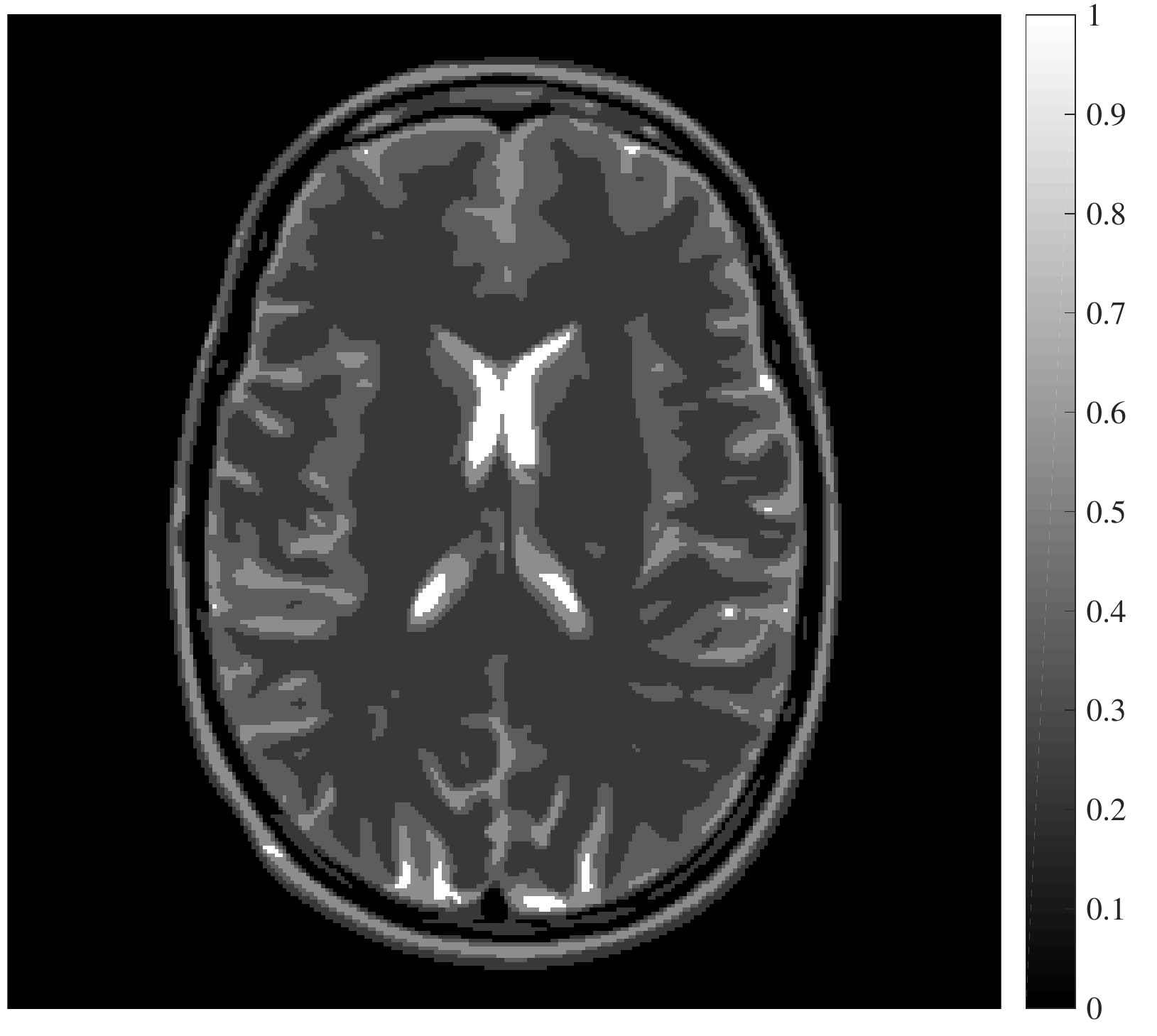}
&	\includegraphics[height=5.5cm]{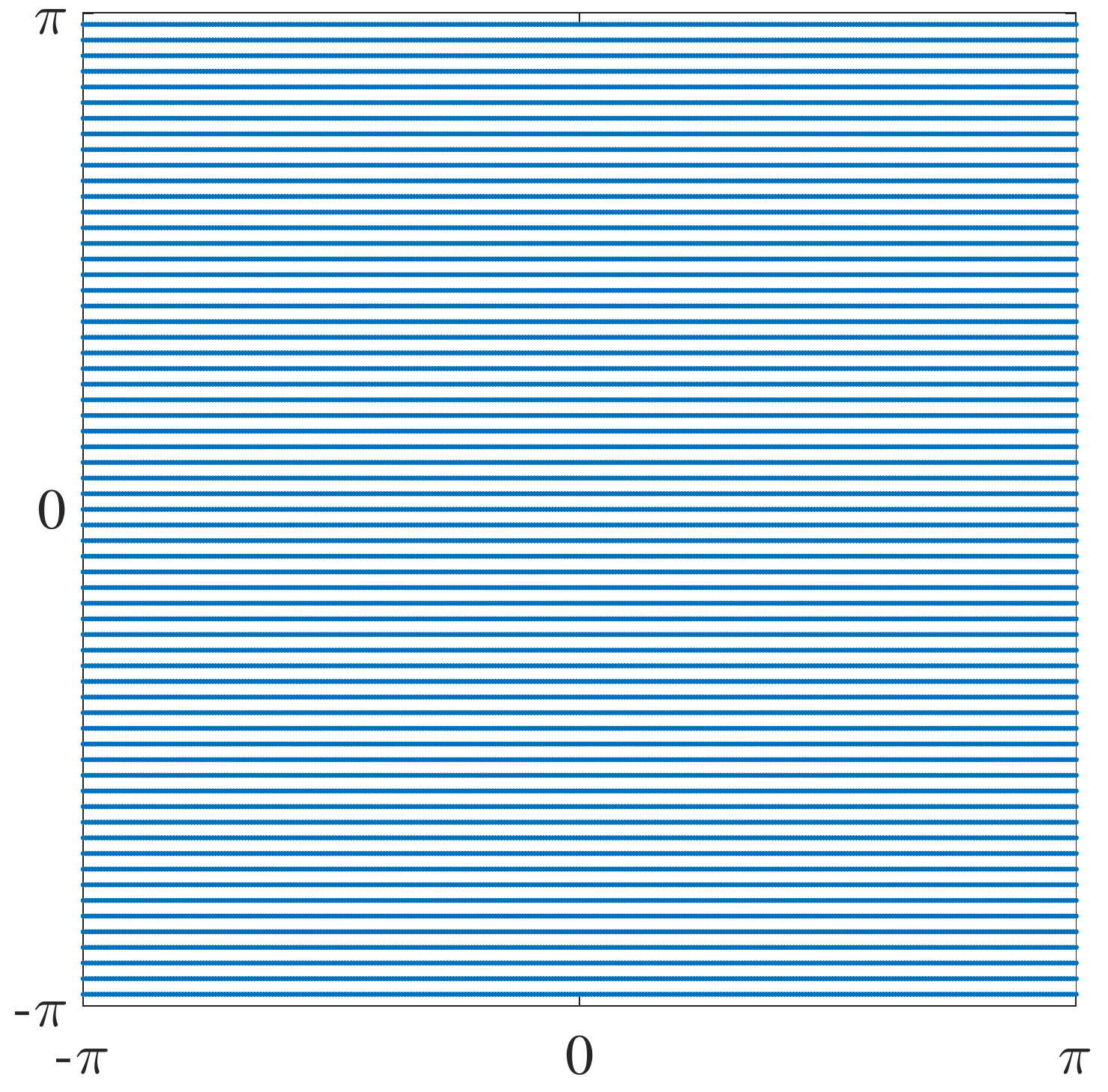}	\\
	(a)	&	(b)		
\end{tabular}
\end{center}
\caption{\label{Fig:MRI:description} \small
Magnetic resonance imaging problem. 
(a)~Original simulated image of a brain. 
(c)~Normalized continuous Fourier space ($k$-space) showing the frequencies selected to obtain $y$, using Cartesian trajectories.
}
\end{figure}

\subsection{Magnetic resonance imaging}
\label{ssec:MRmed}

\subsubsection{Problem description}
\label{Sssec:MRI}

Magnetic resonance imaging is a non-invasive non-ionising medical imaging technique that finds its superiority in the flexibility of its contrast mechanisms. It comes in various modalities ranging from high resolution structural imaging aiming at mapping detailed tissue structures, or high angular resolution diffusion imaging mapping the structural neuronal connectivity by probing molecular diffusion in each voxel of the brain, to dynamic imaging mapping for example the heart dynamics through time. Data acquisition is intrinsically long, sometimes prohibitively, as it relies on sequential measurement of Fourier samples of the image under scrutiny, which can again be of gigapixel dimension. Fast high-resolution imaging constitutes a deep challenge for medical research, which can be addressed by the combination of two acceleration strategies: firstly, the use of multiple acquisition coils, and secondly the acquisition of an incomplete Fourier coverage. This approach gives rise to a large-scale ill-posed inverse problems for the recovery of structural, diffusion of dynamic images under scrutiny. Once more, not only image estimation but also associated uncertainty quantification methodologies, key to the diagnosis and subsequent treatment of potential pathologies, must scale unprecedented dimension.

In this context, an unknown image $\overline{x} \in \R^N$ is observed simultaneously through $n_c$ receiver coils. An example of a simulated brain image, with $N=256 \times 256$, is shown in Figure~\ref{Fig:MRI:description}(a), generated from the magnetic resonance imaging toolbox available at \url{http://bigwww.epfl.ch/algorithms/mri-reconstruction/}. 
Each coil, indexed by $c \in \{1, \ldots , n_c\}$, acquires noisy incomplete Fourier measurements $y_c \in \C^{\widetilde{M}}$ of an image consisting of a multiplication of the unknown image under scrutiny and the spatial sensitivity profile of the coil. More formally, for each receiver coil $c \in \{1, \ldots , n_c\}$, the observation measurements are given by $y_c = \Phi_c x + w_c$, 
where $\Phi_c \in \C^{\widetilde{M} \times N}$ represents the the Fourier sampling operator and $w_c \in \C^{\widetilde{M}}$ is a realization of an additive complex i.i.d. Gaussian noise with zero mean and variance equal to $\sigma^2=0.01$, for both the real and imaginary parts of the noise. 
The global measurements $y \in \C^{M}$ corresponds then to the concatenation of all the coil observations $(y_c)_{1 \le c \le n_c}$, with $M = n_c \widetilde{M}$. 
In our simulations, we will consider measurements acquired from $n_c=4$ receiver coils. 
In magnetic resonance imaging, the Fourier domain (also called $k$-space) can be sampled following different trajectories. In our simulations we use two different undersampling strategies. Firstly, we use the same random sampling as for radio astronomy imaging, described in Section~\ref{Sssec:RI}, considering several sampling ratio values $\widetilde{M}/N \in \{0.1, 0.2\}$. 
Secondly, we use a more realistic random sampling generated using the magnetic resonance imaging toolbox available at \url{http://bigwww.epfl.ch/algorithms/mri-reconstruction/}, consisting of the continuous Fourier Cartesian trajectories displayed in Figure~\ref{Fig:MRI:description}(b). This Fourier sampling selects $\widetilde{M} = 21248$ frequencies, corresponding to under-sampling factors along frequency and phase encoding direction equal to $1/1.3$ and $4$, respectively.

In both the considered simulation settings, we focus on spatially localized structures  corresponding to the definition of $\Sc$ given by Definition~\ref{example:structure}. 
This set $\Sc$ is characterized by $\Lc$, built as in described in Section~\ref{Sssec:RI}, to model a smoothing operator using 2D Gaussian convolution kernels of sizes $3 \times 3$, $7 \times 7$ and $11 \times 11$, and we choose $\tau = \operatorname{std}\big( \Mc(x^\dagger) - \Lc(\Mc(x^\dagger)) \big)$. 
In addition, to define set $\Sc_3$, we choose $b = \big\| \Lc \big( \Mc \big( x^\dagger \big)  \big) \big\|_2$ and $\theta = \vartheta \big\| \Lc \big( \Mc ( x^\dagger )  \big) \big\|_2$, with $\vartheta = 10^{-2}$.

\subsubsection{Uncertainty quantification in magnetic resonance: random sampling}
\label{ssec:UQ_MR:gauss}

\begin{figure}[t]
\begin{center}
\begin{tabular}{@{}c@{}c@{}c@{}}
	\includegraphics[height=3.8cm]{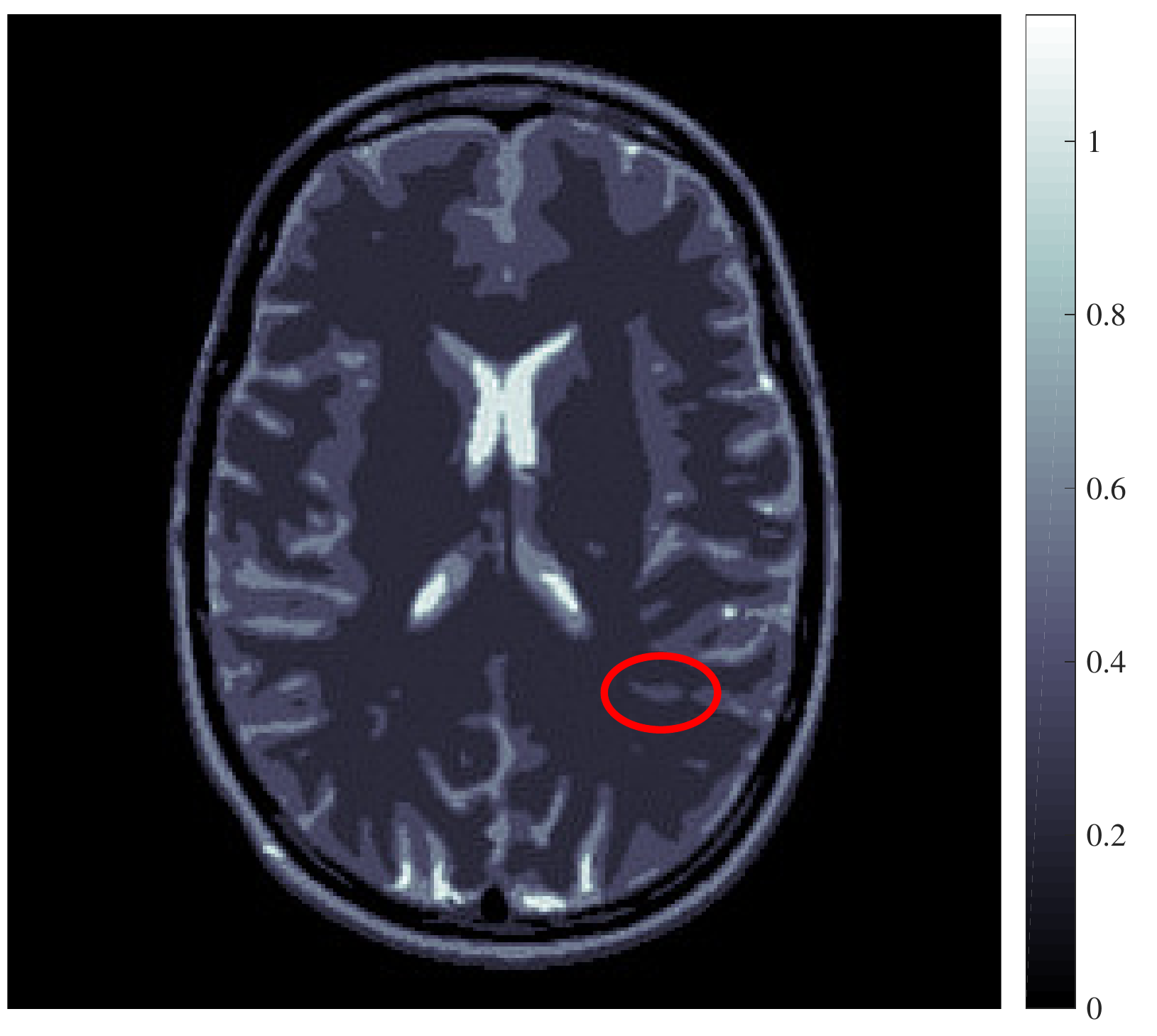}
&	\includegraphics[height=3.8cm]{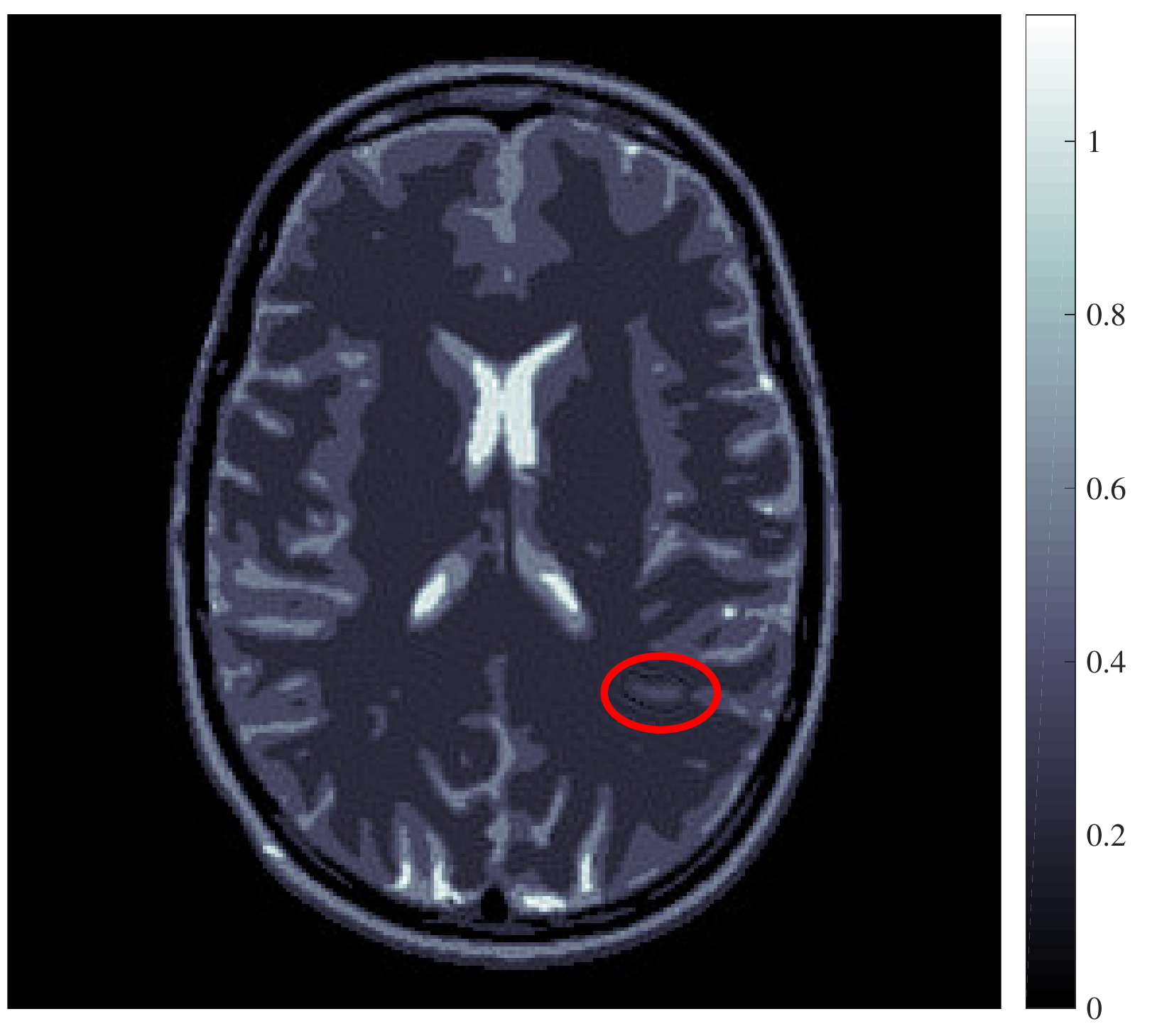}
&	\includegraphics[height=3.8cm]{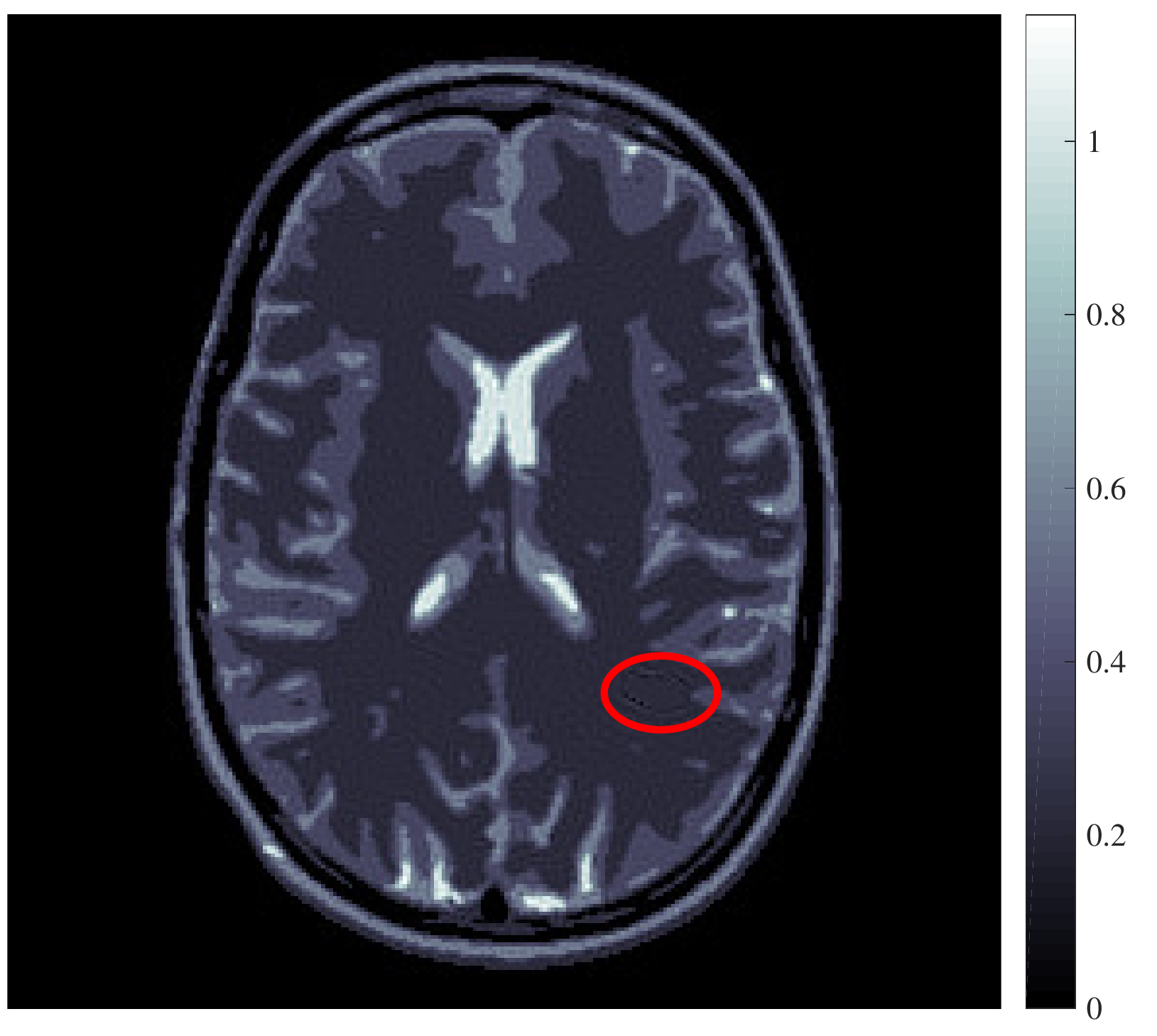}	\\
	\includegraphics[height=3.8cm]{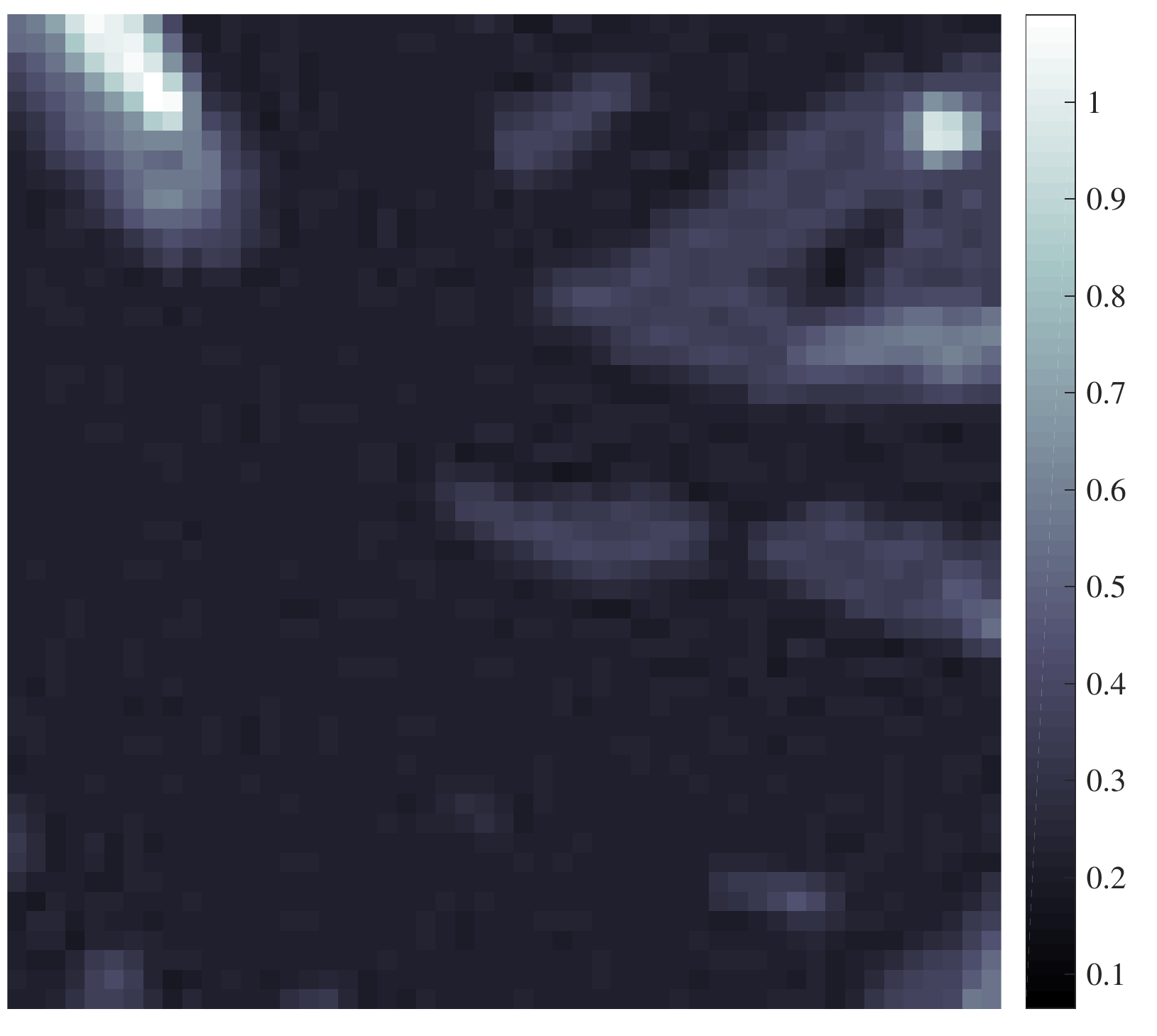}
&	\includegraphics[height=3.8cm]{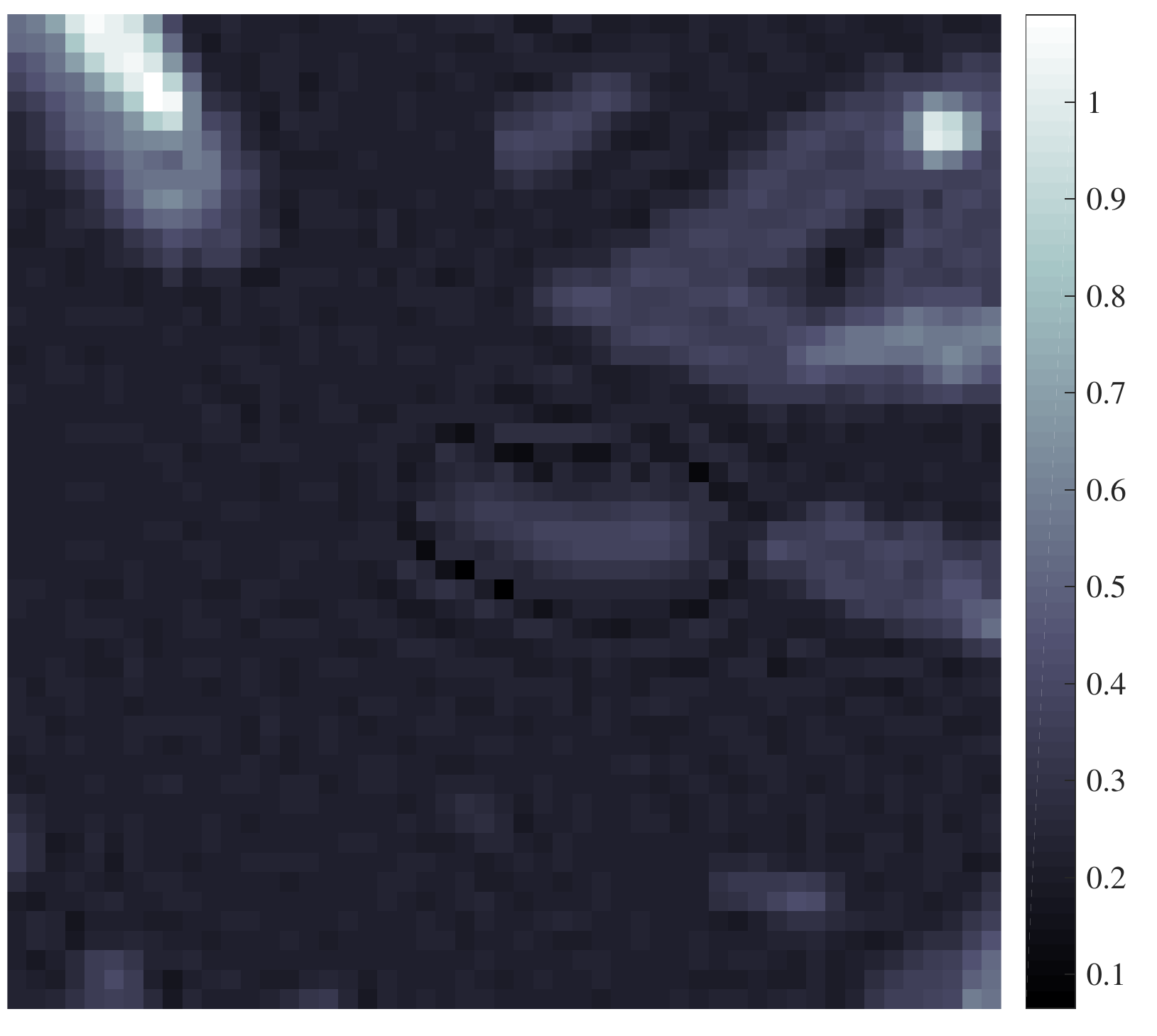}	
&	\includegraphics[height=3.8cm]{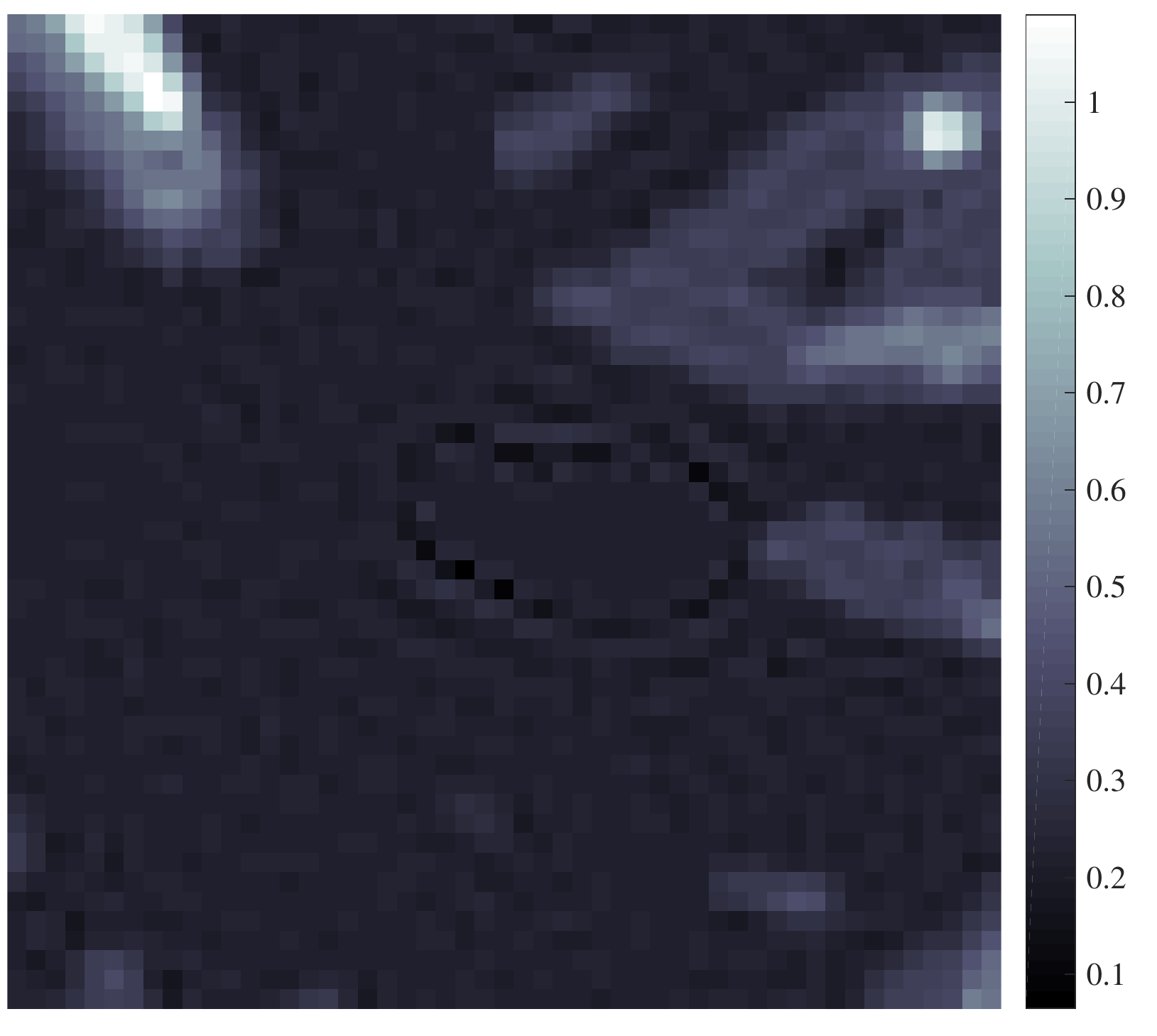}	
\end{tabular}
\end{center}

\vspace*{-0.3cm}

\caption{\label{Fig:MRI:struct1} \small
Simulation results for the magnetic resonance imaging problem with random sampling. Uncertainty quantification of Structure~1, in the case when $M/N = 0.2$ and $\sigma^2=0.01$. In this context, $\rho_\alpha = 96.86\%$ of the intensity's structure is confirmed at $99\%$, and $H_0$ is rejected with significance $1\%$.
Top row: images in linear scale with Structure~1 highlighted in red with, from left to right: 
$x^\dagger$, 
$x^\ddagger_{\widetilde{\Cc}_\alpha}$, 
and $x^\ddagger_{\Sc}$. 
Bottom row: zoomed images in linear scale on the area of Structure~1, corresponding to the images displayed in first row. The scale in the zoomed images is adapted to better emphasize Structure~1.
}
\end{figure}

\begin{figure}[t]
\begin{center}
\begin{tabular}{@{}c@{}c@{}c@{}}
	\includegraphics[height=3.8cm]{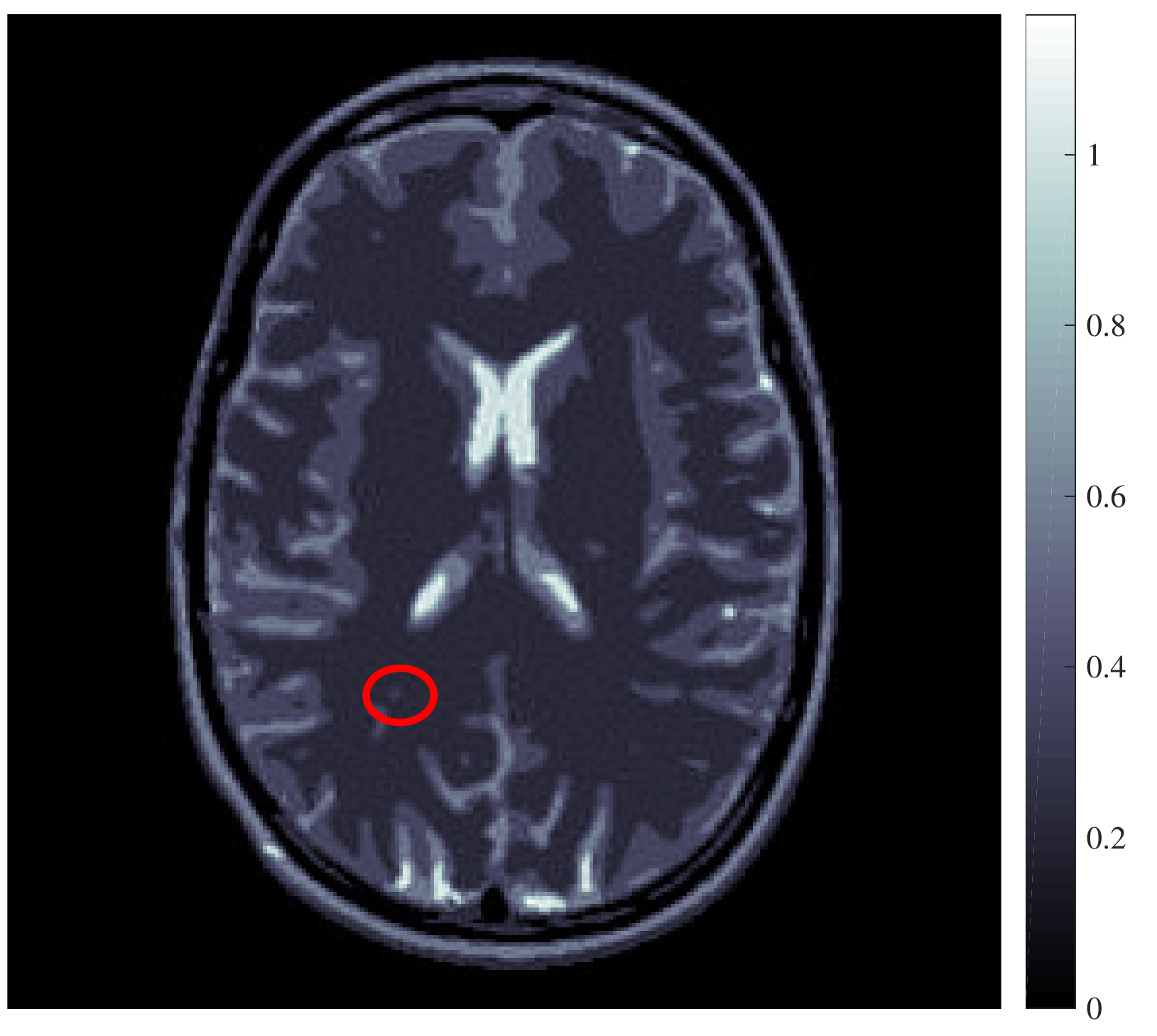}
&	\includegraphics[height=3.8cm]{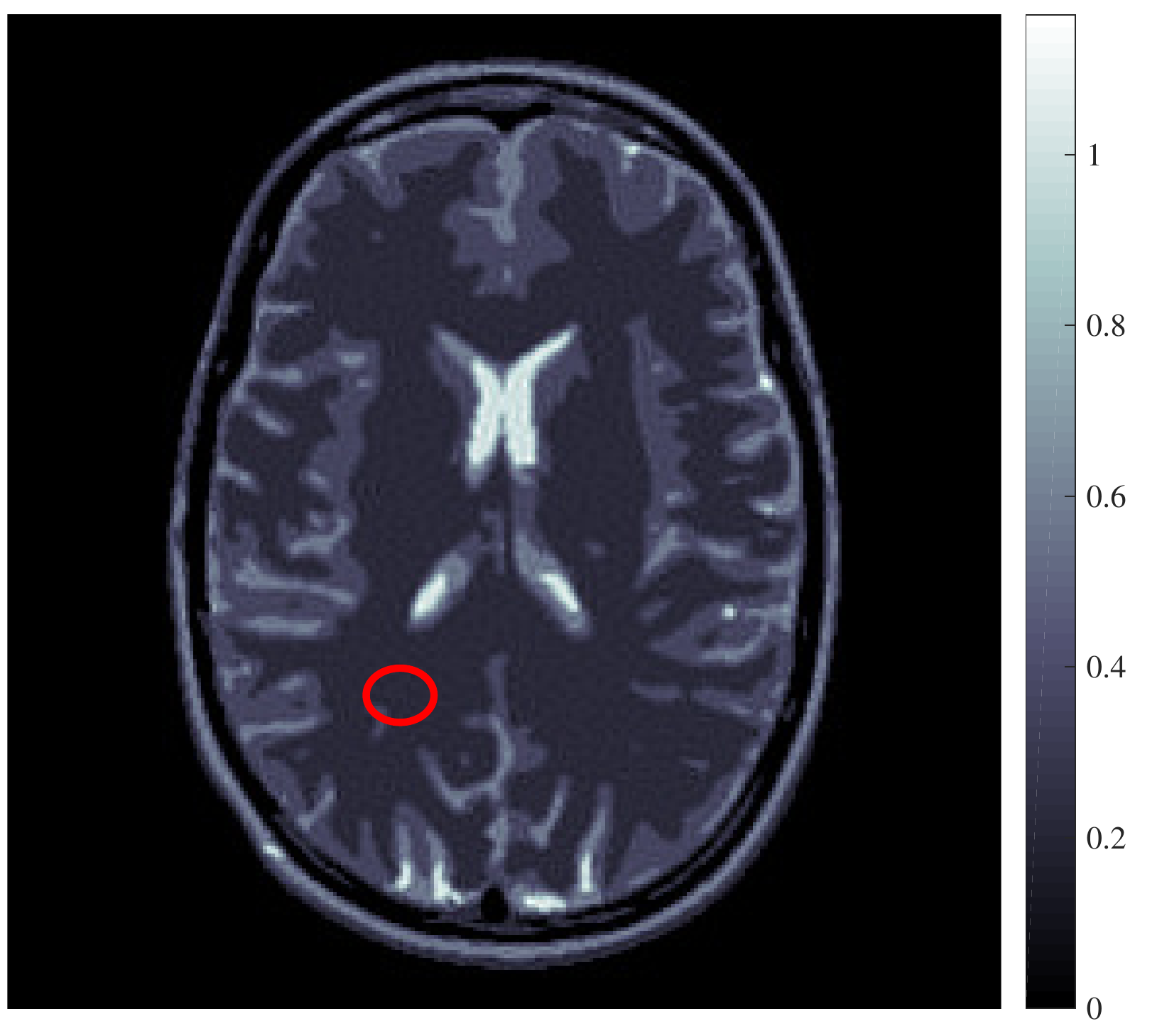}
&	\includegraphics[height=3.8cm]{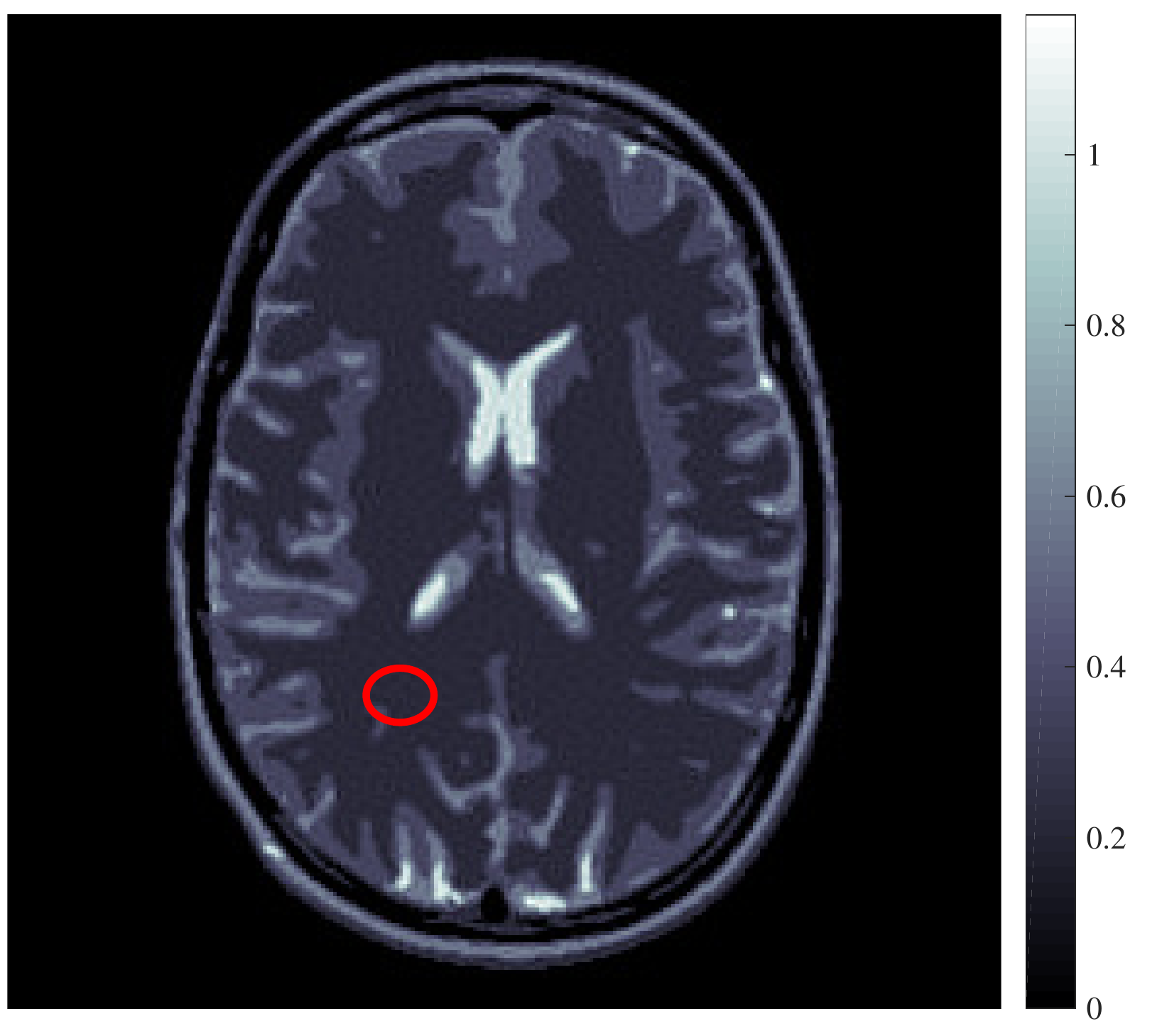}	\\
	\includegraphics[height=3.8cm]{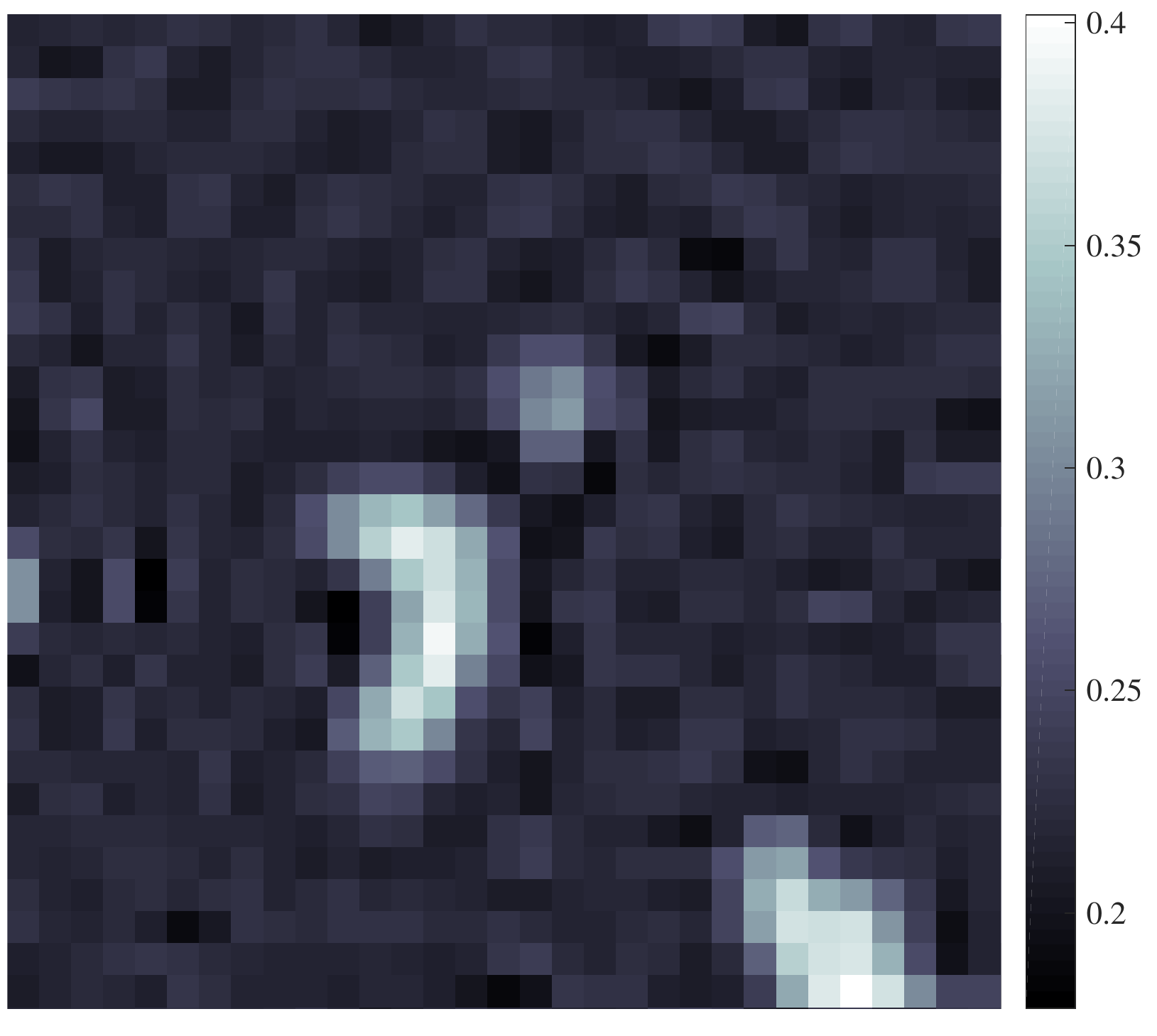}
&	\includegraphics[height=3.8cm]{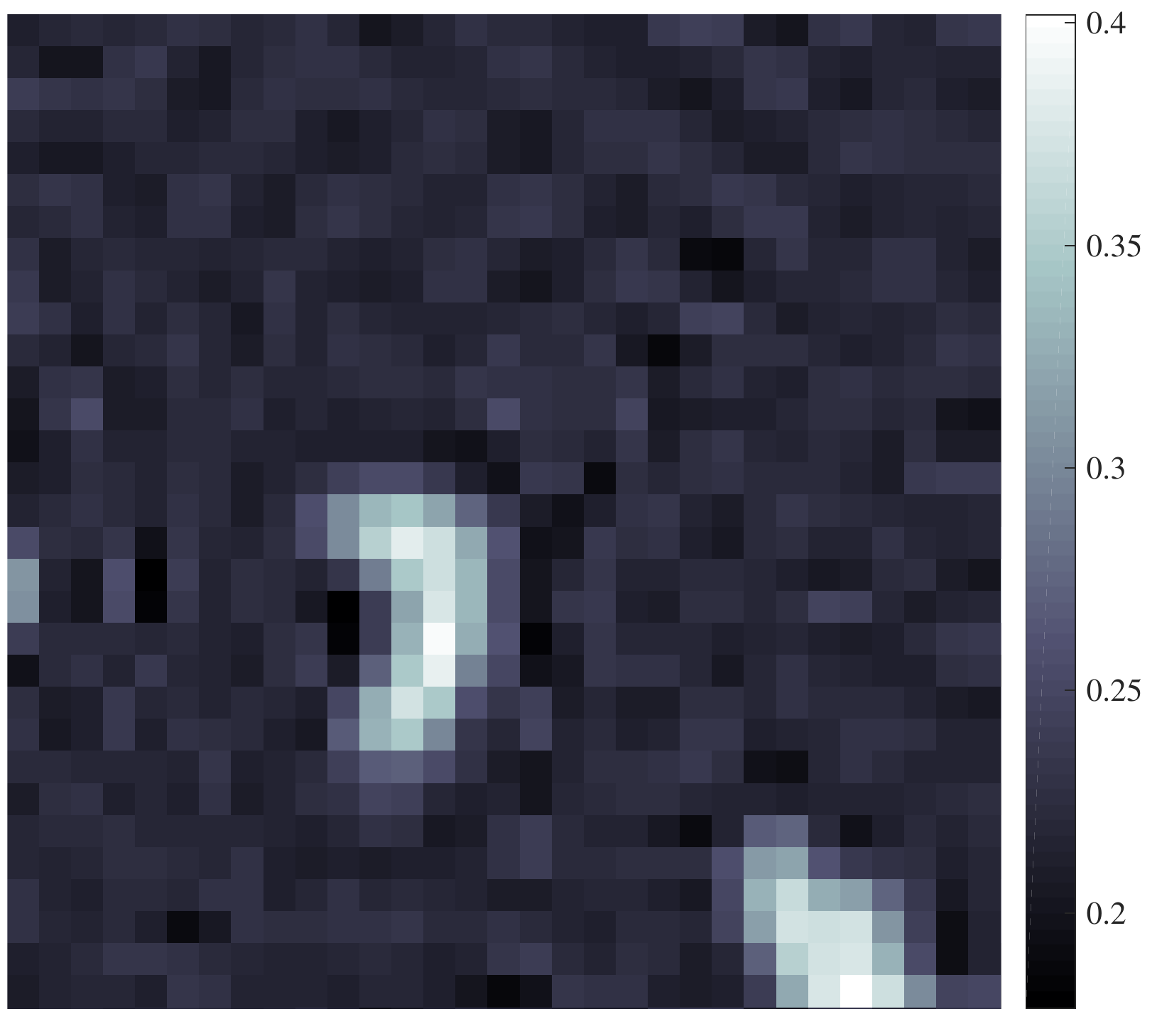}	
&	\includegraphics[height=3.8cm]{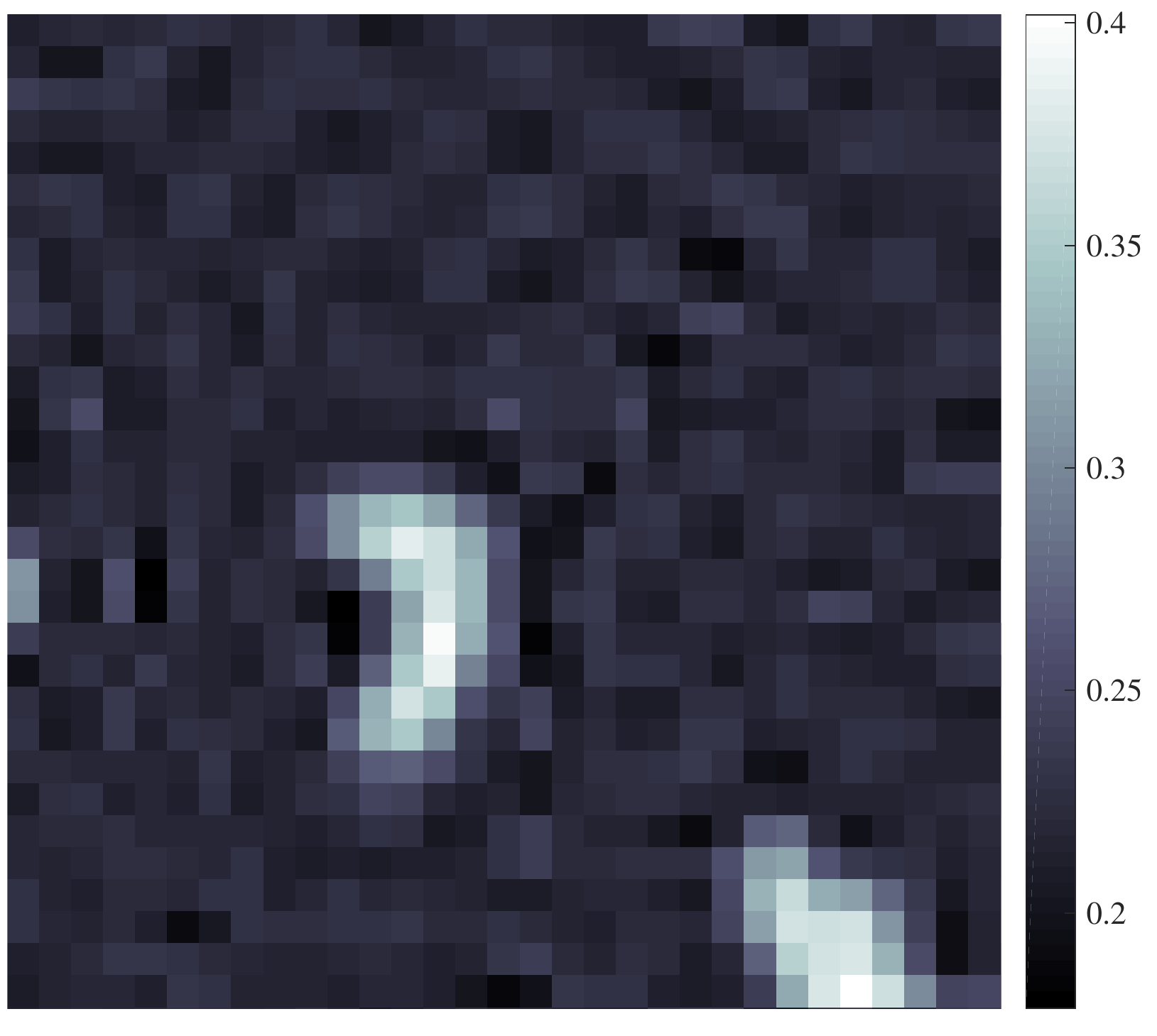}	
\end{tabular}
\end{center}

\vspace*{-0.3cm}

\caption{\label{Fig:MRI:struct2} \small
Simulation results for the magnetic resonance imaging problem with random sampling. Uncertainty quantification of Structure~2, in the case when $M/N = 0.1$ and $\sigma^2=0.03$. In this context, $\rho_\alpha = 0.62\%$ and $H_0$ cannot be rejected.
Top row: images in linear scale with Structure~2 highlighted in red with, from left to right: 
$x^\dagger$, 
$x^\ddagger_{\widetilde{\Cc}_\alpha}$, 
and $x^\ddagger_{\Sc}$. 
Bottom row: zoomed images in linear scale on the area of Structure~2, corresponding to the images displayed in first row. The scale in the zoomed images is adapted to better emphasize Structure~2.
}
\end{figure}

In this section, we present the results obtained for the simulations on the magnetic resonance imaging problem, considering a random sampling, for $\widetilde{M}/N \in \{0.1, 0.2\}$ and $\sigma^2 \in \{0.01, 0.02, 0.03\}$. Note that, since we consider four receiver coils, in total we have $M = 4 \widetilde{M}$ measurements. 
We aim to quantify the uncertainty of the two structures, namely Structure~1 and Structure~2, highlighted in red in Figures~\ref{Fig:MRI:struct1} and \ref{Fig:MRI:struct2}, respectively.

Fig.~\ref{Fig:MRI:struct1} presents the experimental results obtained considering \linebreak $(\widetilde{M}/N, \sigma^2) = (0.2, 0.01)$. In the top-row of Fig.~\ref{Fig:MRI:struct1}, we show, from left to right, the MAP estimate $x^\dagger$ and the results from Algorithm~\ref{algo:POCS_gen}, $x^\ddagger_{\widetilde{\Cc}_\alpha}$ and $x^\ddagger_{\Sc}$. In these images, Structure~1 is highlighted in red. The corresponding images, zoomed in the area of Structure~1, are displayed in the bottom-row of Fig.~\ref{Fig:MRI:struct1}.
For this example, the structure's confirmed intensity percentage is equal to $\rho_\alpha =96.86\%$. Therefore, we conclude that $\widetilde{\Cc}_\alpha \cap \Sc = \emp$, and consequently that $H_0$ is rejected with significance $\alpha=1\%$. 

In Fig.~\ref{Fig:MRI:struct2} are presented the simulation results obtained by considering \linebreak $(\widetilde{M}/N, \sigma^2) = (0.1, 0.03)$. Similarly to Fig.~\ref{Fig:MRI:struct1}, the first row shows the images $x^\dagger$, $x^\ddagger_{\widetilde{\Cc}_\alpha}$ and $x^\ddagger_{\Sc}$, and the second row shows the associated zoomed images for the area of Structure~2. For this experiment, we have $\rho_\alpha = 0.62\% \approx 0 \%$ and we can conclude that $\widetilde{\Cc}_\alpha \cap \Sc \neq \emp$. Consequently, $H_0$ cannot be rejected, and Structure~2 is highly uncertain.

The values of $\rho_\alpha$, in percentage, for the two structures of interests, for the different experimental settings, are provided in Table~\ref{Tab:MRI:results_Gauss}. 
According to Table~\ref{Tab:MRI:results_Gauss}, between $59.09\%$ and $96.86\%$ of Structure~1 is confirmed at $99\%$, depending on the values of $\widetilde{M}/N$ and $\sigma^2$. Therefore, for Structure~1, for all the considered values of $(\widetilde{M}/N, \sigma^2)$, $\widetilde{\Cc}_\alpha \cap \Sc = \emp$ and $H_0$ is rejected. 
Concerning Structure~2, $\rho_\alpha$ ranges between $0.62\%$ and $11.31\%$, for $\big(\widetilde{M}/N, \sigma^2\big) = (0.1, 0.03)$ and $\big(\widetilde{M}/N, \sigma^2\big) = (0.2, 0.01)$, respectively. In particular, higher is the ratio $\widetilde{M}/N$ and higher is $\rho_\alpha$. At the opposite, smaller is $\sigma^2$ and higher is $\rho_\alpha$. For this structure, the conclusion is different depending on the choice of $\big(\widetilde{M}/N, \sigma^2\big)$. For instance, let consider that$\widetilde{\Cc}_\alpha \cap \Sc \neq \emp$ when $\rho_\alpha < 3\%$. In this context, for $\big(\widetilde{M}/N, \sigma^2\big) = (0.2, 0.01)$ (resp. $\big(\widetilde{M}/N, \sigma^2\big) = (0.2, 0.02)$), the null hypothesis $H_0$ is rejected, and $11.31\%$ (resp. $3.18\%$) of Structure~2 is confirmed at $99\%$. For all the other choices of $\big(\widetilde{M}/N, \sigma^2\big) $, the null hypothesis $H_0$ cannot be rejected.

\begin{table}
\begin{center}
\begin{tabular}{c c}
{\renewcommand{\arraystretch}{1.5}
\begin{tabular}{|c|r|r|r|r|}
\multicolumn{5}{c}{}	\\
\cline{3-5}
\multicolumn{2}{c}{}
&	\multicolumn{3}{|c|}{$\sigma^2$}	\\
\cline{3-5}
\multicolumn{2}{c|}{}
&	0.01	&	0.02	&	0.03	\\
\hline
\multirow{2}{*}{$\dfrac{\widetilde{M}}{N}$}
&	0.1
	&	80.17	&	64.95	&	59.09	\\
\cline{2-5}
&	0.2
	&	96.86	&	74.77	&	70.80	\\
\hline
\multicolumn{5}{c}{Structure~1}	
\end{tabular}}
&\hspace*{0.2cm}
{\renewcommand{\arraystretch}{1.5}
\begin{tabular}{|c|r|r|r|r|}
\multicolumn{5}{c}{}	\\
\cline{3-5}
\multicolumn{2}{c}{}
&	\multicolumn{3}{|c|}{$\sigma^2$}	\\
\cline{3-5}
\multicolumn{2}{c|}{}
&	0.01	&	0.02	&	0.03	\\
\hline
\multirow{2}{*}{$\dfrac{\widetilde{M}}{N}$}
&	0.1
	&	\phantom{0}2.49	&	\phantom{0}1.01	&	\phantom{0}0.62	\\
\cline{2-5}
&	0.2
	&	11.31	&	\phantom{0}3.18	&	\phantom{0}2.09	\\
\hline
\multicolumn{5}{c}{Structure~2}	
\end{tabular}}
\end{tabular}
\end{center}
\caption{\label{Tab:MRI:results_Gauss} \small
Values of $\rho_\alpha$ in percentage ($\%$) for the two structures of interest in the magnetic resonance imaging problem with random sampling.}
\end{table}

\subsubsection{Uncertainty quantification in magnetic resonance: Cartesian trajectories}
\label{ssec:UQ_MR:cartesian}

In this section are presented the simulation results obtained for the magnetic resonance imaging problem, considering the Cartesian trajectories given in Figure~\ref{Fig:MRI:description}(b). Due to the particular under-sampling obtained from these trajectories, the MAP estimate presents artefacts non-existing in the original image $\overline{x}$ (see Figure~\ref{Fig:MRI:description}(a)). The MAP estimate is shown on the first column of Figure~\ref{Fig:MRI:cartesian}, where two of the artefact are highlighted in red. Zoomed images are also provided (first column, rows 2 and 4) on the areas of these artefacts. 
We define these two artefact as structures using Definition~\ref{example:structure}, and we investigate their uncertainty. The results are displayed in Figure~\ref{Fig:MRI:cartesian}. 

The two first rows correspond to the uncertainty quantification results for the first artefact, at the center of the brain. The first row gives, from left to right, the MAP estimate $x^\dagger$ and the two results from the alternating projections, $x^\ddagger_{\widetilde{\Cc}_\alpha}$ and $x^\ddagger_{\Sc}$. For this simulation, we obtain $\rho_\alpha = 0.02\% \approx 0\%$. This result can be visually verified by observing that $x^\ddagger_{\widetilde{\Cc}_\alpha} \approx x^\ddagger_{\Sc}$. Consequently, we conclude that $H_0$ cannot be rejected, and that this first structure is highly uncertain, which is consistent with it being an artefact. 

The same observations can be done for the second artefact, at the top of the brain, shown in the last two rows of Figure~\ref{Fig:MRI:cartesian}. In this case we have $\rho_\alpha = 0.01\% \approx 0\%$. In this case, $H_0$ cannot be rejected, and we conclude that the structure defined by this second artefact is not confirmed.

\begin{figure}[h!]
\begin{center}
\begin{tabular}{@{}c@{}c@{}c@{}c@{}}
	\includegraphics[height=3.8cm]{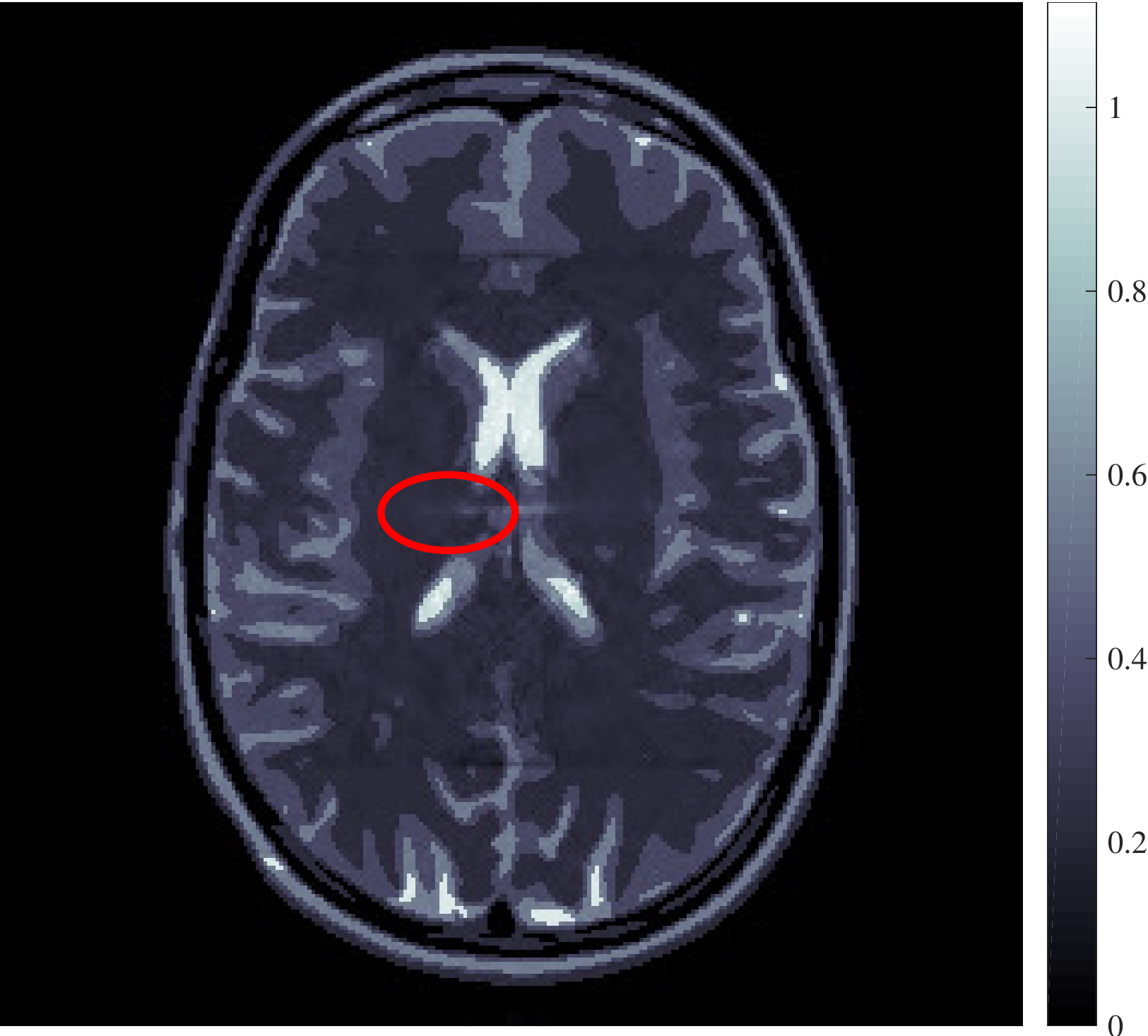}
&	\includegraphics[height=3.8cm]{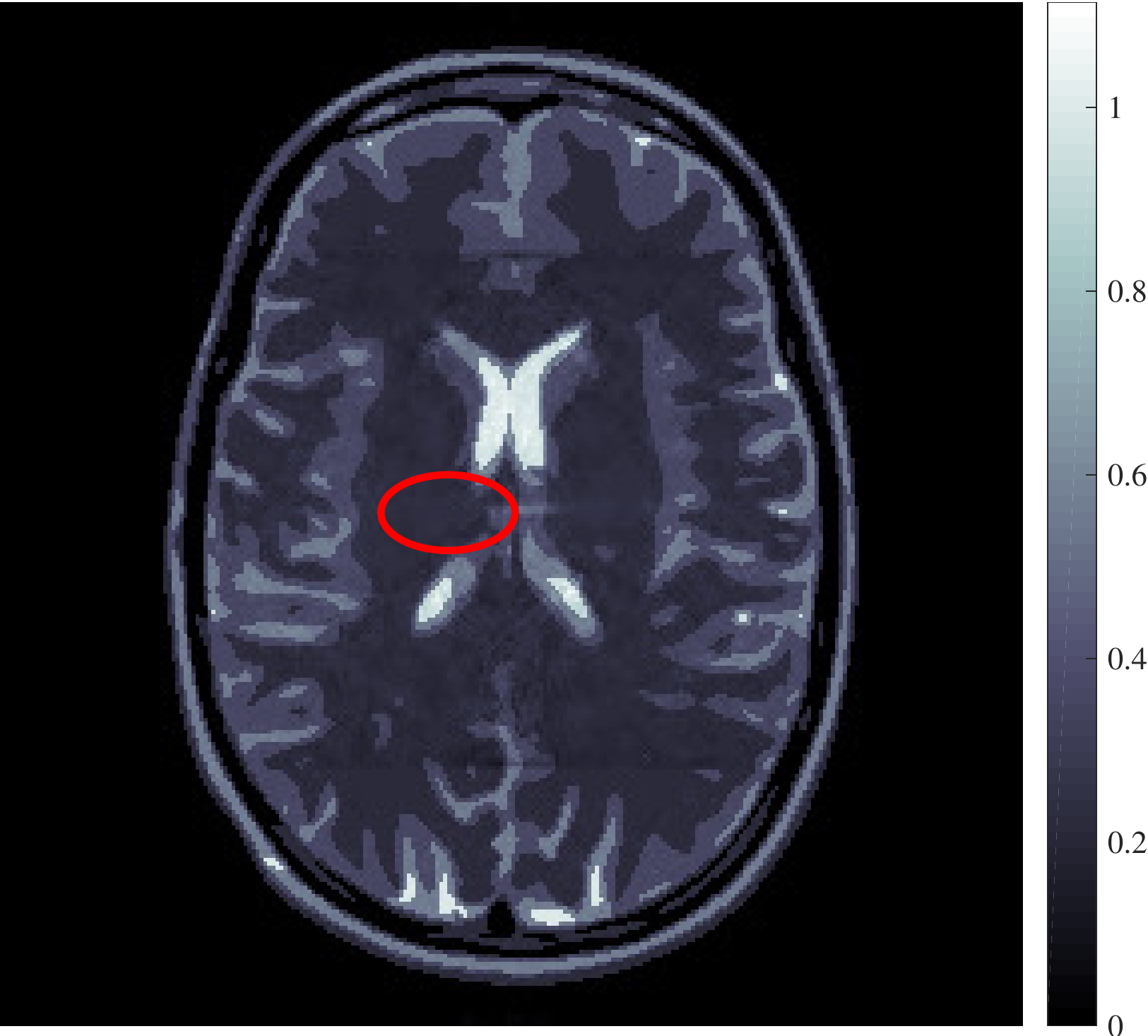}
&	\includegraphics[height=3.8cm]{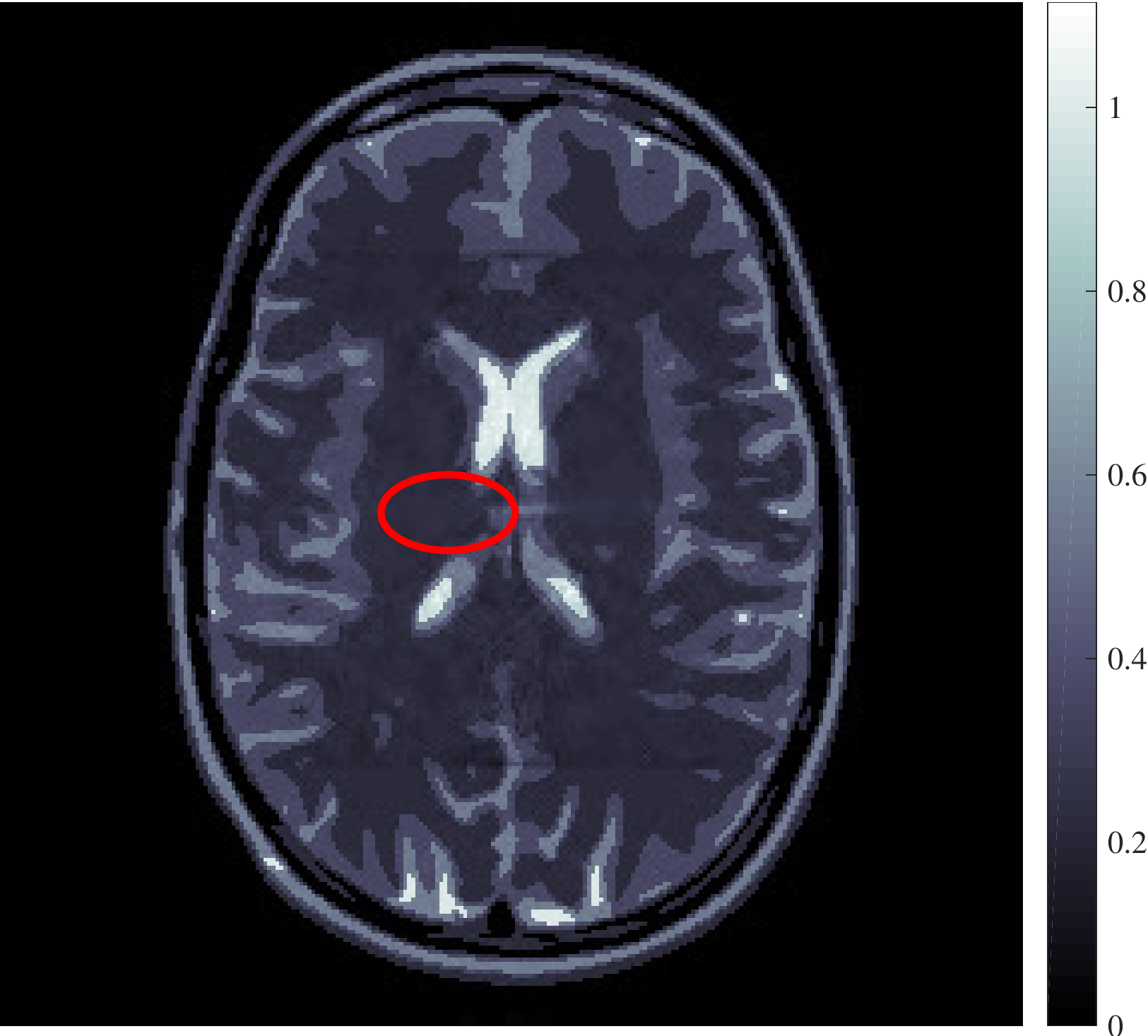}	\\
	\includegraphics[height=3.8cm]{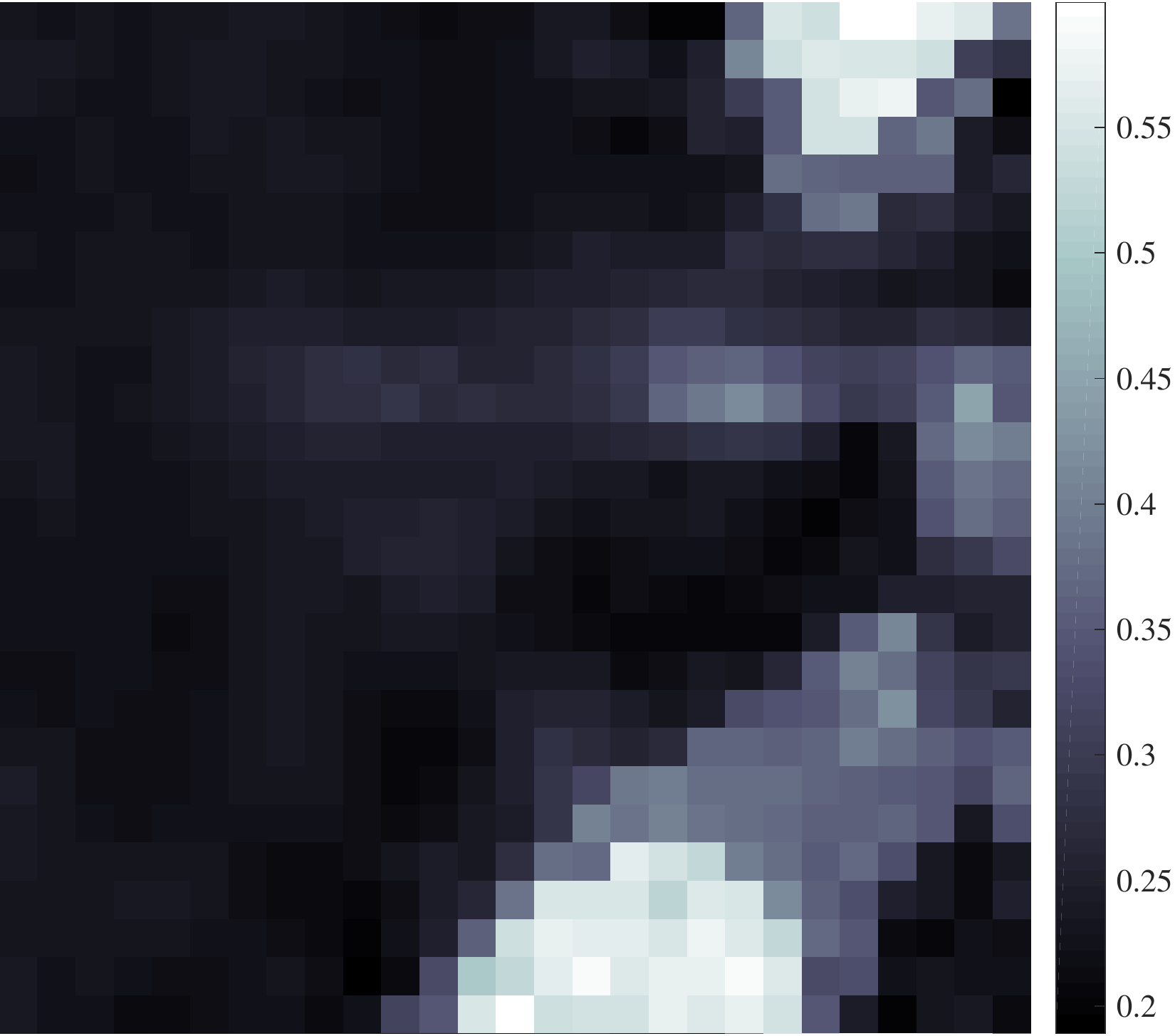}
&	\includegraphics[height=3.8cm]{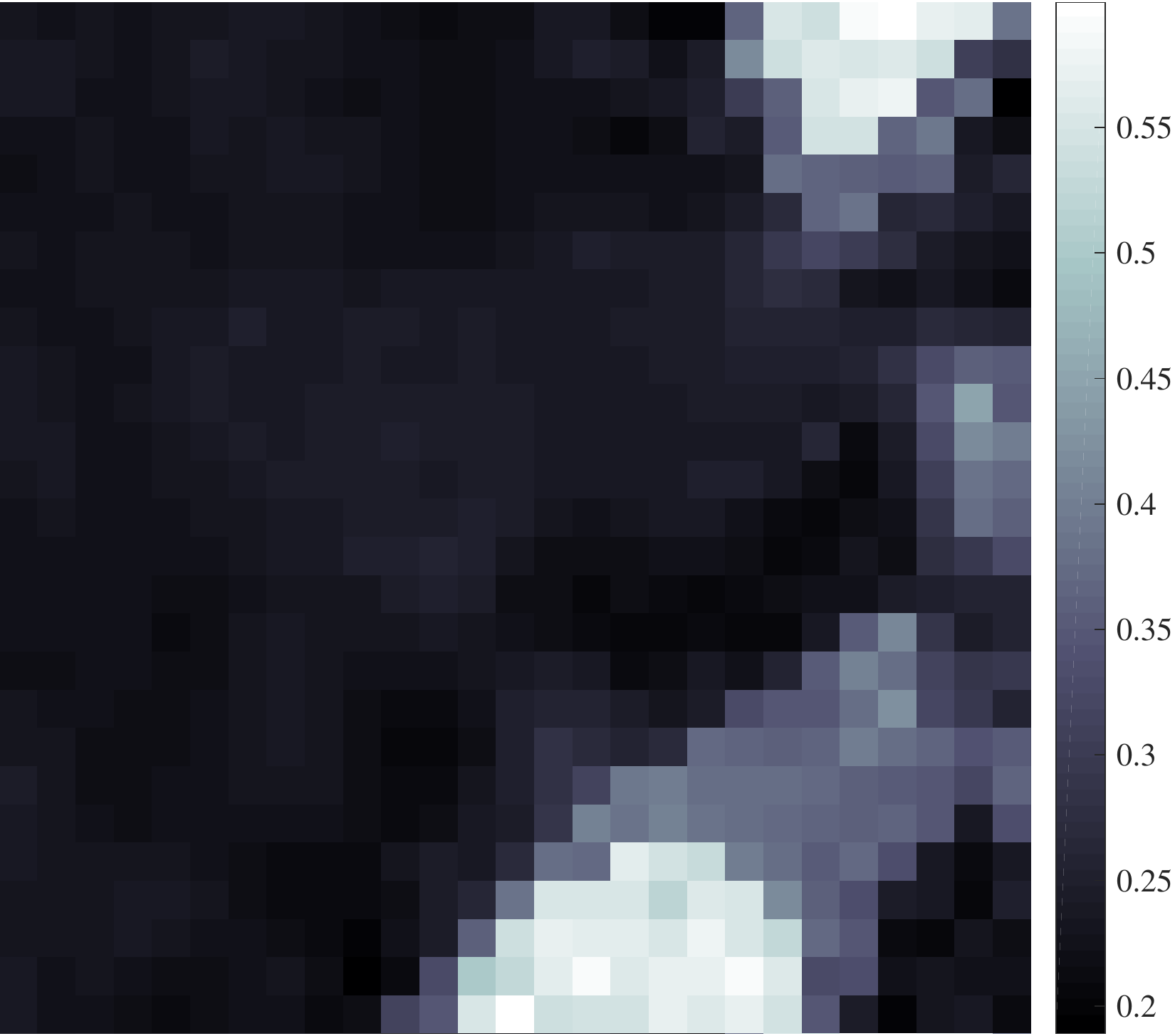}	
&	\includegraphics[height=3.8cm]{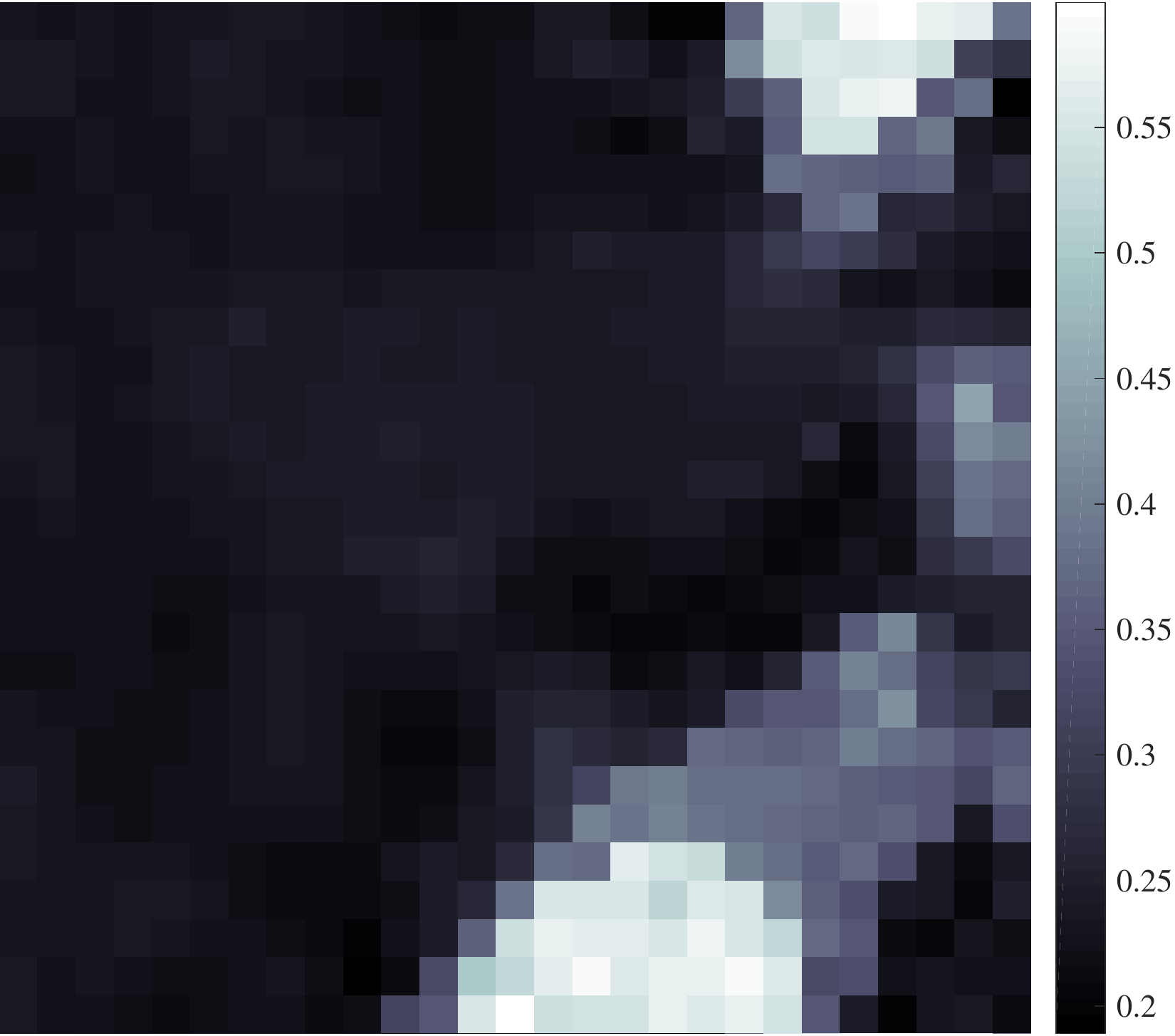}	\\[0.2cm]
	\includegraphics[height=3.8cm]{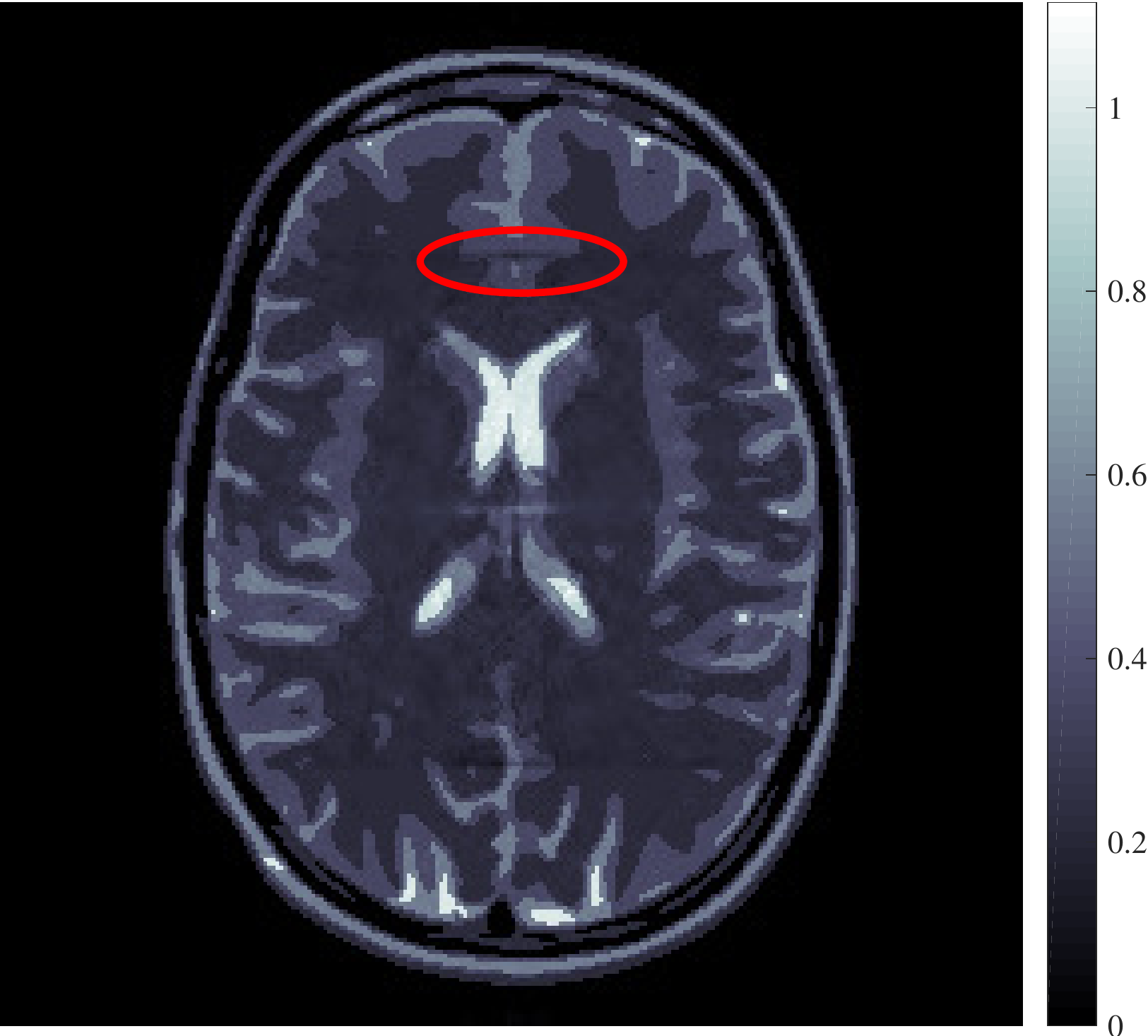}
&	\includegraphics[height=3.8cm]{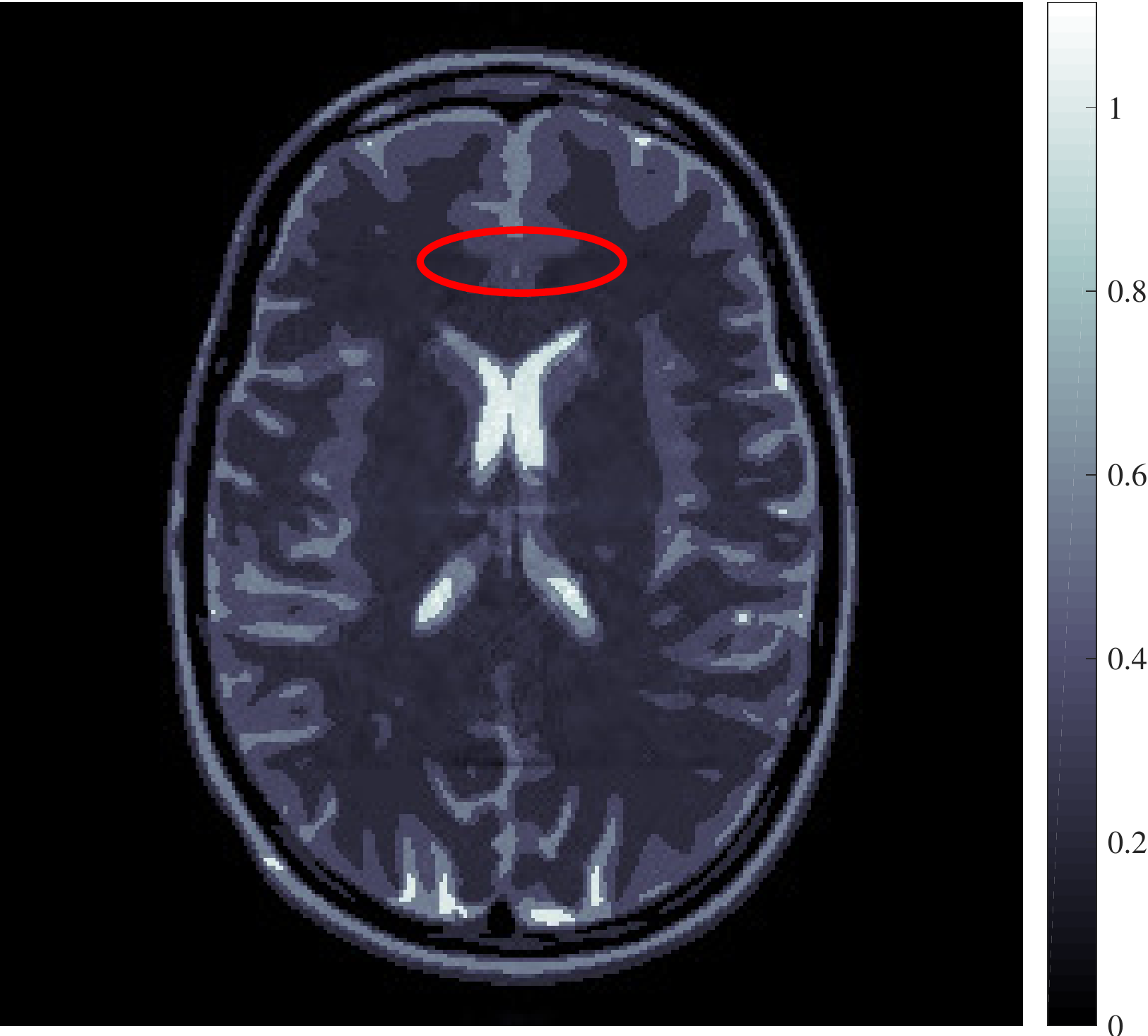}
&	\includegraphics[height=3.8cm]{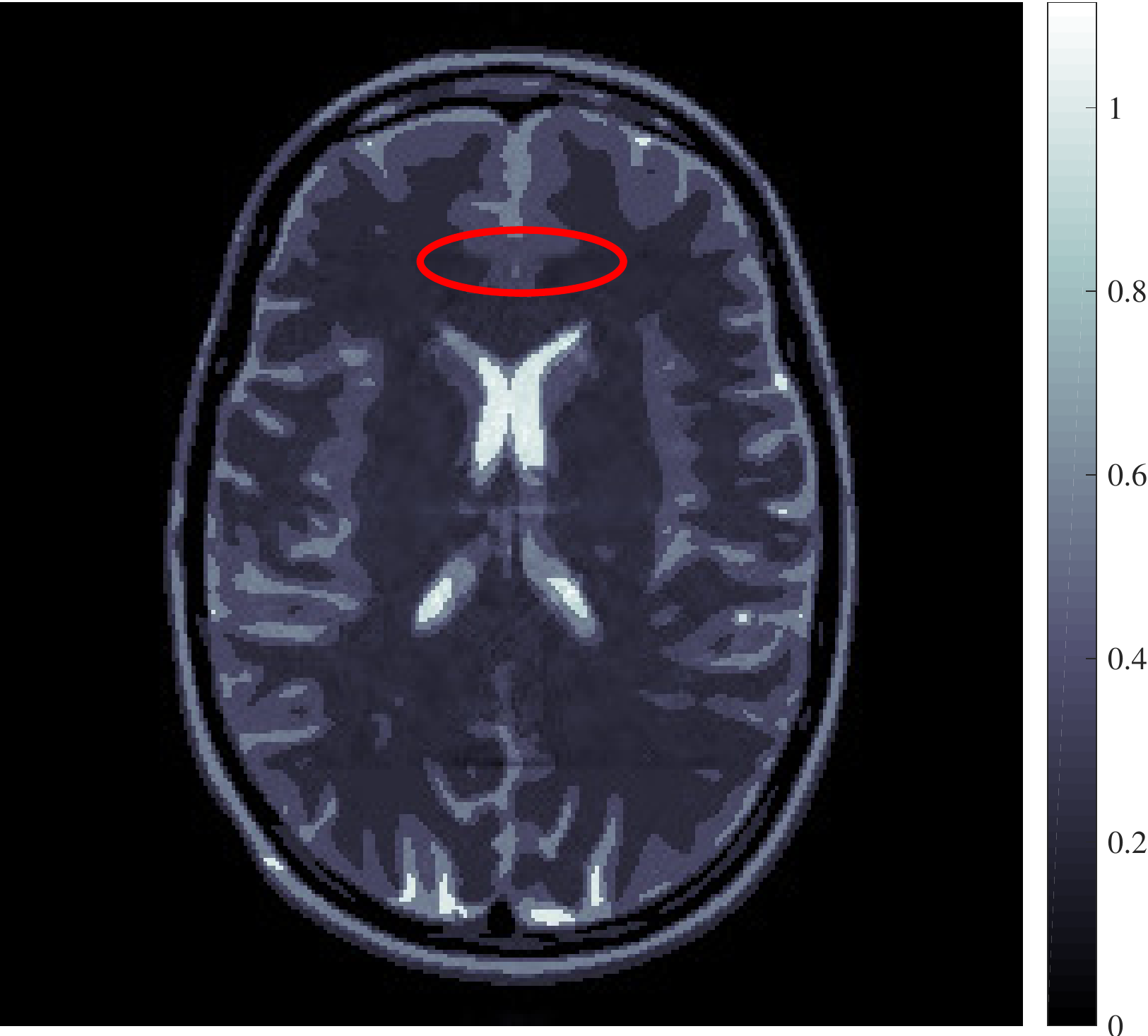}	\\
	\includegraphics[height=3.8cm]{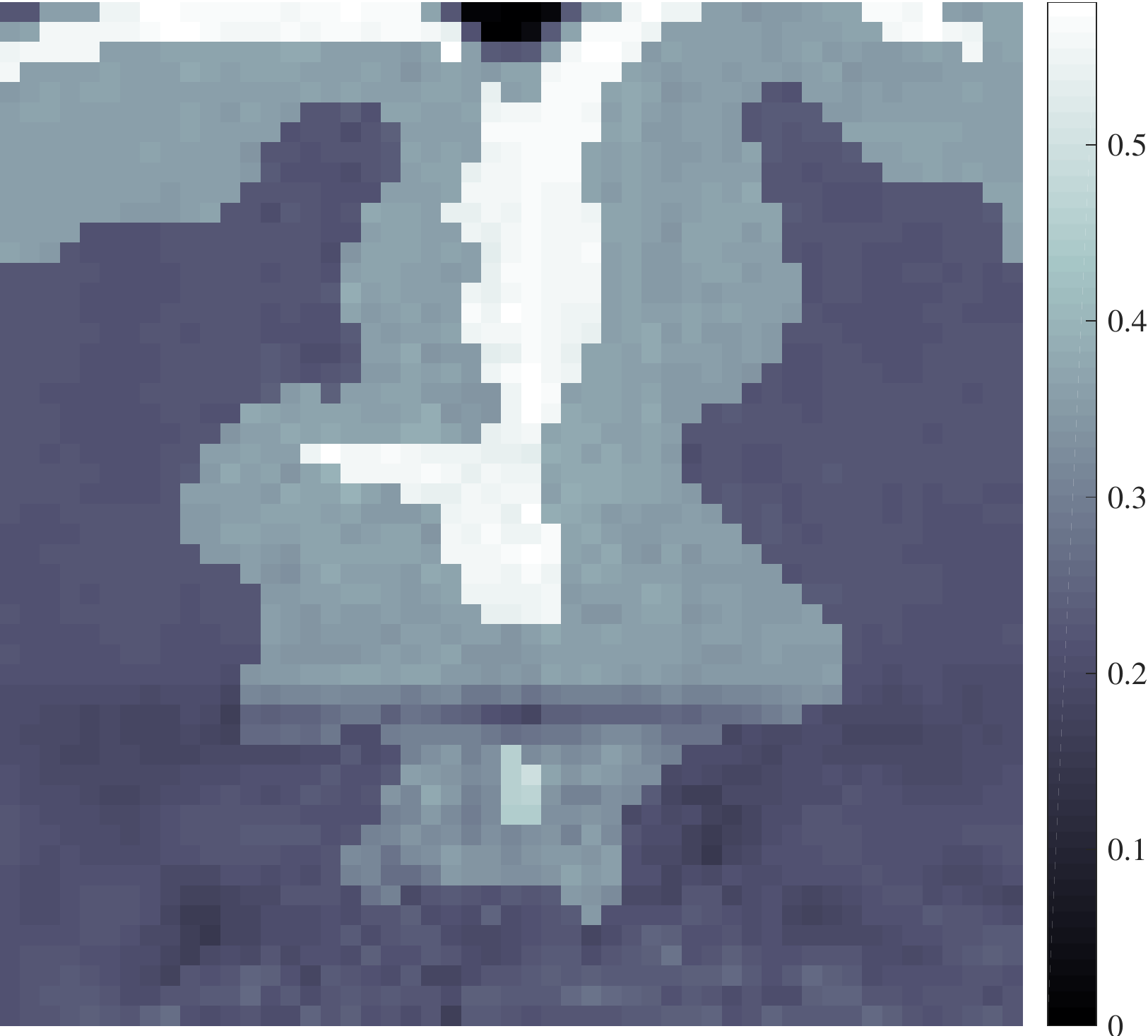}
&	\includegraphics[height=3.8cm]{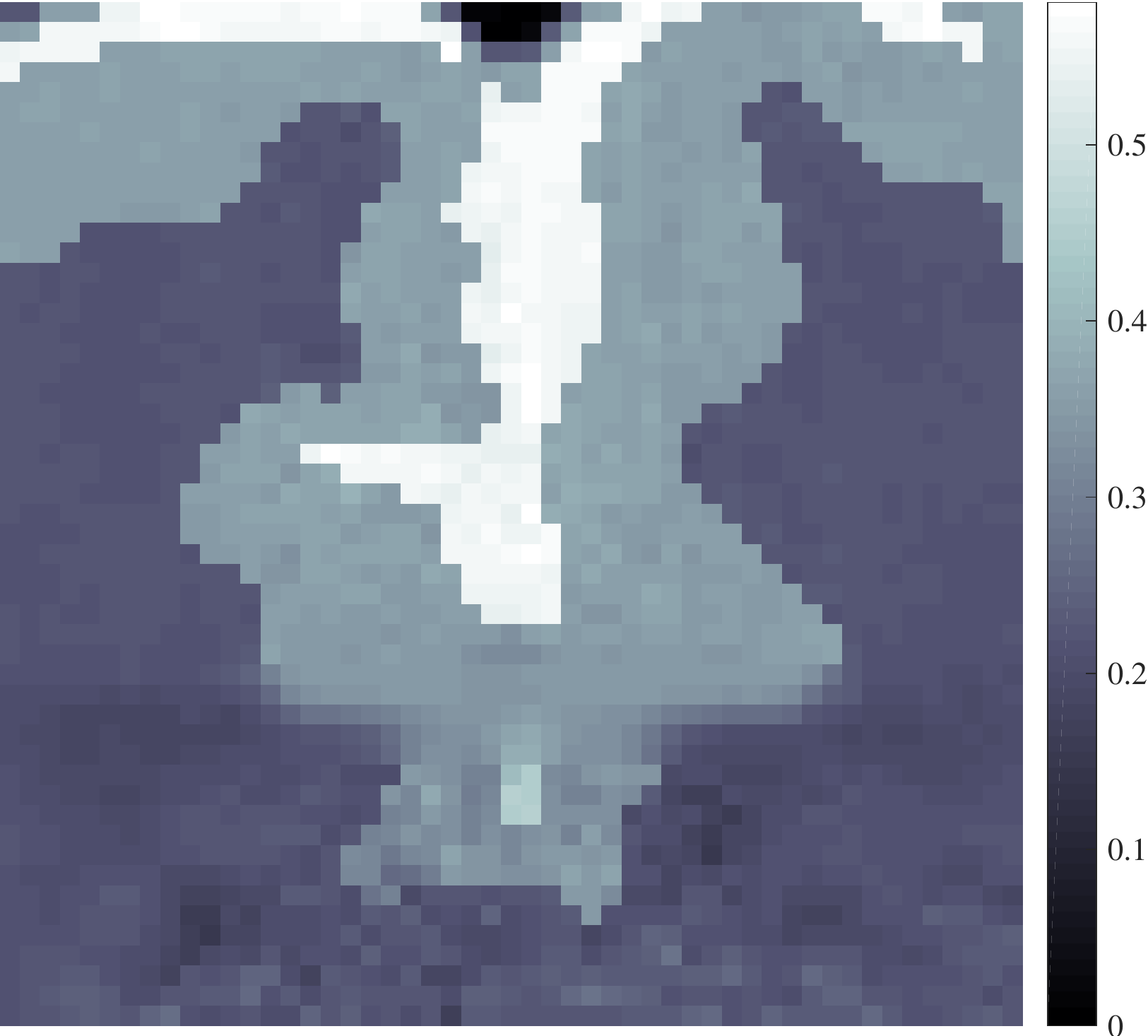}	
&	\includegraphics[height=3.8cm]{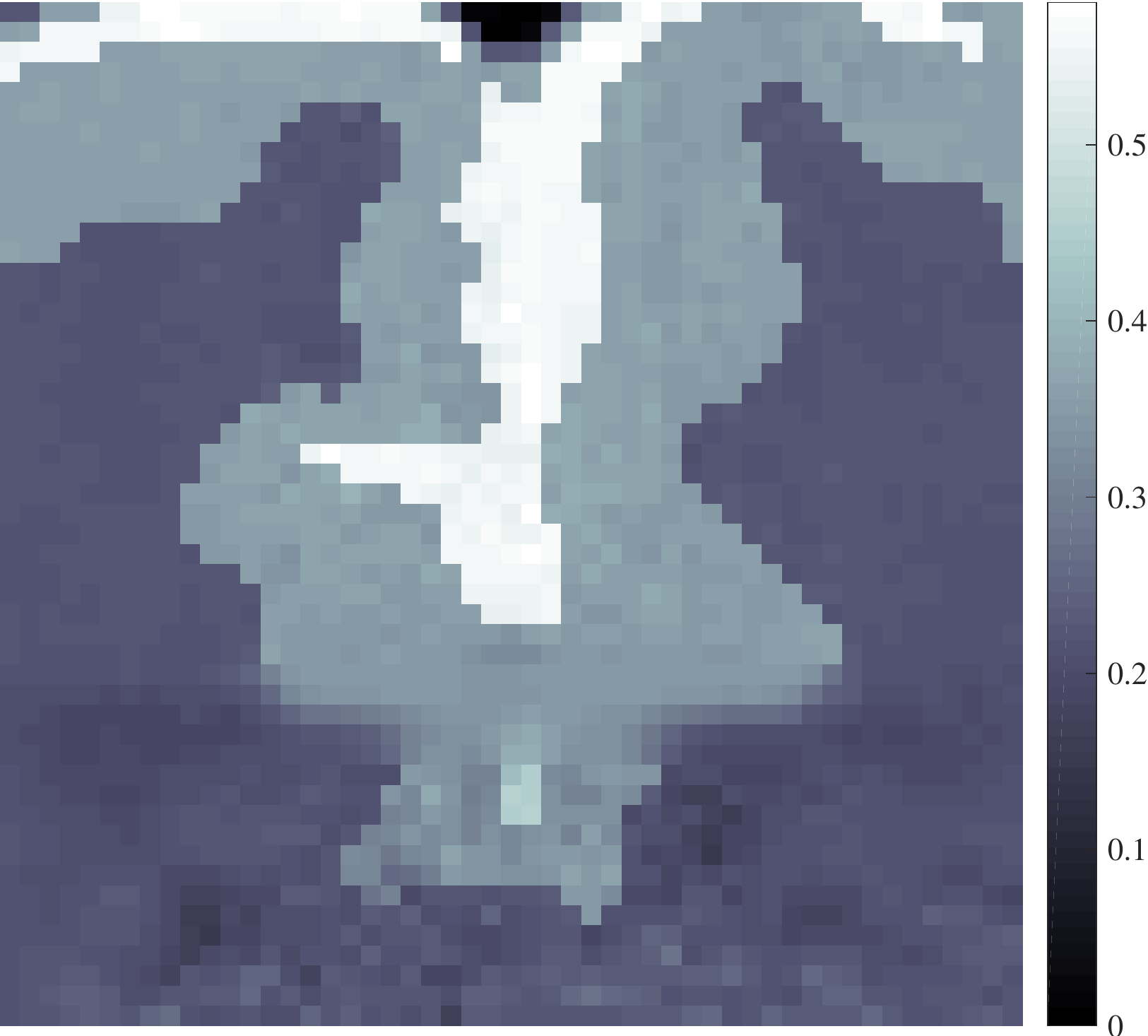}	
\end{tabular}
\end{center}


\caption{\label{Fig:MRI:cartesian} \small
Images showing the simulation results for the magnetic resonance imaging problem considering the Cartesian trajectories displayed in Figure~\ref{Fig:MRI:description}(b). Uncertainty quantification for two artefacts appearing in the MAP estimate, with corresponding $\rho_\alpha = 0.02\%$ (first two rows) and $\rho_\alpha = 0.01\%$ (last two rows).
In both cases $H_0$ cannot be rejected. 
Rows 1 and 3: images in linear scale with the structures of interest highlighted in red with, from left to right: 
$x^\dagger$, 
$x^\ddagger_{\widetilde{\Cc}_\alpha}$, 
and $x^\ddagger_{\Sc}$. 
Rows 2 and 4: zoomed images in linear scale on the area of the structures of interest, corresponding to the images displayed in rows 1 and 3, respectively. The scales in the zoomed images are adapted to better emphasize the two structures of interest. 
}
\end{figure}

%
%

\section{Discussions}
\label{sec:discuss}

\subsection{Model misspecification and approximation errors}
\label{Ssec:discuss:approxBayes}
We now discuss some philosophical aspects of the proposed methodology, and some implicit approximations that their users should be aware of. The first approximation is arguably the mathematical analysis of the imaging problem and its related uncertainty. That is, the fact that we formulate the problem mathematically to operate in a mathematical framework and deliver a mathematical solution for a real imaging problem always involves an approximation, despite the fact that our methodology has been rigorously mathematically derived. The statistical model $p(x|y)$ is of course an explicit approximation because $p(y|x)$ and $p(x)$ are inevitably misspecified. Similarly, $\Sc$ is also a modelling choice motivated by operational considerations (e.g., convexity). Mapping the results of a hypothesis test to statements and conclusions about real structures is also a form of implicit approximation. With this in mind, we understand our methodology as a tool for exploring uncertainty and supporting the use of images as evidence to inform decisions and conclusions. However, we do not attach particular attention to specific significance levels (e.g. $\alpha = 0.01$) because we do not believe that models are sufficiently well calibrated to allow accurate statements of posterior probabilities.

To conclude, we emphasize again that~\eqref{eq:posterior1} is an operational posterior distribution that models our knowledge about $x$ after observing $y$, a model derived from a likelihood function $p(y|x)$ and a prior $p(x)$ that are both operational approximations of some true conditional and marginal probability distributions that are unknown to us. Using an operational model is arguably unavoidable in imaging settings, given that the true marginal distribution of $x$ is difficult to fully characterize, and that the exact likelihood is certainly more complicated than the linear models and exponential-family noise distributions commonly used. As a consequence, our inferences are subjective in this sense and should not be understood as accurate probability statements regarding the underlying true image. Also, they should also not be understood as frequentist probability statements (i.e. related to the relative frequencies of different outcomes if the experiments were repeated a large number of times). Analyzing the frequentist statistical properties of Bayesian procedures in high-dimensional settings is very challenging. In particular, the frequentist properties of the proposed methodology (e.g. the power of the hypothesis test) will depend on the specific model and experiment considered. 
Also note that there are estimators of the form \eqref{hpd2} that are not MAP estimators derived from a Bayesian model \eqref{eq:posterior2} (this point is discussed for instance in \cite{Gribonval2011}). In such cases, we would not recommend using the proposed Bayesian uncertainty quantification methodology.

\subsection{Generalisations to other data observation models}
\label{Ssec:discuss:noise}

In this work, we assume that $w$ in the linear problem~\eqref{pb:inv_pb} has bounded energy. It is important to emphasize that the proposed BUQO method is not restricted to this assumption. 
Indeed, according to~\cite{pereyra2017maximum}, considering another type of noise is leading to a different conservative credible region $\widetilde{\Cc}_\alpha$ than the one given in~\eqref{eq:hpdapp1}. This change only affects the projection $\proj_{\widetilde{\Cc}_\alpha}$ in the proposed method, which needs to be adapted.
 
As a particular example, considering an additive i.i.d. Gaussian noise with zero mean and variance $\sigma$, the associated likelihood is of the form $p(y|x) \propto \exp(- \| \Phi x - y \|^2/(2\sigma^2)) $, and the MAP estimator is given by
\begin{equation}
x^\dagger \in \Argmin_{x \in \R^N} \left\{ g(x) := \frac{1}{2\sigma^2} \|\Phi x - y \|^2 + g_2(x) \right\},
\end{equation}
where $g_2$ is the regularization term. In this context, the conservative credible region $\widetilde{\Cc}_\alpha$ defined in~\cite{pereyra2017maximum} is expressed as follows:
\begin{equation}
\widetilde{\Cc}_\alpha = \left\{ x \in \R^N \mid   \frac{1}{2\sigma^2} \|\Phi x - y \|^2 + g_2(x) \le \widetilde{\eta}_\alpha \right\},
\end{equation}
with $\widetilde{\eta}_\alpha = g(x^\dagger) + N(\tau_\alpha+1)$. It can be noticed that the set $\widetilde{\Cc}_\alpha$ cannot be split into an intersection of simple sets when $g_2$ is not an indicator function. Consequently, to compute the projection onto this set, epigraphical projections must be leveraged \cite{chierchia2015epigraphical}. The remainder of the proposed BUQO approach remains unchanged.

\subsection{Comparison with state-of-the-art MCMC approaches}
\label{Ssec:discuss:comp_MCMC}

As explained in Section~\ref{Ssec:Bayes_quant}, MCMC algorithms can be used as well to perform uncertainty quantification in imaging. However, generally these approached have a computational cost which is several orders of magnitude higher than the computational cost associated with advanced optimization methods.
For example, in the context of our simulations, both for astronomical and medical imaging, computing the hypothesis test by using the state-of-the-art proximal MCMC algorithm \cite{durmus2016efficient} would require using approximately $10^4$ iterations of the algorithm for a small problem. One iteration of this algorithm has a similar computational cost as one iteration of the proposed convex optimisation scheme, which converges in only $10^2$ iterations and as a result is significantly faster. This computational advantage becomes more pronounced as the problem dimension increases, with large problems easily requiring over $10^6$ MCMC iterations with \cite{durmus2016efficient}, and only $10^3$ iterations with the proposed convex optimisation scheme.

\section{Conclusions}
\label{Sec:conclusion}

In this paper, we proposed a Bayesian uncertainty quantification methodology in the context of high dimensional imaging inverse problems. The proposed BUQO approach aims to analyse the degree of confidence in specific image structures (e.g., celestial sources in astronomical images, or  lesions in medical images) appearing in the MAP estimates, when the Bayesian models are log-concave. We proposed to quantify the uncertainty of the structures under scrutiny by performing a Bayesian hypothesis test, leveraging scalable optimization algorithms. Our approach allows to scale to  high-resolution and high-sensitivity imaging problems that are computationally intractable for state-of-the-art Bayesian computation approaches. 
The proposed methodology was demonstrated on challenging Fourier imaging problems related to radio astronomy and magnetic resonance in medicine where there is significant intrinsic uncertainty, and where we considered various types of structures and imaging setups. 
The corresponding \textsc{Matlab} code is available on GitHub (\url{https://basp-group.github.io/BUQO/}).

In future works, we plan to investigate the statistical calibration properties of our models, which will make more precise the limitations of the proposed methodology. We also plan to generalize the proposed approach to solve more sophisticated inverse problems. For instance, often when the inverse problem is non-linear, the MAP approach leads to a non-convex minimization problem \cite{Repetti_2017, Birdi2016, Repetti_SPL_2015, Bolte_2014}. In this case, the theoretical results of \cite{pereyra2017maximum} do not hold, and our approach cannot be directly applied.

\bibliographystyle{siamplain}

\begin{thebibliography}{10}

\bibitem{Alotaibi_A_2014_Solving_ccm}
{\sc A.~Alotaibi, P.~L. Combettes, and N.~Shahzad}, {\em Solving coupled
  composite monotone inclusions by successive {F}ej\'er approximations of their
  {K}uhn-{T}ucker set}, SIAM J. Optim., 24 (2014), pp.~2076--2095.

\bibitem{Altmann2015}
{\sc Y.~Altmann, M.~Pereyra, and J.~Bioucas-Dias}, {\em Collaborative sparse
  regression using spatially correlated supports - application to hyperspectral
  unmixing}, IEEE Trans. Image Process., 24 (2015), pp.~5800--5811.

\bibitem{Attouch_Bolte_2011}
{\sc H.~Attouch, J.~Bolte, and B.~F. Svaiter}, {\em Convergence of descent
  methods for semi-algebraic and tame problems: proximal algorithms,
  forward-backward splitting, and regularized {G}auss-{S}eidel methods}, Math.
  Program., 137 (2011), pp.~91--129.

\bibitem{Bauschke_Borwein_1994}
{\sc H.~H. Bauschke and J.~M. Borwein}, {\em Dykstra's alternating projection
  algorithm for two sets}, Journal Approx. Theory, 79 (1994), pp.~418--443.

\bibitem{Bauschke_Borwein_1996}
{\sc H.~H. Bauschke and J.~M. Borwein}, {\em On projection algorithms for
  solving convex feasibility problems}, SIAM review, 38 (1996), pp.~367--426.

\bibitem{bauschke2017convex}
{\sc H.~H. Bauschke and P.~L. Combettes}, {\em Convex analysis and monotone
  operator theory in Hilbert spaces}, Springer, 2017.

\bibitem{beck2009fast}
{\sc A.~Beck and M.~Teboulle}, {\em A fast iterative shrinkage-thresholding
  algorithm for linear inverse problems}, SIAM J. Imaging Sci., 2 (2009),
  pp.~183--202.

\bibitem{biegler2011large}
{\sc L.~Biegler, G.~Biros, O.~Ghattas, M.~Heinkenschloss, D.~Keyes, B.~Mallick,
  L.~Tenorio, B.~van Bloemen~Waanders, K.~Willcox, and Y.~Marzouk}, {\em
  Large-scale inverse problems and quantification of uncertainty}, vol.~712,
  John Wiley \& Sons, 2011.

\bibitem{BioucasDias2006}
{\sc J.~Bioucas-Dias}, {\em Bayesian wavelet-based image deconvolution: a {GEM}
  algorithm exploiting a class of heavy-tailed priors}, IEEE Trans. Image
  Process., 15 (2006), pp.~937--951.

\bibitem{Birdi2016}
{\sc J.~Birdi, A.~Repetti, and Y.~Wiaux}, {\em A regularized tri-linear
  approach for optical interferometric imaging}, Mon. Not. R. Astron. Soc., 468
  (2017), pp.~1142--1155.

\bibitem{Bot_R_2014_jmiv_conv_primal_apd}
{\sc R.~I. Bo\c{t} and C.~Hendrich}, {\em Convergence analysis for a
  primal-dual monotone + skew splitting algorithm with applications to total
  variation minimization}, J. Math. Imaging Vision, 49 (2014), pp.~551--568.

\bibitem{Bolte_2014}
{\sc J.~Bolte, S.~Sabach, and M.~Teboulle}, {\em Proximal alternating
  linearized minimization for nonconvex and nonsmooth problems}, {M}ath.
  {P}rogram., 146 (2014), pp.~459--494.

\bibitem{boyd2011distributed}
{\sc S.~Boyd, N.~Parikh, E.~Chu, B.~Peleato, and J.~Eckstein}, {\em Distributed
  optimization and statistical learning via the alternating direction method of
  multipliers}, Foundations and Trends{\textregistered} in Machine Learning, 3
  (2011), pp.~1--122.

\bibitem{boyd2004convex}
{\sc S.~Boyd and L.~Vandenberghe}, {\em Convex optimization}, Cambridge
  university press, 2004.

\bibitem{Bregman_1965}
{\sc L.~M. Bregman}, {\em The method of successive projection for finding a
  common point of convex sets}, Soviet Math. Dokl., 162 (1965), pp.~688--692.

\bibitem{Briceno_L_2011_j-siam-opt_mon_ssm}
{\sc L.~M. Brice\~{n}o{-}Arias and P.~L. Combettes}, {\em A monotone + skew
  splitting model for composite monotone inclusions in duality}, SIAM J.
  Optim., 21 (2011), pp.~1230--1250.

\bibitem{Cai_2017_BayesII}
{\sc X.~Cai, M.~Pereyra, and J.~D. McEwen}, {\em Uncertainty quantification for
  radio interferometric imaging: Ii. map estimation}, To appear in {Monthly
  Notices of the Royal Astronomical Society},  (2018),
  \url{https://doi.org/10.1093/mnras/sty2015}.

\bibitem{candes2006compressive}
{\sc E.~J. Cand{\`e}s et~al.}, {\em Compressive sampling}, in Proceedings of
  the international congress of mathematicians, vol.~3, Madrid, Spain, 2006,
  pp.~1433--1452.

\bibitem{chambolle2004algorithm}
{\sc A.~Chambolle}, {\em An algorithm for total variation minimization and
  applications}, J. Math. Imaging Vision, 20 (2004), pp.~89--97.

\bibitem{Chambolle2015}
{\sc A.~Chambolle and C.~Dossal}, {\em On the convergence of the iterates of
  ``fista''}, J. Optim. Theory Appl., 166 (2015), p.~25.

\bibitem{Chambolle_A_2010_first_opdacpai}
{\sc A.~Chambolle and T.~Pock}, {\em A first-order primal-dual algorithm for
  convex problems with applications to imaging}, J. Math. Imaging Vision, 40
  (2011), pp.~120--145.

\bibitem{Chambolle2016}
{\sc A.~Chambolle and T.~Pock}, {\em An introduction to continuous optimization
  for imaging}, Acta Numerica, 25 (2016), pp.~161--319.

\bibitem{Cheney_goldstein_1959}
{\sc W.~Cheney and A.~Goldstein}, {\em Proximity maps for convex sets}, Proc.
  Amer. Math. Soc., 10 (1959), pp.~448--450.

\bibitem{Chierchia_prox_rep}
{\sc G.~Chierchia, E.~Chouzenoux, P.~L. Combettes, and J.-C. Pesquet}, {\em The
  Proximity Operator Repository. User's guide}.
\newblock Available at \url{http://proximity-operator.net/}.

\bibitem{chierchia2015epigraphical}
{\sc G.~Chierchia, N.~Pustelnik, J.-C. Pesquet, and B.~Pesquet-Popescu}, {\em
  Epigraphical splitting for solving constrained convex formulations of inverse
  problems with proximal tools}, Signal, Image and Video Processing, 9 (2015),
  pp.~1737--1749.

\bibitem{Chouzenoux13}
{\sc E.~Chouzenoux, J.-C. Pesquet, and A.~Repetti}, {\em {V}ariable metric
  forward-backward algorithm for minimizing the sum of a differentiable
  function and a convex function}, {J}. {O}ptim. {T}heory {A}ppl., 162 (2014).

\bibitem{Chouzenoux_2016}
{\sc E.~Chouzenoux, J.-C. Pesquet, and A.~Repetti}, {\em A block coordinate
  variable metric forward-backward algorithm}, J. Global Optim., 66 (2016),
  pp.~457--485.

\bibitem{Combettes_Vu_2011}
{\sc P.~L. Combettes, D.~D\~ung, and B.~C. V\~u}, {\em Proximity for sums of
  composite functions}, J. Math. Anal. Appl., 380 (2011), pp.~680--688.

\bibitem{combettes2011proximal}
{\sc P.~L. Combettes and J.-C. Pesquet}, {\em Proximal splitting methods in
  signal processing}, in Fixed-point algorithms for inverse problems in science
  and engineering, Springer, 2011, pp.~185--212.

\bibitem{Combettes_P_2012_j-svva_pri_dsa}
{\sc P.~L. Combettes and J.-C. Pesquet}, {\em Primal-dual splitting algorithm
  for solving inclusions with mixtures of composite, {L}ipschitzian, and
  parallel-sum type monotone operators}, Set-Valued Var. Anal., 20 (2012),
  pp.~307--330.

\bibitem{combettes2005signal}
{\sc P.~L. Combettes and V.~R. Wajs}, {\em Signal recovery by proximal
  forward-backward splitting}, Multiscale Modeling \& Simulation, 4 (2005),
  pp.~1168--1200.

\bibitem{condat2013primal}
{\sc L.~Condat}, {\em A primal--dual splitting method for convex optimization
  involving {L}ipschitzian, proximable and linear composite terms}, J. Optim.
  Theory Appl., 158 (2013), pp.~460--479.

\bibitem{Deutsch_1992_MAP}
{\sc F.~Deutsch}, {\em The method of alternating orthogonal projetions},
  Approximation theory, Spline FUnctions and Applications,  (1992),
  pp.~105--121.

\bibitem{donoho2006compressed}
{\sc D.~L. Donoho}, {\em Compressed sensing}, IEEE Trans. inform. theory, 52
  (2006), pp.~1289--1306.

\bibitem{durmus2016efficient}
{\sc A.~Durmus, E.~Moulines, and M.~Pereyra}, {\em Efficient bayesian
  computation by proximal markov chain monte carlo: when langevin meets
  moreau}, arXiv:1612.07471,  (2016).

\bibitem{Escalande_book_2011}
{\sc R.~Escalande and M.~Raydan}, {\em Alternating Projection Methods}, SIAM,
  2011.

\bibitem{Esser_E_2010_j-siam-is_gen_fcf}
{\sc E.~Esser, X.~Zhang, and T.~Chan}, {\em A general framework for a class of
  first order primal-dual algorithms for convex optimization in imaging
  science}, SIAM J. Imaging Sci., 3 (2010), pp.~1015--1046.

\bibitem{Giovannelli_book_2015}
{\sc J.~F. Giovannelli and J.~Idier}, {\em Regularization and Bayesian Methods
  for Inverse Problems in Signal and Image Processing}, Wiley-ISTE, 2015.

\bibitem{Gribonval2011}
{\sc R.~Gribonval}, {\em Should penalized least squares regression be
  interpreted as maximum a posteriori estimation}, IEEE Trans. Signal Proc., 59
  (2011).

\bibitem{Haldar2011}
{\sc J.~P. Haldar, D.~Hernando, and Z.-P. Liang}, {\em Compressed sensing {MRI}
  with random encoding}, IEEE Trans. Med. Imag., 30 (2011).

\bibitem{Halperin_1962}
{\sc I.~Halperin}, {\em The product of projection operators}, Acta Sci. Math.,
  23 (1962), pp.~96--99.

\bibitem{komodakis2015playing}
{\sc N.~Komodakis and J.-C. Pesquet}, {\em Playing with duality: An overview of
  recent primal? dual approaches for solving large-scale optimization
  problems}, IEEE Signal Process. Mag., 32 (2015), pp.~31--54.

\bibitem{Lebrun2013}
{\sc M.~Lebrun, A.~Buades, and J.~M. Morel}, {\em A nonlocal {B}ayesian image
  denoising algorithm}, SIAM J. Imaging Sci., 6 (2013), p.~16651688.

\bibitem{Mallat_book}
{\sc S.~Mallat}, {\em A Wavelet Tour of Signal Processing}, Academic Press,
  Burlington, MA, 2rd~ed., 2009.

\bibitem{Niknejad2018}
{\sc M.~Niknejad, J.~Bioucas-Dias, and M.~Figueiredo}, {\em Image restoration
  using conditional random fields and scale mixtures of gaussians}, tech.
  report, 2018.
\newblock arXiv:1807.03027.

\bibitem{Pock2013}
{\sc P.~Ochs, Y.~Chen, T.~Brox, and T.~Pock}, {\em i{P}iano: inertial proximal
  algorithm for non-convex optimization}, SIAM J. Imaging Sci., 7 (2014),
  pp.~1388--1419.

\bibitem{onose2016scalable}
{\sc A.~Onose, R.~E. Carrillo, A.~Repetti, J.~D. McEwen, J.-P. Thiran, J.-C.
  Pesquet, and Y.~Wiaux}, {\em Scalable splitting algorithms for big-data
  interferometric imaging in the {SKA} era}, Monthly Notices of the Royal
  Astronomical Society, 462 (2016), pp.~4314--4335.

\bibitem{pereyra2016proximal}
{\sc M.~Pereyra}, {\em Proximal markov chain monte carlo algorithms},
  Statistics and Computing, 26 (2016), pp.~745--760.

\bibitem{pereyra2017maximum}
{\sc M.~Pereyra}, {\em Maximum-a-posteriori estimation with bayesian confidence
  regions}, SIAM J. Imaging Sci., 10 (2017), pp.~285--302.

\bibitem{Pereyra2013}
{\sc M.~Pereyra, N.~Dobigeon, H.~Batatia, and J.-Y. Tourneret}, {\em Estimating
  the granularity coefficient of a {Potts-Markov} random field within an {MCMC}
  algorithm}, IEEE Trans. Image Process., 22 (2013), pp.~2385--2397.

\bibitem{pereyra2016survey}
{\sc M.~Pereyra, P.~Schniter, E.~Chouzenoux, J.-C. Pesquet, J.-Y. Tourneret,
  A.~O. Hero, and S.~McLaughlin}, {\em A survey of stochastic simulation and
  optimization methods in signal processing}, IEEE J. Selected Topics Signal
  Process, 10 (2016), pp.~224--241.

\bibitem{pesquet2014class}
{\sc J.-C. Pesquet and A.~Repetti}, {\em A class of randomized primal-dual
  algorithms for distributed optimization}, J. Nonlinear Convex Anal., 16
  (2015), pp.~2352--2490.

\bibitem{Pustelnik2016}
{\sc N.~Pustelnik, A.~Benazza-Benhayia, Y.~Zheng, and J.-C. Pesquet}, {\em
  Wavelet-Based Image Deconvolution and Reconstruction}, American Cancer
  Society, 2016, pp.~1--34, \url{https://doi.org/10.1002/047134608X.W8294}.

\bibitem{Repetti_2017}
{\sc A.~Repetti, J.~Birdi, A.~Dabbech, and Y.~Wiaux}, {\em Non-convex
  optimization for self-calibration of directiondependent effects in radio
  interferometric imaging}, {Mon. Not. R. Astron. Soc.}, 470 (2017),
  pp.~3981--4006.

\bibitem{Repetti_SPL_2015}
{\sc A.~Repetti, M.~Q. Pham, L.~Duval, E.~Chouzenoux, and J.-C. Pesquet}, {\em
  {Euclid in a Taxicab}: {S}parse blind deconvolution with smoothed
  $\ell_1/\ell_2$ regularization}, IEEE Signal Process. Lett., 22 (2015),
  pp.~539--543.

\bibitem{robert2007bayesian}
{\sc C.~Robert}, {\em The Bayesian choice: from decision-theoretic foundations
  to computational implementation}, Springer Science \& Business Media, 2007.

\bibitem{robert2004monte}
{\sc C.~Robert and G.~Casella}, {\em Monte carlo statistical methods springer},
  New York,  (2004).

\bibitem{rudin1992nonlinear}
{\sc L.~I. Rudin, S.~Osher, and E.~Fatemi}, {\em Nonlinear total variation
  based noise removal algorithms}, Physica D: Nonlinear Phenomena, 60 (1992),
  pp.~259--268.

\bibitem{Tseng_P_2000_j-siam-control-optim_Modified_fbs}
{\sc P.~Tseng}, {\em A modified forward-backward splitting method for maximal
  monotone mappings}, SIAM J. Control Optim., 38 (2000), pp.~431--446.

\bibitem{vonNeumann_1951}
{\sc J.~von Neumann}, {\em Functional operators, Vol. II. The Geometry of
  Orthogonal Spaces}, Princeton University Press, princeton, NJ, 1951.

\bibitem{vu2013splitting}
{\sc B.~C. V{\~u}}, {\em A splitting algorithm for dual monotone inclusions
  involving cocoercive operators}, Advances in Computational Mathematics, 38
  (2013), pp.~667--681.

\bibitem{Wei2016}
{\sc Q.~Wei, J.~M. Bioucas-Dias, N.~Dobigeon, J.-Y. Tourneret, M.~Chen, and
  S.~Godsill}, {\em Multi-band image fusion based on spectral unmixing}, IEEE
  Trans. Geosci. Remote Sens., 54 (2016), pp.~7236--7249.

\bibitem{wiaux2009compressed}
{\sc Y.~Wiaux, L.~Jacques, G.~Puy, A.~M. Scaife, and P.~Vandergheynst}, {\em
  Compressed sensing for radio interferometry: prior-enhanced basis pursuit
  imaging techniques}, in SPARS" 09-Signal Processing with Adaptive Sparse
  Structured Representations, no.~EPFL-CONF-139261, 2009.

\bibitem{Wipf2007}
{\sc D.~P. Wipf and B.~D. Rao}, {\em An empirical bayesian strategy for solving
  the simultaneous sparse approximation problem}, IEEE Trans. Signal Process.,
  55 (2007), pp.~3704--3716.

\end{thebibliography}

\end{document}